\documentclass[reprint,amsfonts, amssymb, amsmath,  showkeys, nofootinbib, prx, superscriptaddress, twocolumn,longbibliography]{revtex4-2}
\usepackage{float}
\makeatletter
\let\newfloat\newfloat@ltx
\makeatother
\usepackage[english]{babel}
\usepackage[utf8]{inputenc}
\usepackage{graphics}
\usepackage{selinput}
\usepackage[normalem]{ulem}
\usepackage[shortlabels]{enumitem}

\usepackage{refcount}
\usepackage{braket}
\usepackage{amsthm}
\usepackage{mathtools}
\usepackage{physics}
\usepackage{xcolor}
\usepackage{graphicx}
\usepackage[left=16mm,right=16mm,top=35mm,columnsep=15pt]{geometry} 
\usepackage{adjustbox}
\usepackage{placeins}
\usepackage[T1]{fontenc}
\usepackage{lipsum}
\usepackage{csquotes}
\usepackage{bm}

\usepackage[linesnumbered,ruled,vlined]{algorithm2e}
\SetKwInput{kwInit}{Init}

\def\HC{\mathcal{H}}

\def\LC{\mathcal{L}}

\usepackage{mathrsfs} 

\def\ad{^{\dagger}}

\newcommand{\fsnull}[1]{}
\newcommand{\old}[1]{}

\usepackage[makeroom]{cancel}
\usepackage[toc,page]{appendix}
\usepackage[colorlinks=true,citecolor=blue,linkcolor=magenta]{hyperref}

\usepackage{tikz}
\tikzset{every picture/.style=remember picture}

\newcommand{\dya}[1]{\ket{#1}\!\bra{#1}}

\newcommand{\poly}{\operatorname{poly}}

\newcommand{\BC}{\mathcal{B}}
\newcommand{\CC}{\mathcal{C}}
\newcommand{\DC}{\mathcal{D}}
\newcommand{\EC}{\mathcal{E}}

\newcommand{\NC}{\mathcal{N}}
\newcommand{\OC}{\mathcal{O}}

\newcommand{\SC}{\mathcal{S}}
\newcommand{\TC}{\mathcal{T}}

\newcommand{\YC}{\mathcal{Y}}

\newcommand{\Var}{{\rm Var}}

\renewcommand{\geq}{\geqslant}
\renewcommand{\leq}{\leqslant}

\renewcommand{\vec}[1]{\boldsymbol{#1}}  

\newcommand*{\id}{\openone}

\newcommand{\bs}{\textsf{BS}}

\newcommand{\sg}{\sigma }

\newcommand{\SWAP}{{\rm SWAP}}

\newcommand{\thv}{\vec{\theta}}

\def\be{\begin{equation}}
\def\ee{\end{equation}}
\def\bs{\begin{split}}
\def\e{\end{split}}
\def\ba{\begin{eqnarray}}
\def\bea{\begin{eqnarray}}

\def\tea{\end{eqnarray}}
\def\ea{\end{eqnarray}}
\def\eea{\end{eqnarray}}

\newtheorem{theorem}{Theorem}
\newtheorem{lemma}{Lemma}
\newtheorem{corollary}{Corollary}

\newtheorem{definition}{Definition}

\usepackage{amssymb}
\usepackage{dsfont}

\def\be{\begin{equation}}
\def\te{\end{equation}}
\def\ee{\end{equation}}
\def\ba{\begin{eqnarray}}
\def\bea{\begin{eqnarray}}

\def\tea{\end{eqnarray}}
\def\ea{\end{eqnarray}}
\def\eea{\end{eqnarray}}

\def\la{{\langle}}
\def\ra{{\rangle}}
\newcommand{\beq}{\begin{equation}}
\newcommand{\eeq}{\end{equation}}
\newcommand{\brauer}{\mathfrak{B}_k(d)}

\def\GL{\mathbb{GL}}

\begin{document}

\title{Quantum neural networks form Gaussian processes}

\author{Diego Garc\'ia-Mart\'in}
\affiliation{Information Sciences, Los Alamos National Laboratory, Los Alamos, New Mexico 87545, USA}
\author{Mart\'{i}n Larocca}
\affiliation{Theoretical Division, Los Alamos National Laboratory, Los Alamos, New Mexico 87545, USA}
\affiliation{Center for Nonlinear Studies, Los Alamos National Laboratory, Los Alamos, New Mexico 87545, USA}
\author{M. Cerezo}
\thanks{cerezo@lanl.gov}
\affiliation{Information Sciences, Los Alamos National Laboratory, Los Alamos, New Mexico 87545, USA}

\begin{abstract}
It is well known that artificial neural networks initialized from independent and identically distributed priors converge to Gaussian processes in the limit of a large number of neurons per hidden layer. In this work we prove an analogous result for  Quantum Neural Networks (QNNs). Namely, we show that the outputs of certain models based on Haar random unitary or orthogonal deep QNNs  converge to Gaussian processes in the limit of large Hilbert space dimension $d$. The derivation of this result is more nuanced than in the classical case due to the role played by the input states, the measurement observable, and the fact that the entries of unitary matrices are not independent. Then, we show that the efficiency of predicting measurements at the output of a QNN using Gaussian process regression depends on the observable's bodyness. Furthermore, our theorems imply that the concentration of measure phenomenon in Haar random QNNs is worse than previously thought, as we prove that expectation values and gradients  concentrate as 
$\mathcal{O}\left(\frac{1}{e^d \sqrt{d}}\right)$. Finally, we discuss how our results improve our understanding of  concentration in $t$-designs.

\end{abstract}

\maketitle

Neural Networks (NNs) have revolutionized the fields of Machine Learning (ML) and artificial intelligence. Their tremendous  success across many fields of research  in a wide variety of applications~\cite{alzubaidi2021review,khurana2023natural,jumper2021highly} is certainly astonishing. While much of this success has come through heuristics, the past few decades have witnessed a significant increase in our theoretical understanding of their inner workings.
One of the most interesting results regarding NNs is that fully-connected models with a single hidden layer converge to Gaussian Processes (GPs) in the limit of large number of hidden  neurons, when the parameters are initialized from independent and identically distributed (i.i.d.) priors~\cite{neal1996priors}. More recently, it has been shown that i.i.d.-initialized, fully-connected, multi-layer NNs also converge to GPs in the infinite-width limit~\cite{lee2017deep}. Furthermore, 
other architectures, such as convolutional NNs~\cite{novak2019bayesian}, transformers~\cite{hron2020infinite} or recurrent NNs~\cite{yang2019wide}, are also GPs under certain assumptions. More than just a mathematical curiosity, the correspondence between NNs and GPs opened up the possibility of performing exact Bayesian inference for regression and learning tasks using wide NNs~\cite{neal1996priors,rasmussen2006gaussian}. Training wide NNs with GPs requires inverting the covariance matrix of the training set, a process that can be computationally expensive. Recent studies have explored the use of quantum linear algebraic techniques to efficiently perform these matrix inversions, potentially offering polynomial speedups over standard classical methods~\cite{zhao2019bayesian,kus2021sparse}.

Indeed, with the advent of quantum computers there has been an enormous interest in merging quantum computing with ML, leading to the thriving field of Quantum Machine Learning (QML)~\cite{biamonte2017quantum,cerezo2022challenges,cerezo2020variationalreview,liu2021representation,schuld2021machine}.  Rapid progress has been made in this field, largely fueled by the hope that QML may provide a quantum advantage in the near-term for some practically-relevant problems. While the prospects for such a practical quantum advantage remain unclear~\cite{schuld2022is}, a number of promising analytical results have already been put forward~\cite{larocca2021theory,anschuetz2022interpretable,abbas2020power,huang2020predicting}. Still, much remains to be learned about QML models.

In this work, we contribute to the QML body of knowledge by proving that under certain conditions, the outputs of deep Quantum Neural Networks (QNNs) -- i.e., parametrized quantum circuits acting on input states drawn from a training set-- converge to GPs in the limit of large Hilbert space dimension (see Fig.~\ref{fig:Summary}). 
 Our results are derived for QNNs that are Haar random over the unitary and orthogonal groups. 
Unlike the classical case, where the proof of the emergence of GPs stems from the central limit theorem, the situation becomes more intricate in the quantum setting as the entries of the QNN are not independent -- the rows and columns of a unitary matrix are constrained to be mutually orthonormal.  Hence, our proof strategy boils down to showing that each moment of the QNN's output distribution converges to that of a multivariate Gaussian. Importantly, we also show that the Bayesian distribution of a QNN acting on qubits is efficient (inefficient) for predicting  local (global) measurements. We then use our results to provide a precise characterization of the concentration of measure phenomenon in deep random quantum circuits~\cite{mcclean2018barren,cerezo2020cost,marrero2020entanglement,patti2020entanglement,holmes2020barren,arrasmith2021equivalence}. Here, our theorems indicate that the expectation values, as well as the gradients, of Haar random processes concentrate faster than previously reported~\cite{popescu2006entanglement}. Finally, we discuss how our results can be leveraged to study QNNs that are not fully Haar random but instead form $t$-designs, which constitutes a much more practical assumption~\cite{harrow2009random,harrow2018approximate,haferkamp2022random}.

\section{Gaussian processes and classical machine learning}

We begin by introducing GPs.
\begin{definition}[Gaussian process]\label{def:GP}
A collection  of random variables $\{X_{1},X_{2},\dots\}$ is a GP if and only if, for every finite set of indices $\{1,2,\dots,m\}$, the vector $(X_{1},X_{2},\dots,X_{m})$ follows a multivariate Gaussian  distribution, which we denote as $\NC(\vec{\mu},\vec{\Sigma})$. Said otherwise, every linear combination of  $\{X_{1},X_{2},\dots,X_{m}\}$ follows a univariate Gaussian distribution.
\end{definition}
In particular, $\NC(\vec{\mu},\vec{\Sigma})$ is determined by its $m$-dimensional mean vector $\vec{\mu}=(\mathbb{E}[X_{1}],\ldots,\mathbb{E}[X_{m}])$, and its $m\times m$ dimensional covariance matrix with entries $(\vec{\Sigma})_{\alpha\beta}={\rm Cov}[X_{\alpha},X_{\beta}]$. 

GPs are extremely important in ML since they can be used as a form of kernel method to solve learning tasks~\cite{neal1996priors,rasmussen2006gaussian}. For instance,  consider a regression problem where the data domain is $\mathscr{X}=\mathbb{R}$ and the label domain is $\mathscr{Y}=\mathbb{R}$. Instead of finding a single function $f:\mathscr{X}\rightarrow\mathscr{Y}$ which solves the regression task, a GP instead assigns probabilities to a set of possible $f(x)$, such that the probabilities are higher for the ``more likely'' functions. Following a Bayesian inference approach,  
one then selects the functions that best agree with some set of empirical observations~\cite{rasmussen2006gaussian,schuld2021machine}.

Under this framework, the output over the distribution of functions $f(x)$, for  $x\in\mathscr{X}$, is a random variable. Then, given a  set of training samples $x_1,\ldots,x_m$, and some covariance function $\kappa(x,x')$, Definition~\ref{def:GP} implies that if one has a GP, the outputs $f(x_1),\ldots,f(x_m)$ are random variables sampled from some multivariate Gaussian distribution $\NC(\vec{\mu},\vec{\Sigma})$. From here, the GP is used to make predictions about the output $f(x_{m+1})$ (for some new data instance  $x_{m+1}$), given the previous observations $f(x_1),\ldots,f(x_m)$. Explicitly, one constructs the joint distribution $P(f(x_1),\ldots,f(x_m),f(x_{m+1}))$ from the averages and the covariance function $\kappa$, and obtains the sought-after ``predictive distribution'' $P(f(x_{m+1})|f(x_1),\ldots,f(x_m))$ via marginalization. The power of the GP relies on the fact that this distribution usually contains less uncertainty than  $P(f(x_{m+1}))=\NC(\mathbb{E}[f(x_{m+1})],\kappa(x_{m+1},x_{m+1}))$ (see the Methods).

\begin{figure}[t]
    \centering
\includegraphics[width=1\columnwidth]{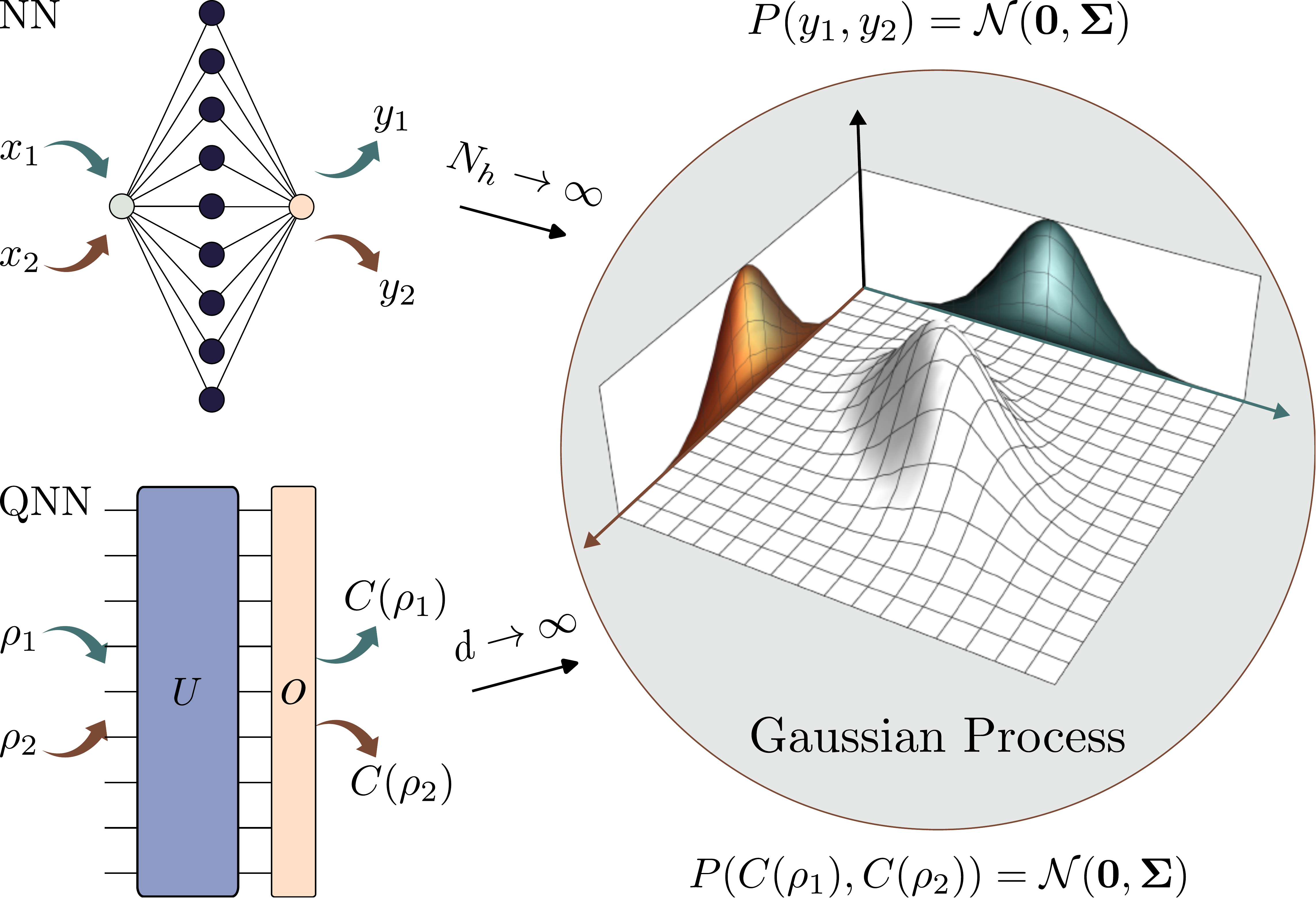}
    \caption{\textbf{Schematic of our main results.} It is well known that certain classical NNs with $N_h$ neurons per hidden layer become GPs when $N_h\rightarrow \infty$. That is, given  inputs $x_1$ and $x_2$, and corresponding outputs $y_1$ and $y_2$, then the joint probability $P(y_1,y_2)$ is a multivariate Gaussian $\NC(\vec{0},\vec{\Sigma})$. In this work, we show that a similar result holds under certain conditions for deep QNNs in the limit of large Hilbert space dimension, $d\rightarrow \infty$. Now, given quantum states $\rho_1$ and $\rho_2$, $C(\rho)=\Tr[U\rho U\ad O]$ is such that  $P(C(\rho_1),C(\rho_2))=\NC(\vec{0},\vec{\Sigma})$.
    \label{fig:Summary}}
\end{figure}

\section{Haar random deep QNNs form GPs}

In what follows we consider a setting where one is given repeated access to a dataset $\mathscr{D}$ containing quantum states $\{\rho_{i}\}_{i}$ on a $d$-dimensional Hilbert space that satisfy $\Tr[\rho_i^2]\in \Omega(\frac{1}{\poly(\log(d))})$ for all $i$. We will make no assumptions regarding the origin of these states, as they can correspond to classical data encoded in quantum states~\cite{lloyd2020quantum,perez2020data}, or quantum data obtained from some quantum mechanical process~\cite{schatzki2021entangled,larocca2022group}. Then, the states are sent through a deep QNN, denoted  $U$. While in general $U$ can be parametrized by some set of trainable parameters $\thv$, we leave such dependence implicit for ease of notation.   At the output of the circuit one measures the expectation value of a traceless Hermitian operator taken from a set  $\mathscr{O}=\{O_j\}_j$ such that $\Tr[O_j O_{j'}]=d\delta_{j,j'}$ and $O_j^2=\id$, for all $j,j'$ (e.g., Pauli strings). We denote the QNN outputs as
\begin{equation}\label{eq:cij}
    C_j(\rho_i)=\Tr[U\rho_i U\ad O_j]\,.
\end{equation}
Then, we collect these quantities over some set of states from $\mathscr{D}$ and some set of measurements from $\mathscr{O}$ in a vector
\beq \label{eq:random_vector} 
\mathscr{C}=\left(C_j(\rho_i),\ldots, C_{j'}(\rho_{i'}),\ldots \right) \,.
\eeq
As we will show below, in the large-$d$ limit $\mathscr{C}$ converges to a GP when the QNN unitaries $U$ are sampled
according to the Haar measure on $\mathbb{U}(d)$ and $\mathbb{O}(d)$, the degree-$d$ unitary and orthogonal groups, respectively (see Fig.~\ref{fig:Summary}). We recall that $\mathbb{U}(d)=\{U\in\mathbb{C}^{d\times d}\,,\, UU\ad=U\ad U=\id\}$ and that $\mathbb{O}(d)=\{U\in\mathbb{R}^{d\times d}\,,\, UU^T=U^TU=\id\}$. We will henceforth  use the notation $\mathbb{E}_{\mathbb{U}(d)}$ and $\mathbb{E}_{\mathbb{O}(d)}$ to respectively  denote Haar averages over $\mathbb{U}(d)$ and $\mathbb{O}(d)$. Moreover, we assume that when the circuit is sampled from $\mathbb{O}(d)$,  the states in $\mathscr{D}$ and the measurement operators in $\mathscr{O}$ are real valued.

\subsection{Moment computation in the large-$d$ limit}

As we discuss in the Methods, we cannot rely on simple central-limit-theorem arguments  to show that  $\mathscr{C}$ forms a GP. Hence, our proof strategy is based on computing all the moments of the vector $\mathscr{C}$  and showing that they  asymptotically match those of a multivariate Gaussian distribution. To conclude the proof we show that these moments unequivocally  determine the distribution, for which we can use Carleman's condition~\cite{petz2004asymptotics,kleiber2013multivariate}. We refer the reader to the Supplemental Information (SI) for the detailed proofs of the results in this manuscript.

First, we present the following lemma.
\begin{lemma}\label{lem:exp-cov}
    Let $C_j(\rho_i)$ be the expectation value of a Haar random QNN as in Eq.~\eqref{eq:cij}. Then for any $\rho_i\in \mathscr{D}$, $O_j\in \mathscr{O}$,
    \begin{equation}
        \mathbb{E}_{\mathbb{U}(d)}[C_j(\rho_i)]=\mathbb{E}_{\mathbb{O}(d)}[C_j(\rho_i)]=0\,.
    \end{equation}
    Moreover, for any pair of states $\rho_i,\rho_{i'}\in \mathscr{D}$ and operators $O_{j},O_{j'}\in \mathscr{O}$ we have
    \begin{equation}
       {\rm Cov}_{\mathbb{U}(d)}[C_j(\rho_i)C_{j'}(\rho_{i'})]={\rm Cov}_{\mathbb{O}(d)}[C_j(\rho_i)C_{j'}(\rho_{i'})]=0\nonumber\,,
    \end{equation}
    if $j\neq j'$ and 
    \small
    \begin{align}
       \vec{\Sigma}_{i, i'}^{\mathbb{U}}&=\frac{d}{d^2-1}\left(\Tr[\rho_{i}\rho_{i'}]-\frac{1}{d}\right)\,,\\
       \vec{\Sigma}_{i, i'}^{\mathbb{O}}&=\frac{2(d+1)}{(d+2)(d-1)}\left(\Tr[\rho_{i}\rho_{i'}]\left(1-\frac{1}{d+1}\right)-\frac{1}{d+1}\right)\,,
    \end{align}
    \normalsize
    if $j=j'$. Here, we have defined $\vec{\Sigma}_{i, i'}^{G}={\rm Cov}_{G}[C_j(\rho_i)C_{j}(\rho_{i'})]$, where $G=\mathbb{U}(d),\mathbb{O}(d)$.
\end{lemma}

Lemma~\ref{lem:exp-cov} shows that the expectation value of the QNN outputs is always zero. More notably, it  indicates that the covariance between the outputs is null if we measure different observables  (even if we use the same input state and the same circuit).  This implies that the distributions $C_j(\rho_i)$ and $C_{j'}(\rho_{i'})$ are uncorrelated if $j\neq j'$. That is, knowledge of the measurement outcomes for one observable and different input states does not provide any information about the outcomes  of other measurements, at these or any other input states.Therefore, in what follows, we will focus on the case where $\mathscr{C}$ contains expectation values for different states but the same operator. In this case, Lemma~\ref{lem:exp-cov} shows that the covariances will be positive, zero, or negative depending on whether $\Tr[\rho_{i}\rho_{i'}]$ is larger, equal, or smaller than $\frac{1}{d}$, respectively. 

\begin{table}[t]
    \centering
\begin{tabular}{|l|*{3}{c|}}\hline
Dataset. For all $\rho_i\neq\rho_{i'}\in\mathscr{D}$:  & GP & Correlation & Statement\\\hline\hline
$\Tr[\rho_i\rho_{i'}]\in\Omega\left(\frac{1}{\poly(\log(d))}\right)$   & Yes & Positive & Theorem~\ref{th:gp_main}\\\hline
$\Tr[\rho_i\rho_{i'}]=\frac{1}{d}$   & Yes & Null& Theorem~\ref{th:gaus-indept}\\\hline
$\Tr[\rho_i\rho_{i'}]=0$   & Yes & Negative& Theorem~\ref{th:GP-negative}\\\hline
\end{tabular}
    \caption{\textbf{Summary of main results.} In the first column we present conditions for the states in the dataset under which the deep QNN's outputs form GPs. In the remaining columns we  report the correlation in the GP variables and the associated theorem where the main result is stated. In all cases we assume that we measure the same operator $O_j$ for all $\rho_i,\rho_{i'}\in \mathscr{D}$. In Theorem~\ref{cod:dataset-average}, we extend  some of these results to the cases where the conditions are only met on average when sampling states over $\mathscr{D}$.}
    \label{tab:train_states}
\end{table}

We now state a useful result. 

\begin{lemma}\label{lem:moments}
    Let $\mathscr{C}$ be a vector of expectation values of a Haar random QNN as in Eq.~\eqref{eq:random_vector}, where one measures the same operator $O_j$ over a set  of states from $\mathscr{D}$. Furthermore, let $\rho_{i_1},\ldots,\rho_{i_k}\in\mathscr{D}$ be a multiset of states taken from those appearing in $\mathscr{C}$. In the large-$d$ limit, if $k$ is odd, then $\mathbb{E}_{\mathbb{U}(d)}\left[C_{j}(\rho_{i_1})\cdots C_{j}(\rho_{i_k})\right]=\mathbb{E}_{\mathbb{O}(d)}\left[C_{j}(\rho_{i_1})\cdots C_{j}(\rho_{i_k})\right]=0$. Moreover, if $k$ is even and  $\Tr[\rho_i\rho_{i'}]\in\Omega\left(\frac{1}{\poly(\log(d))}\right)$ for all $i, i'$,  we have
    \begin{align}
\mathbb{E}_{\mathbb{U}(d)}\left[C_{j}(\rho_{i_1})\cdots C_{j}(\rho_{i_k})\right]&=\frac{1}{d^{k/2}}\sum_{\sigma\in T_{k}}\prod_{\{t,t'\}\in \sigma} \Tr[\rho_{t}\rho_{t'}]\\
&=\frac{\mathbb{E}_{\mathbb{O}(d)}\left[C_{j}(\rho_{i_1})\cdots C_{j}(\rho_{i_k})\right]}{2^{k/2}}\,,\nonumber
\end{align}
where the summation runs over all the possible disjoint pairing of indexes in the set $\{1,2,\ldots,k\}$, $T_k$, and the product is over the different pairs in each pairing.
\end{lemma}
Using Lemma~\ref{lem:moments} as our main tool, we will be able to prove that deep QNNs form GPs for different types of datasets. In Table~\ref{tab:train_states}, we present a summary of our main results.

\subsection{Positively correlated GPs}

We begin by studying the case when the states in the dataset satisfy $\Tr[\rho_i\rho_{i'}]\in\Omega\left(\frac{1}{\poly(\log(d))}\right)$ for all $\rho_i,\rho_{i'}\in\mathscr{D}$. According to Lemma~\ref{lem:exp-cov}, this implies that the variables are positively correlated. 
In the large-$d$ limit, we can derive the following theorem.
\begin{theorem} \label{th:gp_main}
Under the same conditions for which Lemma~\ref{lem:moments} holds, the vector $\mathscr{C}$ forms a GP with mean vector $\vec{\mu}=\vec{0}$ and covariance matrix given by $
    \vec{\Sigma}_{i, i'}^{\mathbb{U}}=\frac{\vec{\Sigma}_{i, i'}^{\mathbb{O}}}{2}= \frac{\Tr[\rho_{i}\rho_{i'}]}{d}$.
\end{theorem}
Theorem~\ref{th:gp_main} indicates that  the covariances for the orthogonal group are twice as large as those arising from the unitary group. 
In  Fig.~\ref{fig:multivariate_gauss}, we present results obtained by numerically simulating a unitary Haar random QNN for  a system of $n=18$ qubits. The circuits were sampled using known results for the distribution of the entries of random unitary matrices~\cite{petz2004asymptotics}. In the left panels of  Fig.~\ref{fig:multivariate_gauss}, we show the corresponding two-dimensional GP obtained for two initial states that satisfy $\Tr[\rho_i\rho_{i'}]\in\Omega(1)$. Here, we see that the variables are positively correlated in accordance with the prediction in Theorem~\ref{th:gp_main}.

The fact that the outputs of deep QNNs form GPs reveals a deep connection between QNNs and quantum kernel methods. While it has already been pointed out that QNN-based QML constitutes a form of kernel-based learning~\cite{schuld2021quantum}, our results solidify this connection for the case of Haar random circuits. Notably, we can recognize that the kernel arising in the GP covariance matrix is proportional to the Fidelity kernel, i.e., to the Hilbert-Schmidt inner product between the data states~\cite{havlivcek2019supervised,schuld2021quantum,thanasilp2022exponential}. Moreover, since the predictive distribution of a GP can be expressed as a function of the covariance matrix (see the Methods), and thus of the kernel entries, our results further  cement the fact that quantum models such as those in Eq.~\eqref{eq:cij} are functions in the reproducing kernel Hilbert space~\cite{schuld2021quantum}.

\subsection{Uncorrelated GPs}

We now consider the case when $\Tr[\rho_i\rho_{i'}]=\frac{1}{d}$ for all $\rho_i\neq \rho_{i'}\in\mathscr{D}$. We find the following result.
\begin{theorem}\label{th:gaus-indept}
    Let $\mathscr{C}$ be a vector of expectation values of an operator in $\mathscr{O}$ over a set of states from $\mathscr{D}$, as in Eq.~\eqref{eq:random_vector}. If $\Tr[\rho_i\rho_{i'}]=\frac{1}{d}$ for all $i\neq i'$, then in the large-$d$ limit $\mathscr{C}$ forms a GP with mean vector $\vec{\mu}=\vec{0}$ and diagonal covariance matrix
\begin{equation}
\vec{\Sigma}_{i, i'}^{\mathbb{U}}= \frac{\vec{\Sigma}_{i, i'}^{\mathbb{O}}}{2} = \begin{cases}
    \frac{\Tr[\rho_i^2]}{d}\quad\text{if $i=i'$} \\
    0\qquad\;\;\, \text{if $i\neq i'$}\
\end{cases}\,.
\end{equation}
\end{theorem}

\begin{figure}
\includegraphics[width=1\columnwidth]{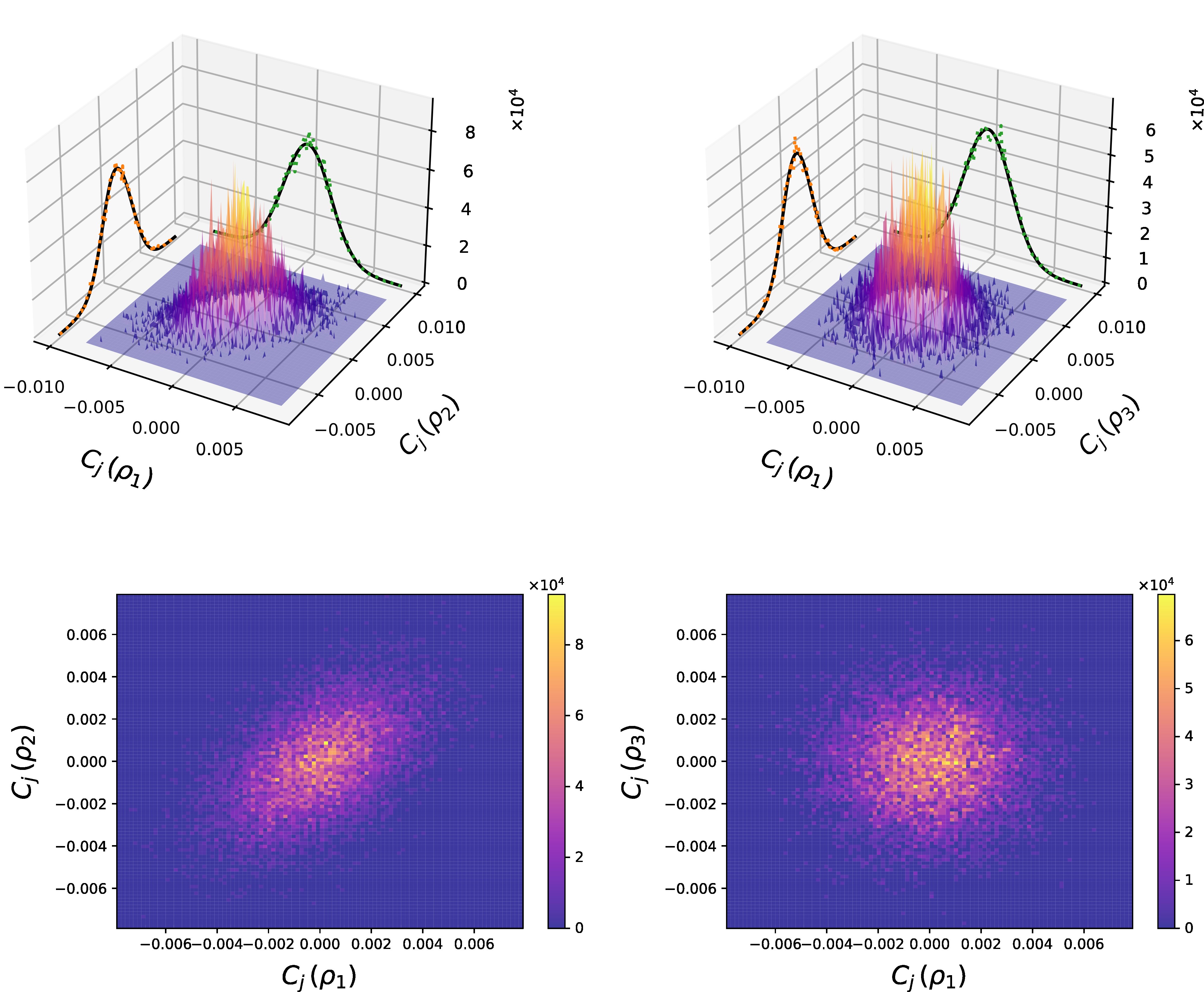}

    \caption{\textbf{Two-dimensional GPs.} We plot the joint probability density function, as well as its scaled marginals, for the measurement outcomes at the output of a unitary Haar random QNN acting on $n=18$ qubits.  The measured observable is $O_j=Z_1$, where $Z_1$ denotes the Pauli $z$ operator on the first qubit. Moreover, the input states are: $\rho_1=\dya{0}^{\otimes n}$ and $\rho_2=\dya{{\rm GHZ}}$ with $\ket{{\rm GHZ}}=\frac{1}{\sqrt{2}}(\ket{0}^{\otimes n}+\ket{1}^{\otimes n})$, for the left panel; $\rho_1$ and $\rho_3=\dya{\Psi}$ with $\ket{\Psi}=\frac{1}{\sqrt{d}}\ket{0}^{\otimes n}+\sqrt{1-\frac{1}{d}}\ket{1}^{\otimes n}$ for the right panel. In both cases we took $10^4$ samples.  }
    \label{fig:multivariate_gauss}
\end{figure}

In the right panel of Fig.~\ref{fig:multivariate_gauss}, we plot the  GP corresponding to two initial states such that $\Tr[\rho_i\rho_{i'}]=\frac{1}{d}$. In this case, the variables appear uncorrelated as predicted by Theorem~\ref{th:gaus-indept}. Importantly, in  the SI we show that when $\Tr[\rho_i\rho_{i'}]\in o(\frac{1}{\poly(\log(d))})$ for all $\rho_i\neq \rho_{i'}$, $\mathscr{C}$ will form an uncorrelated GP if one takes the covariance matrix to be approximately diagonal in the large-$d$ limit. Then, in the Methods we show that the results of Theorems~\ref{th:gp_main} and~\ref{th:gaus-indept} are valid for generalized datasets, where the conditions on the overlaps need only be met on average.

\subsection{Negatively correlated GPs}

Here we study the case of orthogonal states, i.e. when $\Tr[\rho_i\rho_{i'}]=0$  for all $\rho_i\neq \rho_{i'}\in\mathscr{D}$. We prove the following theorem. 

\begin{theorem} \label{th:GP-negative}
 Let $\mathscr{C}$ be a vector of expectation values of an operator in $\mathscr{O}$ over a set of states from $\mathscr{D}$, as in Eq.~\eqref{eq:random_vector}. If $\Tr[\rho_i\rho_{i'}]=0$ for all $i\neq i'$, then in the large-$d$ limit $\mathscr{C}$ forms a GP  with mean vector $\vec{\mu}=\vec{0}$ and covariance matrix
\begin{equation}
\vec{\Sigma}_{i, i'}^{\mathbb{U}(d)} = \frac{\vec{\Sigma}_{i, i'}^{\mathbb{O}(d)} }{2} =\begin{cases}
    \frac{\Tr[\rho_i^2]}{d}\quad\text{if $i=i'$} \\
    -\frac{1}{d^2}\quad\;\, \text{if $i\neq i'$}\,
\end{cases}.  
\end{equation}
\end{theorem}
Note that the magnitude of the covariances is $\Theta(\frac{1}{d^2})$ while that of the variances is $\Theta\left(\frac{1}{d\poly(\log(d))}\right)$. That is, in the large-$d$ limit the covariances are much smaller than the variances.

\subsection{Deep QNN outcomes, and their linear combination}

\begin{figure}[t!]
    \centering
    \includegraphics[width=\columnwidth]{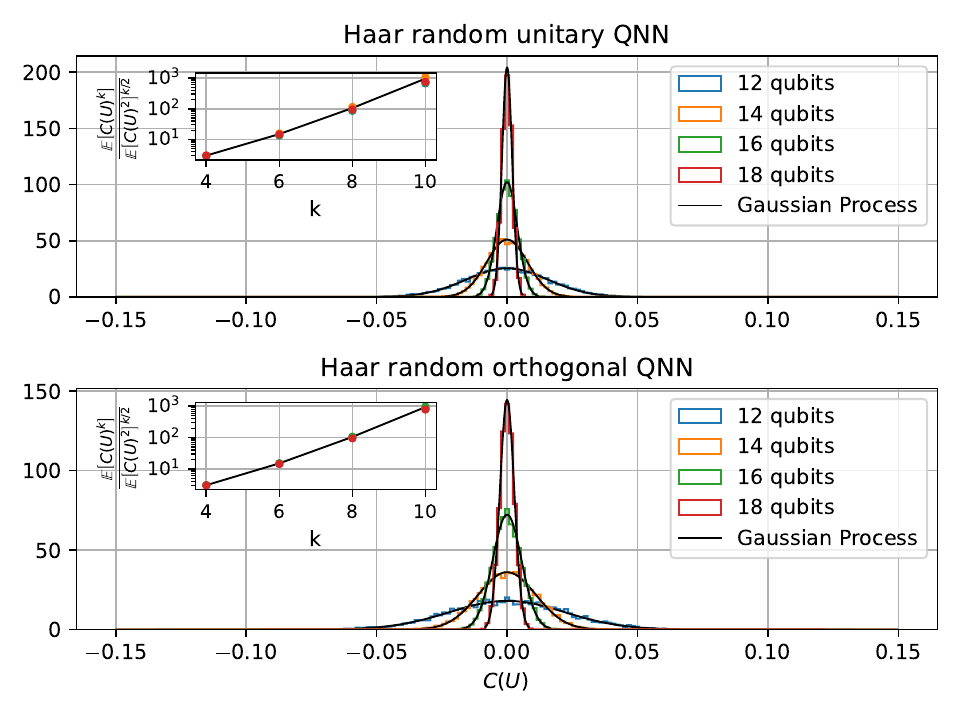}
    \caption{\textbf{Probability density function for $C_j(\rho_i)$, for Haar random  QNNs and different problem sizes.} We consider 
    unitary and orthogonal QNNs with $n$ qubits, and we take $\rho_i=\dya{0}^{\otimes n}$, and $O_j=Z_1$. The colored histograms are built from $10^4$ samples in each case, and the solid black lines represent the corresponding Gaussian distributions $\NC\left(0,\sigma^2\right)$, where $\sigma^2$ is given in Corollary~\ref{cor:gaussian}. The insets show the numerical versus predicted value of $\mathbb{E}[C_j(\rho_i)^k]/\mathbb{E}[C_j(\rho_i)^2]^{k/2}$. For a Gaussian distribution with zero mean,  such quotient is  $\frac{k!}{2^{k/2}(k/2)!}$ (solid black line).}
    \label{fig:Pauli_GP}
\end{figure}

In this section, and the following ones, we will study the implications of  Theorems~\ref{th:gp_main},~\ref{th:gaus-indept} and~\ref{th:GP-negative}. Unless stated otherwise, the corollaries we present can be applied to all considered  datasets (see Table~\ref{tab:train_states}).  

First, we study the univariate probability distribution  $P(C_j(\rho_i))$. 
\begin{corollary}\label{cor:gaussian}
    Let $C_j(\rho_i)$ be the expectation value of a Haar random QNN as in Eq.~\eqref{eq:cij}.   Then, for any $\rho_i\in \mathscr{D}$ and $O_j\in \mathscr{O}$, we have 
    \begin{equation}
        P(C_j(\rho_i))=\NC(0,\sigma^2)\,,
    \end{equation}
 where $\sigma^2=\frac{1}{d},\frac{2}{d}$ when $U$ is Haar random over $\mathbb{U}(d)$ and $\mathbb{O}(d)$, respectively.
\end{corollary}

Corollary~\ref{cor:gaussian} shows that when a single state from $\mathscr{D}$  is sent through the QNN, and a single operator from $\mathscr{O}$ is measured, the outcomes follow a Gaussian distribution with a variance that vanishes inversely proportional to the Hilbert space dimension.  This means that for large problem sizes, we can expect the results to be extremely concentrated around their mean (see below for more details). In Fig.~\ref{fig:Pauli_GP} we compare the predictions from Corollary~\ref{cor:gaussian} to numerical simulations. We find that the simulations match our theoretical results very closely, for both the unitary and the orthogonal groups. Moreover, we can observe that the standard deviation for orthogonal Haar random QNNs is larger than that for unitary ones. In Fig.~\ref{fig:Pauli_GP} we also plot the quotient $\frac{\mathbb{E}[C_j(\rho_i)^k]}{\mathbb{E}[C_j(\rho_i)^2]^{k/2}}$ obtained from our numerics, and we verify that  it follows the value $\frac{k!}{2^{k/2}(k/2)!}$ of a Gaussian distribution.

At this point, it is worth making an important remark. According to Definition~\ref{def:GP}, if $\mathscr{C}$ forms a GP, then any linear combination of its entries will follow a univariate Gaussian distribution. In particular, if $\{C_j(\rho_1),C_j(\rho_2),\ldots,C_j(\rho_m)\}\subseteq \mathscr{C}$, then $P(C_j(\widetilde{\rho}))$ with $\widetilde{\rho}=\sum_{i=1}^m c_i \rho_i$ will be equal to $\NC(0,\widetilde{\sigma}^2)$ for some $\widetilde{\sigma}$. Note that the real-valued coefficients $\{c_i\}_{i=1}^m$ need not be a probability distribution, meaning that $\widetilde{\rho}$ is not necessarily a quantum state. The previous then raises an important question: What happens if $\widetilde{\rho}\propto \id$? A direct calculation shows that $C_j(\widetilde{\rho})=\sum_{i=1}^m  C_j(c_i \rho_i)\propto \Tr[U\id U\ad O_j]=\Tr[O_j]=0$. How can we then unify these two perspectives? On the one hand  $C_j(\widetilde{\rho})$ should be normally distributed, but on the other hand we know that it is always constant. To solve this issue, we note that the only dataset we considered where the identity can be constructed is the one where $\Tr[\rho_i\rho_{i'}]=0$ for all $i\neq i'$\footnote{This follows from the fact that if $\mathscr{D}$ contains a complete basis then for any $\widetilde{\rho}\in \mathscr{D}^\perp$, one has that if $\Tr[\widetilde{\rho}\rho_i]=0$ for all $\rho_i\in\mathscr{D}$, then $\widetilde{\rho}=0$. Here,  $\mathscr{D}^\perp$ denotes the kernel of the projector onto the subspace spanned by the vectors in $\mathscr{D}$.  }. In that case, we can leverage Theorem~\ref{th:GP-negative} along with the identity $\widetilde{\sigma}^2=
\Var_{G}[\sum_{i=1}^dC_j(\rho_i)]=\sum_{i,i'}{\rm Cov}_{G}[C_j(\rho_i),C_j(\rho_{i'})]$  to explicitly prove that  $\Var_G[\sum_{i=1}^dC_j(\rho_i)]=0$ (for $G=\mathbb{U}(d),\mathbb{O}(d)$).  Hence, we find a zero-variance Gaussian distribution, i.e., a delta distribution in the QNN's outcomes (as expected).

\subsection{Predictive power of the GP for qubit systems}

Consider we are given a (potentially continuous) set $\mathcal{D}$ of $n$-qubit states  and the following task, divided into two phases. In a first \textit{data-aquisition phase}, one is allowed to send some states from $\mathcal{D}$ through some fixed unknown unitary $V$ acting on $n'\leq n$ qubits, perform measurements, and record the outcomes. Crucially, the unitary $V$ need not be Haar random, but could be, in principle,  any unitary (even the identity). Then, during a second \textit{prediction phase}, access to $V$ is no longer granted, but one has to predict the value of $\Tr[V\rho_iV^\dagger O]$ for any  $\rho_i\in\mathcal{D}$, where $O$ is some fixed Pauli string acting on the $n'$ qubits. 

While one could opt for some tomographic approach to learn $V$ and solve the previous task,  our work enables the use of the predictive power of the GP. In particular, one starts with the prior that $V$ could be any unitary in $\mathbb{U}(2^{n'})$, and therefore the probability distribution of $\Tr[V\rho_iV^\dagger O]$ is a univariate Gaussian  as per our main theorems (assuming $\DC$ satisfies the appropriate conditions). In the first stage, one  measures the expectation value $\Tr[V\rho_iV\ad O]$ for some training set $\mathscr{D}\subset\DC$. Then,  during the second stage one computes the overlaps between the states in $\mathscr{D}$ (to build the covariance matrix), as well as with any new state from $\DC$ on which we wish to apply the predictive power of the GP. These measurements can then  be used to update the prior and make predictions (see the  Methods for details on  this procedure). 

As evidenced from Lemma~1, the entries of the covariance matrix are  suppressed as $2^{n'}$, i.e., exponentially in $n'$. Hence, and as explained in the Methods, this implies that if $V$ acts on all $n$ qubits (or on $\Theta(n)$ qubits), then an exponential number of measurements will be needed in order to  use Bayesian inference to learn any information about new outcomes given previous ones. However, the situation becomes much more favorable if the QNN acts on $n'\in\mathcal{O}(\log(n))$ qubits, as here only a polynomial number of measurements are needed to use the GP's predictive power (provided that the overlaps between the $n'$-qubit quantum states are not super-polynomially vanishing in $n$). 
In fact, we show in Fig.~\ref{fig:4} simulations on up to $n=200$ qubits where a GP is used as a regression tool to interpolate/extrapolate and accurately predict measurement results at the output of a quantum dynamical process (see the Methods for details).  

\begin{figure}[t]
    \centering
    \includegraphics[width=\linewidth]{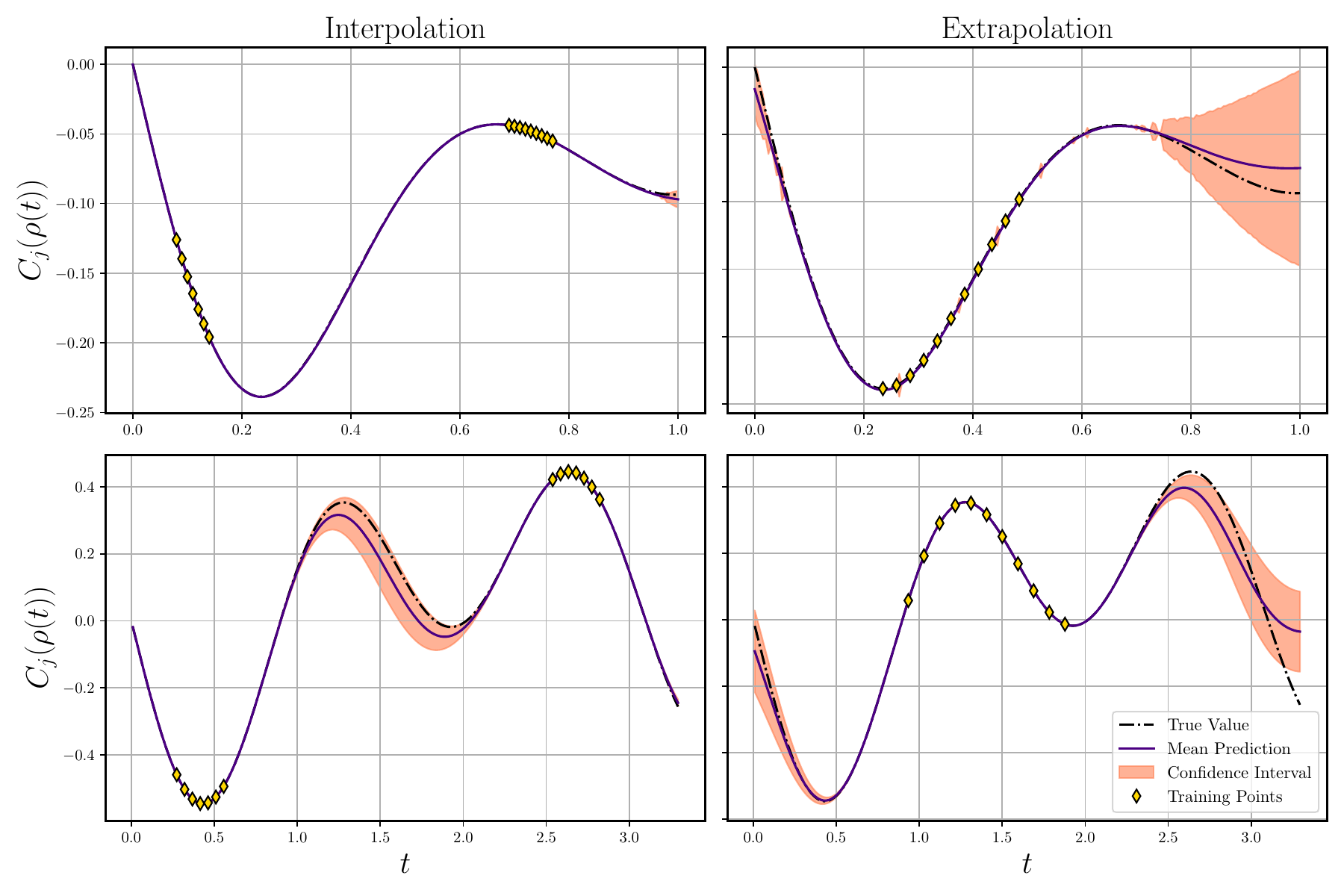}
    \caption{{\bf Quantum GP regression.}  The plots show the time evolution of two \emph{local} random Pauli operators of an $n$-qubit system under an $XY$ Hamiltonian with random transverse fields in one (bottom panels, $n=200$) and two spatial dimensions (top panels, $n=25$). The details can be found in the Methods. The dots are used as observations from which predictions are inferred. The latter correspond to the solid purple line, with the shaded regions indicating a two-sigma ($\sim95\%$) confidence interval. We also plot the true value (black line) for reference.}
    \label{fig:4}
\end{figure}

\subsection{Concentration of measure}

In this section, we show that Corollary~\ref{cor:gaussian}  provides a more precise characterization of the concentration of measure and the barren plateau phenomena for Haar random circuits than that found in the literature~\cite{mcclean2018barren,cerezo2020cost,marrero2020entanglement,patti2020entanglement,holmes2020barren,arrasmith2021equivalence,popescu2006entanglement}. First, it implies that deep orthogonal QNNs  will  exhibit barren plateaus, a result not previously known.  Second, we recall that in standard barren plateau analyses, one only looks at the first two moments of the distribution of cost values $C_j(\rho_i)$ (or, similarly, of gradient values $\partial_{\theta} C_j(\rho_i)$). Then one uses  Chebyshev's inequality, which states that for any $c>0$, the probability $P(|X|\geq c)\leq \frac{\Var[X]}{c^2}$, to prove that  $P(|C_j(\rho_i)|\geq c)$ and  $P(|\partial_{\theta}C_j(\rho_i)|\geq c)$ are in $\OC(\frac{1}{d})$~\cite{cerezo2020cost,arrasmith2021equivalence}. However, 
 having a full characterization of $P(C_j(\rho_i))$ allows us to compute  tail probabilities and obtain a much tighter bound. For instance, for $U$ being Haar random over $\mathbb{U}(d)$, we find
\begin{corollary}\label{cor:double}
    Let $C_j(\rho_i)$ be the expectation value of a Haar random QNN as in Eq.~\eqref{eq:cij}. Assuming that there exists a parametrized gate in $U$ of the form $e^{-i \theta H}$ for some Pauli operator $H$, then 
    \begin{equation}
       P(|C_j(\rho_i)|\geq c), \,P(|\partial_{\theta}C_j(\rho_i)|\geq c)\in\OC\left(\frac{1}{ce^{dc^2}\sqrt{d}}\right)\,.\nonumber
    \end{equation}
\end{corollary}
Corollary~\ref{cor:double} indicates that the QNN outputs, and their  gradients, actually  concentrate with a probability which vanishes exponentially with $d$. In an $n$-qubit system, where $d=2^n$, then $P(|C_j(\rho_i)|\geq c)$ and $P(|\partial_{\theta}C_j(\rho_i)|\geq c)$ are doubly exponentially vanishing with $n$. The tightness of our bound arises from the fact that Chebyshev's inequality is loose for highly narrow Gaussian distributions. Moreover, our bound is also tighter than that provided by Levi's lemma~\cite{popescu2006entanglement}, as it includes an extra $\OC\left(\frac{1}{\sqrt{d}}\right)$ factor. Corollary~\ref{cor:double} also implies that the  narrow gorge region of the landscape~\cite{arrasmith2021equivalence}, i.e., the  fraction of non-concentrated $C_j(\rho_i)$ values, also decreases exponentially with $d$. 

In the Methods we furthermore show how our results can be used to study the concentration of functions of QNN outcomes, e.g., standard loss functions used in the literature, like the mean-squared error.

\subsection{Implications for $t$-designs}

We now note that our results allow us to characterize the output distribution for QNN's that form $t$-designs, i.e., for QNNs whose unitary  distributions have the same properties up to the first $t$ moments as sampling random unitaries from $\mathbb{U}(d)$ with respect to the Haar measure. With this in mind, one can readily see that the following corollary holds.
\begin{corollary}\label{cor:t-design}
    Let $U$ be drawn from a $t$-design. Then, under the same conditions for which Theorems~\ref{th:gp_main},~\ref{th:gaus-indept} and~\ref{th:GP-negative} hold, the vector $\mathscr{C}$ matches the first $t$ moments of a GP.
\end{corollary}
Corollary~\ref{cor:t-design} extends our results beyond the strict condition of the QNN being Haar random to being a $t$-design, which is a more realistic assumption~\cite{harrow2009random,harrow2018approximate,haferkamp2022random}. In particular, we can study the concentration phenomenon in $t$-designs: using an extension of Chebyshev's inequality to higher order moments leads to $P(|C_j(\rho_i)|\geq c), \,P(|\partial_{\theta}C_j(\rho_i)|\geq c)\in\OC\left(\frac{\left(2\left\lfloor\frac{t}{2}\right\rfloor\right)!}{2^{\left\lfloor\frac{t}{2}\right\rfloor }(dc^2)^{ \left\lfloor\frac{t}{2}\right\rfloor}\left(\left\lfloor\frac{t}{2}\right\rfloor\right)!}\right)$ (see the SI for a proof). Note that for $t=2$ we recover the known concentration result for barren plateaus, but for $t\geq 4$ we obtain new polynomial-in-$d$-tighter bounds.

\section{ Discussion and Outlook}

In this manuscript we have shown that under certain conditions, the output distribution of deep Haar random QNNs converges to a Gaussian process in the limit of large Hilbert space dimension. While this result had been conjectured in~\cite{liu2021representation}, a formal proof was still lacking. We remark that although our result mirrors its classical counterpart --that certain classical NNs form GPs--, there exist nuances that differentiate our findings from the classical case. For instance, we need to make assumptions on the states processed by the QNN, as well as on the measurement operator. Moreover, some of these assumptions are unavoidable, as Haar random QNNs will not necessarily always converge to a GP. That is, not all QNNs  and all measurements will lead to a GP. As an example, we have that if $O_j$ is a projector onto a computational basis state, then one recovers a Porter-Thomas distribution~\cite{porter1956fluctuations}. Ultimately, these subtleties arise because the entries of unitary matrices are not independent. In contrast, classical NNs are not subject to this constraint.

It is worth noting that our theorems have further implications beyond those discussed here. First and foremost, the fact that GPs can be efficiently used for regression in certain cases paves the way for new and exciting research avenues at the intersection of quantum information and Bayesian learning. Moreover, we envision that our methods and results will be useful in more general settings where Haar random unitaries or $t$-designs are considered, such as quantum scramblers and black holes~\cite{hayden2007black,oliviero2022black,holmes2020barren}, many-body physics~\cite{nahum2018operator}, quantum decouplers and quantum error correction~\cite{brown2015decoupling}. Finally, we leave for future work the study of whether GPs arise in other architectures, such as matchgate circuits~\cite{jozsa2008matchgates,diaz2023showcasing}.

\section{Methods}

\subsection{Sketch of the proof of our main results}

Since our main results are mostly based on Lemmas~\ref{lem:exp-cov} and~\ref{lem:moments}, we will here outline the main steps used to prove these lemmas. In particular, to prove them,  we need to calculate, in the large-$d$ limit, quantities of the form
\beq \label{eq:moments} 
\mathbb{E}_{G}\left[\Tr\left[U^{\otimes k} \Lambda  (U\ad)^{ \otimes k} O^{\otimes k}\right] \right] \,, 
\eeq 
for arbitrary $k$,  and for $G=\mathbb{U}(d),\mathbb{O}(d)$. Here, the operator $\Lambda$ is defined as $\Lambda=\rho_{i_1}\otimes\cdots\otimes\rho_{i_k}$, where the states $\rho_{i}$ belong to $\mathscr{D}$, and where $O$ is an operator in $\mathscr{O}$. The first moment ($k=1$), $\vec{\mu}$, and the second moments ($k=2$), $\vec{\Sigma}^G_{i,i'}$ can be directly computed using standard formulas for integration over the unitary and orthogonal groups (see the SI). This readily recovers the results in Lemma~\ref{lem:exp-cov}. However, for larger $k$, a direct computation quickly becomes intractable, and we need to resort to asymptotic Weingarten calculations. More concretely, let us exemplify our calculations for the unitary group and for the case when the states in the dataset are such that $\Tr[\rho_{i}\rho_{i'}]\in\Omega\left(\frac{1}{\poly(\log(d))}\right)$ for all $\rho_i,\rho_{i'}\in\mathscr{D}$. As shown in the SI, we can  prove the following lemma.

\begin{lemma}\label{lem:twirl_u}
Let $X$ be an operator in $\BC(\HC^{\otimes k})$, the set of bounded linear operators acting on the $k$-fold tensor product of a $d$-dimensional Hilbert space $\HC$. Let $S_k$ be the symmetric group on $k$ items, and let $P_d$ be the subsystem permuting representation of $S_k$ in $\HC^{\otimes k}$. Then, for large Hilbert space dimension ($d\rightarrow\infty$), the twirl of $X$ over $\mathbb{U}(d)$ is
\begin{equation}
\begin{split}
    \mathbb{E}_{\mathbb{U}(d)}[U^{\otimes k}X(U\ad)^{\otimes k}]=&\frac{1}{d^k}\sum_{\sigma\in S_k}\Tr[XP_d(\sigma)]P_d(\sigma^{-1})\\&+\frac{1}{d^k}\sum_{\sigma,\pi\in S_k}c_{\sigma,\pi}\Tr[XP_d(\sigma)]P_d(\pi)\,,\nonumber\end{split}
\end{equation}
where the constants $c_{\sigma,\pi}$ are in $\OC(1/d)$.    
\end{lemma}
We recall that the 
subsystem permuting representation of a permutation $\sigma\in S_k$ is 
\begin{equation}\label{eq:rep-S_k-main}
P_d(\sigma)=\sum_{i_1,\dots,i_k=0}^{d-1} |i_{\sigma^{-1}(1)},\dots,i_{\sigma^{-1}(k)} \ra\la i_1,\dots,i_k|\,.
\end{equation}

Lemma~\ref{lem:twirl_u} implies that~\eqref{eq:moments} is equal to

\beq \label{eq:moments_weingarten_u} \begin{split} \mathbb{E}_{\mathbb{U}(d)}&\left[\Tr\left[U^{\otimes k} \Lambda  (U\ad)^{ \otimes k} O^{\otimes k}\right] \right]\\ &=\frac{1}{d^k}\sum_{\sigma\in S_k}\Tr[\Lambda P_d(\sigma)]\Tr[P_d(\sigma^{-1})O^{\otimes k} ] \\& +\frac{1}{d^k}\sum_{\sigma,\pi\in S_k}c_{\sigma,\pi}\Tr[\Lambda P_d(\sigma)]\Tr[P_d(\pi)O^{\otimes k}] \,.\end{split} \eeq 
We now note that, by definition, since $O$ is traceless and such that $O^2=\id$, then $\Tr[P_d(\sigma)O^{\otimes k}]=0$ for odd $k$ (and for all $\sigma$). This result  implies that all the odd moments are exactly zero, and also that the non-zero contributions in Eq.~\eqref{eq:moments_weingarten_u} for the even moments come from permutations consisting of cycles of even length. We remark that as a direct consequence, the first  moment, $\mathbb{E}_{\mathbb{U}(d)}\left[ \Tr[U\rho_{i}U\ad O] \right]$, is zero for any $\rho_{i}\in \mathscr{D}$, and thus we have $\vec{\mu}=\vec{0}$. 
To compute higher moments, we show that $\Tr[P_d(\sigma)O^{\otimes k} ]=d^{r}$ if  $k$ is even and $\sigma$ is a product of $r$ disjoint cycles of even length. The maximum of $\Tr[P_d(\sigma)O^{\otimes k} ]$ is therefore achieved when $r$ is maximal, i.e., when $\sigma$ is a product of $k/2$ disjoint transpositions (cycles of length two), leading to $\Tr[P_d(\sigma)O^{\otimes k} ]=d^{k/2}$.
Then, we look at the factors $\Tr[\Lambda P_d(\sigma)]$ and include them in the analysis. We have that for all $\pi$ and $\sigma$ in $S_k$,
\begin{align} 
    \frac{1}{d^k}  \Big|(c_{\sigma,\pi}&\Tr[\Lambda P_d(\sigma)]\Tr[P_d(\pi)O^{\otimes k}]\\+c_{\sigma^{-1},\pi}&\Tr[\Lambda P_d(\sigma^{-1})])\Tr[P_d(\pi)O^{\otimes k}]\Big| \in\OC\left(\frac{1}{d^{\frac{k+2}{2}}}\right)\,. \nonumber
\end{align}
Moreover, since $\Tr[\rho_{i}\rho_{i'}]\in\Omega\left(\frac{1}{\poly(\log(d))}\right)$ for all pair of states $\rho_i,\rho_{i'}\in\mathscr{D}$, it holds that if $\sigma$ is a product of $k/2$ disjoint transpositions, then 
     \begin{equation}
         \frac{1}{d^k}\Tr[\Lambda P_d(\sigma)]\Tr[P_d(\sigma^{-1})O^{\otimes k} ]\in\widetilde{\Omega}\left(\frac{1}{d^{k/2}}\right)\,,
     \end{equation}
      where the $\widetilde{\Omega}$ notation omits $\poly(\log(d))^{-1}$ factors, whereas
     \begin{equation} \begin{split}
         \frac{1}{d^k}\Big|&\Tr[\Lambda P_d(\sigma)]\Tr[P_d(\sigma^{-1})O^{\otimes k} ]+\\ &\Tr[\Lambda P_d(\sigma^{-1})]\Tr[P_d(\sigma)O^{\otimes k} ]\Big|\in\OC\left(\frac{1}{d^{\frac{k+2}{2}}}\right) \,,\end{split}
     \end{equation}
     for any other $\sigma$. We remark that if $\sigma$ consist only of transpositions, then it is its own inverse, that is, $\sigma=\sigma^{-1}$.

It immediately follows that for fixed $k$ and $d\rightarrow\infty$, the second sum in Eq.~\eqref{eq:moments_weingarten_u} is suppressed at least inversely proportional to the dimension of the Hilbert space with respect to the first one (i.e. exponentially in the number of qubits for QNNs made out of qubits)\footnote{We remark that  as long as $k$ scales with $d$ as $\mathcal{O}(\log \log d)$, our asymptotic analysis and hence the convergence to a GP are still valid.
This can be seen from the fact that there are $k! - \frac{k!}{2^{k/2}(k/2)!}$ permutations that are not the product of disjoint transpositions. 
Hence,  we find
$
\frac{k! - \frac{k!}{2^{k/2}(k/2)!}}{\frac{k!}{2^{k/2}(k/2)!}} = \frac{1- \frac{1}{2^{k/2}(k/2)!}}{\frac{1}{2^{k/2}(k/2)!}} \approx 2^{k/2}(k/2)!\approx  2^{k/2} \sqrt{\pi \log\log d} \left(\frac{ \log\log d}{e}\right)^{\log\log d}  \nonumber  < \sqrt{\log\log d} \,(\log\log d)^{\log\log d} = \sqrt{\log\log d} \, (\log d)^{\log\log\log d}\,,
$
where we used Stirling approximation for the factorial, and replaced $k=\log\log d$. As this ratio is quasi-polynomial in $\log d$, but all contributions that arise from permutations that are not the product of disjoint transpositions are suppressed as $\mathcal{O}(\frac{1}{d})$, the conclusion  follows.}. Likewise, the contributions in the first sum in~\eqref{eq:moments_weingarten_u} coming from permutations that are not the product of $k/2$ disjoint transpositions are also suppressed at least inversely proportional to the Hilbert space dimension. Therefore, in the large-$d$ limit we arrive at 
\small
\begin{equation}
 \label{eq:gp_moments-met} 
\mathbb{E}_{\mathbb{U}(d)}\left[\Tr\left[U^{\otimes k} \Lambda  (U^{\dagger})^{\otimes k} O^{\otimes k}\right] \right]= \frac{1}{d^{k/2}}\sum_{\sigma\in T_{k}}\prod_{\{t,t'\}\in \sigma} \Tr[\rho_{t}\rho_{t'}] \,, \end{equation}
\normalsize
where we have defined $T_k\subseteq S_k$ to be the subset of permutations which are exactly given by a product of $k/2$ disjoint transpositions. Note that this is precisely the statement in Lemma~\ref{lem:moments}.

From here we can easily see that if every state in $\Lambda$ is the same, i.e., if $\rho_{i_\lambda}=\rho$ for $\lambda=1,\ldots,k$, then $\Tr[\rho_{t}\rho_{t'}]=1$ for all $t,t'$, and we need to count how many terms are there in Eq.~\eqref{eq:gp_moments-met}. Specifically, we need to count how many different ways there exist to split $k$ elements into pairs (with $k$ even). A straightforward calculation shows that 
\begin{equation}
    \sum_{\sigma\in T_{k}}\prod_{\{t,t'\}\in \sigma} 1=\frac{1}{(k/2)!}\binom{k}{2,2,\dots,2}  = \frac{k!}{2^{k/2} (k/2)!} \,.
\end{equation} 
Therefore, we arrive at 
\begin{equation}
\mathbb{E}_{\mathbb{U}(d)}\left[\Tr\left[U^{\otimes k} \Lambda  (U^{\dagger})^ {\otimes k} O^{\otimes k}\right] \right]= \frac{1}{d^{k/2}} \frac{k!}{2^{k/2}(k/2)!}\,.
\end{equation} 
Identifying $\sigma^2=\frac{1}{d}$ implies that the moments $\mathbb{E}_{\mathbb{U}(d)}\left[\Tr\left[U\rho U\ad O\right]^k \right]$ exactly match those of a Gaussian distribution $\NC(0,\sigma^2)$. 

To prove that these moments unequivocally  determine the distribution of $\mathscr{C}$, we use Carleman's condition.

\begin{lemma}[Carleman's condition, Hamburger case~\cite{kleiber2013multivariate}]\label{lem:carleman-main}
Let $\gamma_k$ be the (finite) moments of the distribution of a random variable $X$ that can take values on the real line $\mathbb{R}$. These moments determine uniquely the distribution of $X$ if
\beq \sum_{k=1}^\infty \gamma_{2k}^{-1/2k} = \infty \ .\eeq
\end{lemma}
Explicitly, we have
\begin{align}
    \sum_{k=1}^\infty \left(\frac{1}{d^k} \frac{(2k)!}{2^k k!}\right)^{-1/2k} &= \sqrt{2d} \sum_{k=1}^\infty \left((2k)\cdots(k+1)\right)^{-1/2k}\nonumber\\
    &\geq   \sum_{k=1}^\infty \left((2k)^k\right)^{-1/2k}\nonumber\\
    &= \sum_{k=1}^\infty \frac{1}{\sqrt{2k}} = \infty \,.
\end{align}
Hence, according to  Lemma~\ref{lem:carleman-main}, Carleman's condition is satisfied, and $P(C_j(\rho_i))$ is distributed following a Gaussian distribution.

A similar argument can be given to show that the moments of $\mathscr{C}$ match those of a GP. Here, we need to compare Eq.~\eqref{eq:gp_moments-met}  with the $k$-th order moments of a GP, which are provided by  Isserlis  theorem~\cite{isserlis1918formula}. Specifically, if we want to compute a $k$-th order  moment of a GP, then we have that $\mathbb{E}[X_{1}X_{2}\cdots X_{k}]=0$ if $k$ is odd, and
\begin{equation}\label{eq:wick-met}
    \mathbb{E}[X_{1}X_{2}\cdots X_{k}]=
        \sum_{\sigma\in T_{k}}\prod_{\{t,t'\}\in \sigma} {\rm Cov} [X_{t},X_{t'}]\,,
\end{equation}
if $k$ is even. Clearly, Eq.~\eqref{eq:gp_moments-met} matches Eq.~\eqref{eq:wick-met} by identifying ${\rm Cov} [X_{t},X_{t'}]=\frac{\Tr[\rho_t \rho_{t'}]}{d}$. We can again prove that these moments uniquely determine the distribution of $\mathscr{C}$ from the  fact that since its marginal distributions are determinate via Carleman's condition (see above), then so is the distribution of $\mathscr{C}$~\cite{kleiber2013multivariate}.  Hence, $\mathscr{C}$ forms a GP.

\subsection{Generalized datasets}

Up to this point we have derived our theorems by imposing strict conditions on the overlaps between every pair of states in the dataset. However, we can extend these results to the cases where the conditions are only met on average when sampling states over $\mathscr{D}$. 
\begin{theorem}\label{cod:dataset-average}
    The results of Theorems~\ref{th:gp_main} and~\ref{th:gaus-indept} will hold, on average, if $\mathbb{E}_{\rho_i,\rho_{i'}\sim \mathscr{D}}\Tr[\rho_i\rho_{i'}]\in\Omega\left(\frac{1}{\poly(\log(d))}\right)$ and  $\mathbb{E}_{\rho_i,\rho_{i'}\sim \mathscr{D}}\Tr[\rho_i\rho_{i'}]=\frac{1}{d}$, respectively.
\end{theorem}

In Theorem~\ref{cod:dataset-average} we generalized the results of Theorems~\ref{th:gp_main} and~\ref{th:gaus-indept} to hold  on average when a) $\mathbb{E}_{\rho_i,\rho_{i'}\sim \mathscr{D}}\Tr[\rho_i\rho_{i'}]\in\Omega\left(\frac{1}{\poly(\log(d))}\right)$ and  b) $\mathbb{E}_{\rho_i,\rho_{i'}\sim \mathscr{D}}\Tr[\rho_i\rho_{i'}]=\frac{1}{d}$, respectively. Interestingly, these two cases have practical relevance. Let us start with Case a). Consider a multiclass classification problem, where each state $\rho_i$ in $\mathscr{D}$ belongs to one of $Y$ classes, with $Y\in\OC(1)$, and where the dataset is composed of an (approximately) equal number of states from each class. That is, for each $\rho_i$ we can assign a label $y_i=1,\ldots,Y$. Then, we assume that the classes are well separated in the Hilbert feature space, a standard and sufficient assumption for the model to be able to solve the learning task~\cite{lloyd2020quantum,larocca2022group}. By well separated we mean that
\begin{align}
&\Tr[\rho_i\rho_{i'}]\in\Omega\left(\frac{1}{\poly(\log(d))}\right)\,,\quad\text{if}\quad y_i=y_{i'}\,,\\
     &\Tr[\rho_i\rho_{i'}]\in\OC\left(\frac{1}{2^n}\right)\,,\quad\text{if}\quad y_i\neq y_{i'}\,.
\end{align}
In this case, it can be verified that for any pair of states $\rho_i$ and $\rho_{i'}$ sampled from $\mathscr{D}$, one has $\mathbb{E}_{\rho_i,\rho_{i'}\sim\mathscr{D}}[\Tr[\rho_i\rho_{i'}]]\in\Omega\left(\frac{1}{\poly(\log(d))}\right)$.

Next, let us evaluate Case b). Such situation arises precisely if the sates in $\mathscr{D}$ are Haar random states. Indeed, we can readily show that 
\small
\begin{align} \mathbb{E}_{\rho_i,\rho_{i'}\sim\mathscr{D}}[\Tr[\rho_i\rho_{i'}]]&=\mathbb{E}_{\rho_i,\rho_{i'}\sim {\rm Haar}}[\Tr[\rho_i\rho_{i'}]]\nonumber\\
    &=\int_{\mathbb{U}(d)}d\mu(U)d\mu(V) \Tr[U\rho_0U\ad V\rho_{0}'V\ad]\nonumber\\
    &=\int_{\mathbb{U}(d)}d\mu(U) \Tr[U\rho_0U\ad\rho_{0}']\nonumber\\
    &=\frac{\Tr[\rho_0]\Tr[\rho_{0}']}{d}\nonumber\\
    &=\frac{1}{d}\,.
\end{align}
\normalsize
Here, in the first equality we have used that sampling Haar random pure states $\rho_i$ and $\rho_{i'}$ from the Haar measure is equivalent to taking two reference pure states $\rho_0$ and $\rho_{0}'$ and evolving them with Haar random unitaries. In the second equality we have used the left-invariance of the Haar measure, and in the third equality we have explicitly performed the integration (see the SI).

\subsection{Learning with the GP}

In this section we will review the basic formalism for learning with GPs, and then discuss conditions under which such learning will be efficient.

Let $\vec{C}$ be a Gaussian process. Then, by definition, given a collection of inputs $\{x_i\}_{i=1}^m$, $\vec{C}$ is determined by its $m$-dimensional mean vector $\vec{\mu}$, and its $m\times m$-dimensional covariance matrix $\vec{\Sigma}$. In what follows we will assume that the mean of $\vec{C}$ is zero, and that the entries of its covariance matrix are expressed as $\kappa(x_i,x_{i'})$.
That is,

\footnotesize
\begin{equation}
    P\left(\!\!\begin{pmatrix}
        C(x_1)\\
        \vdots\\
        C(x_m)
    \end{pmatrix}\!\!\right)=\NC\left(\!\!\vec{\mu}=\begin{pmatrix}
        0\\
        \vdots\\
        0
    \end{pmatrix},\vec{\Sigma}=\begin{pmatrix}
        \kappa(x_1,x_1) & \cdots & \kappa(x_1,x_m) \\
        \vdots & & \vdots \\
        \kappa(x_m,x_1) & \cdots & \kappa(x_m,x_m) 
    \end{pmatrix}\!\!\right)\,.\nonumber
\end{equation}
\normalsize
The previous allows us to know that, \textit{a priori}, the distribution of values for any $f(x_i)$ will take the form
\begin{equation}\label{eq:normal-1}
    P(C(x_i))=\NC(0,\sigma_i^2)\,,
\end{equation}
with $\sigma_i^2=\kappa(x_i,x_i)$.

Now, let us consider the task of using $m$ observations, which we will collect in a vector $\vec{y}$, to predict the value at $x_{m+1}$. First, if the observations are noiseless, then $\vec{y}=(y(x_1),\cdots, y(x_m))$ is equal to $\vec{C}=(C(x_1),\cdots, C(x_m))$. That is, $\vec{C}=\vec{y}$. Here, we can use the fact that $C$ forms a Gaussian process to find~\cite{rasmussen2006gaussian,mukherjee2020preparation}
\begin{align}
    P(C(x_{m+1})|\vec{C})&=P(C(x_{m+1})|C(x_1),C(x_2),\ldots,C(x_m))\nonumber\\
    &=\NC\left(\mu(C(x_{m+1})),\sigma^2(C(x_{m+1}))\right)\,,\label{eq:normal-kp1}
\end{align}
where $\mu(C(x_{m+1}))$ and $\sigma^2(C(x_{m+1}))$ respectively denote the mean and variance of the  associated Gaussian probability distribution, and which are given by
\begin{align} \label{eq:gp-prediction-mu}
    \mu(C(x_{m+1}))&=\vec{m}^T \cdot \vec{\Sigma}^{-1}\cdot \vec{C}\\
    \sigma^2(C(x_{m+1}))&=\sigma_{m+1}^2-\vec{m}^T\cdot \vec{\Sigma}^{-1}\cdot \vec{m} \,. \label{eq:gp-prediction-sigma}
\end{align}
The vector $\vec{m}$ has entries $\vec{m}_i=\kappa(x_{m+1},x_i)$.  We can compare Eqs.~\eqref{eq:normal-1} and~\eqref{eq:normal-kp1} to see that using Bayesian statistics to obtain the predictive distribution of $P(C(x_{m+1})|\vec{C})$ shifts the mean from zero to $\vec{m}^T \cdot \vec{\Sigma}^{-1}\cdot \vec{C}$ and the variance is decreased from $\sigma_{m+1}^2$ by a quantity $\vec{m}^T\cdot \vec{\Sigma}^{-1}\cdot \vec{m}$. The decrease in variance follows from the fact that we are incorporating knowledge about the observations, and thus decreasing the uncertainty.

From the above discussion, we can provide some intuition behind the differences between our three main theorems. From Eq.~\eqref{eq:gp-prediction-sigma}, it is clear that we can learn the most when the states in the dataset, and the new state $\rho_{m+1}$, are similar (Theorem~\ref{th:gp_main}). 
Intuitively, this makes sense as the more similar the training states are, the better we can predict the output through the QNN of a new state closely resembling the training set.
One can readily verify that if one wishes to make predictions on a new state $\rho_{m+1}$ for which $\Tr[\rho_{m+1}\rho_i]=\frac{1}{d}$ for all states $\rho_i$ in the training set, then  $\boldsymbol{m}=\boldsymbol{0}$, meaning that we cannot update the prior (Theorem~\ref{th:gaus-indept}). This again makes perfect sense, as an overlap of $\frac{1}{d}$ is precisely the expected overlap between a Haar random state and any other pure state. This result thus implies that we cannot use training data to make predictions on a Haar random $\rho_{m+1}$. Finally, the case of orthogonal states in Theorem~\ref{th:GP-negative} is fundamentally different from the uncorrelated one, due to the fact that two generic states are not expected to be orthogonal, and thus we can still extract information as per Eq.~\eqref{eq:gp-prediction-sigma}.


In a realistic scenario, we can expect that noise will occur during our observation procedure. For simplicity we model this noise as Gaussian noise, so that $y(x_i)=C(x_i)+\varepsilon_i$, where the noise terms $\varepsilon_i$ are assumed to be independently drawn from the same distribution $P(\varepsilon_i)=\NC(0,\sigma_N^2)$. Now, since we have assumed that the noise is drawn independently, we know that the likelihood of obtaining a set of observations $\vec{y}$ given the model values $\vec{C}$ is given by $P(\vec{y}|\vec{C})=\NC(\vec{C},\sigma_N^2\id)$. In this case, we can  find the probability distribution~\cite{rasmussen2006gaussian,mukherjee2020preparation}
\small
\begin{align}
     P(C(x_{m+1})|\vec{C})&=\int d\vec{C}P(x_{m+1}|\vec{C})P(\vec{C}|\vec{y})\nonumber\\
     &=\int d\vec{C}P(C(x_{m+1})|\vec{C})P(\vec{y}|\vec{C})P(\vec{C})/P(\vec{y})\nonumber\\
     &=\NC\left(\widetilde{\mu}(C(x_{m+1})),\widetilde{\sigma}^2(C(x_{m+1}))  \right)\,,
\end{align}
\normalsize
where now we have
\begin{align}
\widetilde{\mu}(C(x_{m+1}))&=\vec{m}^T \cdot( \vec{\Sigma}+\sigma_N^2\id)^{-1}\cdot \vec{C}\\
    \widetilde{\sigma}^2(C(x_{m+1})) &=\sigma_{m+1}^2-\vec{m}^T\cdot (\vec{\Sigma}+\sigma_N^2\id)^{-1}\cdot \vec{m}\,. \label{eq:prediction-noise}
\end{align}
In the first and the second equality we have used the explicit decomposition of the probability, along with Bayes and marginalization rules.
We can see that the probability is still governed by a Gaussian distribution but where the inverse of $\vec{\Sigma}$ has been replaced by the inverse of $\vec{\Sigma}+\sigma_N^2\id$. 

The previous results can be readily used to study whether learning with the GP will be efficient in the presence of finite sampling. First, let us assume that the QNN acts on all the qudits of the states in $\mathscr{D}$, and that we measure the same $O_j$ at the output of the circuit. As such, the noise terms $\varepsilon_i$ are taken to be drawn from the same distribution $P(\varepsilon_i)=\NC(0,\sigma_N^2)$ with $\sigma_N^2=\frac{1}{N}$, and $N$ the number of shots used to estimate each $y(\rho_i)$.
 In this case, we can prove that the GP cannot be used to efficiently predict the outputs of the QNN via Bayesian statistics, as stated in the following theorem, whose proof can be found in the SI.
\begin{theorem}\label{eq:pred-dist}
    Consider a GP obtained from a Haar random QNN. Given the set of observations $(y(\rho_1),\ldots,y(\rho_m))$ obtained from $N\in\OC(\poly(\log(d)))$ measurements, then  the predictive distribution of the GP is trivial:
    \small
    \begin{equation}
P(C_j(\rho_{m+1})|C_j(\rho_{1}),\ldots,C_j(\rho_{m}))=P(C_j(\rho_{m+1}))=\NC(0,\sigma^2)\,,\nonumber
    \end{equation}
    \normalsize
where $\sigma^2$ is given by Corollary~\ref{cor:gaussian}.
\end{theorem}
Specifically, Theorem~\ref{eq:pred-dist} shows that by spending only a poly-logarithmic-in-$d$ (polynomial in $n$) number of measurements, one cannot use Bayesian statistical theory to learn any information about new outcomes given previous ones. The key insight behind Theorem~\ref{eq:pred-dist} is that the covariance-matrix entries are suppressed as $\OC\left(\frac{1}{d}\right)$  while the noise terms produce a statistical variance that is inversely proportional to the number of measurements. Hence, $\vec{\Sigma} + \sigma_N^2\id\approx\sigma_N^2\id$ in the large-$d$ limit.

Next, for simplicity, let us focus on the case when the system is composed of $n$ qubits, so that the Hilbert space dimension is $d=2^n$ (as in the main text). Moreover, we assume that the QNN and the measurement operator $O_j$ act on $m\leq n$ qubits and that expectation values are again measured with $N\in\OC(\poly(\log(d)))$ shots. When $m\in\OC(\log(n))$, Lemma~\ref{lem:exp-cov} tells us that the covariance-matrix entries are only suppressed as $\Omega\left(\frac{1}{\poly(n)}\right)$, provided that the overlaps on the reduced states on $m$ qubits are in $\Omega\left(\frac{1}{\poly(n)}\right)$. Since $\sigma_N^2=\frac{1}{N}$, it suffices to choose $N$ polynomially large in $n$ to attain $\vec{\Sigma}  + \sigma_N^2\id\approx \vec{\Sigma}$ in the large-$d$ limit.

\subsection{Details for the numerical simulations}

We provide here the details of the numerical simulations showcasing GP regression, see Fig.~\ref{fig:4}.  In order to create the dataset $\DC$, we consider a quantum dynamical process in which an initial state $\rho(0)$ is evolved under an $XY$ Hamiltonian with local random transverse fields to produce the state $\rho(t)$ at time $t$. Therefore, the states in $\DC$ are states at arbitrary times, and the learning task consists on making predictions in some dynamical process. More precisely, we define the Hamiltonians in one and two spatial dimensions as
\begin{equation}
    H_{1} = \sum_{l=1}^{n-1} X_l X_{l+1} + Y_l Y_{l+1} + \sum_{l=1}^n h_l Z_l\,,
\end{equation}
and
\small
\begin{equation}
    H_{2} = \sum_{l=1}^{\sqrt{n}-1}\sum_{l'=0}^{\sqrt{n}-1} X_{l+l'\sqrt{n}} X_{l+l'\sqrt{n}+1} + Y_{l+l'\sqrt{n}} Y_{l+l'\sqrt{n}+1} + \sum_{l=1}^n h_l Z_l\,,
\end{equation}
\normalsize
respectively, where the coefficients $h_l$ are uniformly drawn from $[-1,1]$ and $X_l,Y_l,Z_l$ indicate the usual $X,Y,Z$ Pauli matrices acting on qubit $l$. For the one-dimensional lattice we choose a system size of $n=200$ qubits, while for the two-dimensional  square lattice we have $n=5\times5=25$ qubits. The initial states are $\rho(0)=\ketbra{0}{0}^{\otimes n}$ and $\rho(0)=\ketbra{+}{+}^{\otimes n}$, respectively, with $\ket{+}=\frac{1}{\sqrt{2}}(\ket{0}+\ket{1})$. We then randomly pick a Pauli operator $O_j$ with support on at most $\lceil\log (n)\rceil$ qubits, namely, $O_j=\id^{\otimes 4}\otimes Y\otimes Y$ on 6 qubits for the one-dimensional lattice and $O_j=\id^{\otimes 3}\otimes X\otimes Z$ on 5 qubits for the two-dimensional lattice. The goal is to predict the time series $\Tr[\rho(t) O_j]$ using GP regression (in particular, Eqs.~\eqref{eq:gp-prediction-mu} and~\eqref{eq:gp-prediction-sigma} as explained in the previous subsection), given access to $m$ training points of the form $\{\rho(t_i),C_j(\rho(t_i))\}_{i=1}^m$.

\subsection{Concentration of functions of QNN outcomes}
In the main text we have evaluated the distribution of QNN outcomes and their linear combinations. However, in many cases one is also interested in evaluating a function of the elements of $\mathscr{C}$. For instance, in a standard QML setting the QNN outcomes are used to compute some loss function $\LC(\mathscr{C})$ which one wishes to optimize~\cite{biamonte2017quantum,cerezo2022challenges,cerezo2020variationalreview,liu2021representation,schuld2021machine}. While we do not aim here at exploring all possible relevant functions $\LC$, we will present two simple examples that showcase how our results can be used to study the distribution of $\LC(\mathscr{C})$, as well as its concentration. 

First, let us consider the case when $\LC(C_j(\rho_i))=C_j(\rho_i)^2$. It is well known that given a random variable with a Gaussian distribution $\NC(0,\sigma^2)$, then its square follows a Gamma distribution $\Gamma(\frac{1}{2},2\sigma^2)$. Hence, we know that $P(\LC(C_j(\rho_i)))=\Gamma(\frac{1}{2},2\sigma^2)$. Next, let us consider the case when $\LC(C_j(\rho_i))=(C_j(\rho_i)-y_i)^2$ for $y_i\in[-1,1]$. This case is relevant for supervised learning as the mean-squared error loss function is composed of a linear combination of such terms. Here,  $y_i$ corresponds to the label associated to the state $\rho_i$. We can exactly compute all the moments of $\LC(C_j(\rho_i))$ as
\begin{equation}
    \mathbb{E}_G[\LC(C_j(\rho_i))^k]=\sum_{r=0}^{2k} \binom{2k}{r}\mathbb{E}_G[C_j(\rho_i)^r](-y_i)^{2k-r}\,,
\end{equation}
for $G=\mathbb{U}(d),\mathbb{O}(d)$. We can then use Lemma~\ref{lem:moments} to obtain 
\footnotesize
\begin{align}
    \mathbb{E}_{\mathbb{U}(d)}[C_j(\rho_i)^r]&=\frac{r!}{d^{r/2}2^{r/2} (r/2)!}=\frac{\mathbb{E}_{\mathbb{O}(d)}[C_j(\rho_i)^r]}{2^{r/2}}\,,\nonumber
\end{align}
\normalsize
if $r$ is even, and $\mathbb{E}_{\mathbb{U}(d)}[C_j(\rho_i)^r]=\mathbb{E}_{\mathbb{O}(d)}[C_j(\rho_i)^r]=0$ if $r$ is odd. We obtain 
\begin{equation}
    \mathbb{E}_{\mathbb{U}(d)}[\LC(C_j(\rho_i))^k]=\frac{2^k}{(-d)^k} M\left(-k,\frac{1}{2},{-\frac{dy^2}{2}}\right)\,,
\end{equation}
with $M$ the Kummer's confluent hypergeometric function.

Furthermore, we can also study the concentration of $\LC(C_j(\rho_i))$ and show that $P\left(|\LC(C_j(\rho_i))-\mathbb{E}_{\mathbb{U}(d)}(\LC(C_j(\rho_i)))|\geq c\right)$, where the average $\mathbb{E}_{\mathbb{U}(d)}(\LC(C_j(\rho_i)))=y_i^2 + \frac{1}{d}$, is in $\OC\left( \frac{1}{|\sqrt{c} +y_i| e^{d |\sqrt{c}+y_i|^2}\sqrt{d}}\right)$.

\subsection{Infinitely-wide neural networks as Gaussian processes}

Finally, we will briefly review the seminal work of Ref.~\cite{neal1996priors}, which proved that artificial NNs with a single infinitely-wide hidden layer form GPs. Our main motivation for reviewing this result is that, as we will see below, the simple technique used in its derivation cannot be directly applied to the quantum case.

For simplicity, let us consider a network consisting of a single input neuron, $N_h$ hidden neurons, and a single output neuron (see Fig.~\ref{fig:Summary}).  The input of the network is $x\in\mathbb{R}$, and the output is given by
\begin{align}
f(x)&=b+\sum_{l=1}^{N_h}v_{l}h_l(x)\,,\label{eq:output}
\end{align}
where $h_l(x)=\phi(a_l+u_l x)$ models the action of each neuron in the hidden layer. Specifically, $u_l$ is the weight between the input neuron and the $l$-th hidden neuron, $a_l$ is the respective bias and $\phi$ is some (non-linear) activation function such as the hyperbolic tangent or the sigmoid function. Similarly, $v_{l}$ is the weight connecting the $l$-th hidden neuron to the output neuron, and $b$ is the output bias.   From Eq.~\eqref{eq:output} we can see that the output of the NN is a weighted sum of the hidden neurons' outputs plus some bias.   

Next, let us assume that the $v_{l}$ and $b$ are taken i.i.d. from a Gaussian distribution with zero mean and standard deviations $\sigma_v/\sqrt{N_h}$ and $\sigma_b$, respectively. Likewise, one can assume that the hidden neuron weights and biases are taken i.i.d. from some Gaussian distributions. Then, in the limit of $N_h\rightarrow\infty$, one can conclude via the central limit theorem that, since the NN output is a sum of infinitely many i.i.d. random variables, then it will converge to a Gaussian distribution with zero mean and variance $\sigma_b^2 + \sigma_v^2\mathbb{E}[h_l(x)^2]$. Similarly, it can be shown that in the case of multiple inputs $x_1,\dots,x_m$ one gets a multivariate Gaussian distribution for $f(x_1),\ldots,f(x_m)$, i.e., a GP~\cite{neal1996priors}. 

 Naively, one could try to mimic the technique in Ref.~\cite{neal1996priors} to prove our main results. In particular, we could start by noting that $C_j(\rho_i)$ can always be expressed as
\begin{equation}\label{eq:decomp}
C_j(\rho_i)=\sum_{k,k',r,r'=1}^d u_{kk'}\rho_{k'r}u^*_{r'r}o_{r'k}\,,
\end{equation}
where $u_{kk'}$, $u^*_{r'r}$, $\rho_{k'r}$ and $o_{r'l}$ are the matrix entries of $U$ and $U\ad$,  $\rho_i$ and $O_j$, respectively.
Although Eq.~\eqref{eq:decomp} is a summation over a large number of random variables, we cannot apply the central limit theorem (or its variants) here, since the matrix entries $u_{kk'}$ and  $u^*_{rr'}$ are not independent. 

In fact, the correlation between the entries in the same row, or column, of a Haar random unitary are of order $\frac{1}{d}$, while those in different rows, or columns, are of order $\frac{1}{d^2}$~\cite{petz2004asymptotics}. This small, albeit critical, difference makes it such that we cannot simply use the central limit theorem to prove that $\mathscr{C}$ converges to a GP. Instead, we need to rely on the techniques described in the main text.

\section*{Acknowledgements}
We acknowledge Francesco Caravelli, Fr\'ed\'eric Sauvage, Lorenzo Leone, Cinthia Huerta, Matthew Duschenes, Paolo Braccia and Antonio Anna Mele
  for useful conversations.
D.G-M. was supported by the Laboratory Directed Research and Development (LDRD) program of Los Alamos National Laboratory (LANL) under project number 20230049DR. M.L. acknowledges support by the Center for Nonlinear Studies at Los Alamos National Laboratory (LANL).  M.C. acknowledges support by the LDRD program of LANL under project number  20230527ECR. This work was also supported by LANL ASC Beyond Moore’s Law project.

\section*{Author contributions}
The project was conceived by DGM. Theoretical results were proven by MC and DGM, with ML checking them. Numerical simulations were performed by DGM and MC. All authors contributed to writing the manuscript. 

\section*{Data availability}

Data generated and analyzed during the current study are available from the corresponding author upon reasonable request.

\section*{Code availability}

Code generated during the current study is available from the corresponding author upon reasonable request.

\section*{Competing interests}

The authors declare no competing interests.

\bibliography{quantum.bib}

\makeatletter
\close@column@grid
\makeatother
\cleardoublepage
\newpage
\onecolumngrid
\renewcommand\appendixname{Supp. Info.}
\appendix

\renewcommand\figurename{Supplemental Figure}
\setcounter{figure}{0}
\setcounter{lemma}{0}
\setcounter{theorem}{0}
\setcounter{corollary}{0}

\newcounter{supfig}
\setcounter{supfig}{\value{figure}}

\newtheorem{supdefinition}{Supplemental Definition}
\newtheorem{supproposition}{Supplemental Proposition}
\newtheorem{suptheorem}{Supplemental Theorem}
\newtheorem{suplemma}{Supplemental Lemma}
\newtheorem{supcorollary}{Supplemental Corollary}

\section*{Supplemental Information for ``\textit{Quantum neural networks form Gaussian processes}''}

In this Supplemental Information (SI) we present detailed proofs of our main results. First, in Supp. Info.~\ref{si:prel} we introduce preliminary definitions that will be used throughout the rest of this SI. In Supp. Info.~\ref{si:weingarten} we review the basics of the Weingarten calculus that will allow us to compute expectation values when sampling Haar random unitaries over $\mathbb{U}(d)$ and $\mathbb{O}(d)$.  Next, in Supp. Info.~\ref{si:unitary} and~\ref{si:ortogonal} we respectively present results for twirling over the unitary and orthogonal groups.  Then, Supp. Info~\ref{sec:proof-lem1} contains the proof of  Lemma 1, while Supp. Info.~\ref{sec:proof-lem2} contains that of Lemma 2. In Supp. Info.~\ref{sec:proof-cor1} we derive Corollary 1. From here, we use these results to prove Theorems 1, 2, 3 and 4 in Supp. Info.~\ref{sec:proof-theo1},~\ref{sec:proof-theo2},~\ref{sec:proof-theo3}, and~\ref{sec:proof-theo4},  respectively. Supp. Info.~\ref{sec:proof-coro3} contains a proof for Corollary 2, and Supp. Info.~\ref{sec:proo-coro4} a proof for Corollary 3. Finally, we present the derivation of Theorem 5 in Supp. Info.~\ref{sec:theo5}.

\section{Preliminaries}\label{si:prel}

\subsection{Useful definitions}

 Let $\HC=\mathbb{C}^d$ be a complex $d$-dimensional vector space. We denote as  $\BC(\HC)$ the space of bounded linear operators in $\HC$, and by $\GL(d)$  the  general linear group of $d\times d$ complex non-singular matrices, that is, all invertible linear transformations acting on $\HC$. We now introduce the following definitions.

\textbf{Representation.} Let $G$ be a group. A \textit{representation} $R$ of $G$ on $\HC$  is a group homomorphism $R:G\to \GL(d)$. Given $R$, we recall that we can always build the \textit{$k$-fold tensor representation} of $G$, acting on $\HC^{\otimes k}$ as $R(g)^{\otimes k}$ for any $g\in G$. It is not hard to see that if $R$ is a valid representation, then its $k$-fold tensor product is also a representation.

\textbf{Character.} A useful way of characterizing representations of a group $G$ is through their associated \textit{characters}. Given a representation $R$, the character $\chi:G\rightarrow \mathbb{C} $ is a function that associates to each element of $g\in G$ a complex number
    \begin{equation}\label{eq:character}
        \chi(g)=\Tr[R(g)]\,.
    \end{equation}

\textbf{Commutant.} Given some representation $R$ of $G$, we define its \textit{$k$-th order commutant}, denoted as $\CC^{(k)}(G)$, to be the vector subspace of the space of linear operators on $\HC^{\otimes k}$ that commutes with $R(g)^{\otimes k}$ for all $g$ in $G$. That is,
\begin{equation}\label{eq:kth-commutant}
    \CC^{(k)}(G)=\{A\in\BC(\HC^{\otimes k})\quad|\quad[A,R(g)^{\otimes k}]=0\,,\quad \forall g\in G\}\,.
\end{equation}
Then, we denote as $\SC^{(k)}(G)$ a basis for $\CC^{(k)}(G)$, which consists of $D$ elements.

\textbf{Useful operators.} Let us now introduce two useful operators acting on $\HC^{\otimes 2}$. The first is the $\SWAP$ operator, whose action is to permute the two copies of $\HC$. In matrix notation, we have
\begin{equation}\label{eq:swap}
    \SWAP=\sum_{i_1,i_2=0}^{d-1}\ket{i_2,i_1}\bra{i_1,i_2}\,,
\end{equation}
and one can readily verify that for any $A$, $B$ in $\BC(\HC^{\otimes 2})$ one has
\begin{align}\label{eq:action_swap}
    (A\otimes B)\SWAP&=\SWAP (B\otimes A)\,.
\end{align}
The second operator, which we denote as $\Pi$, acts on $\HC^{\otimes 2}$ as
\begin{equation}\label{eq:Pi}
    \Pi=\sum_{i_1,i_2=0}^{d-1}\ket{i_1,i_1}\bra{i_2,i_2}\,.
\end{equation}
One can see that $\Pi$ is proportional to the projector onto the maximally-entangled Bell state between the two copies of $\HC$. That is, defining as $\ket{\Phi^+}$ the  Bell state between the two Hilbert spaces, 
\begin{equation}
    \ket{\Phi^+}=\frac{1}{\sqrt{d}}\sum_{i=0}^{d-1}\ket{i,i}\,,
\end{equation}
then one has $\Pi=d\dya{\Phi^+}$. However, we note that $\Pi$ is not a projector itself as $\Pi^2=d\Pi$. The identification of $\Pi$ with the Bell state shows that $\Pi$ satisfies the so-called \textit{ricochet} property,
\begin{equation}\label{eq:action_Pi}
    (A\otimes B)\Pi=(\id\otimes BA^T)\Pi=(AB^T\otimes \id)\Pi\,,
\end{equation}
where $A^T$ denotes the transpose of $A$.

\begin{figure}[t]
    \centering
    \includegraphics[width=1\columnwidth]{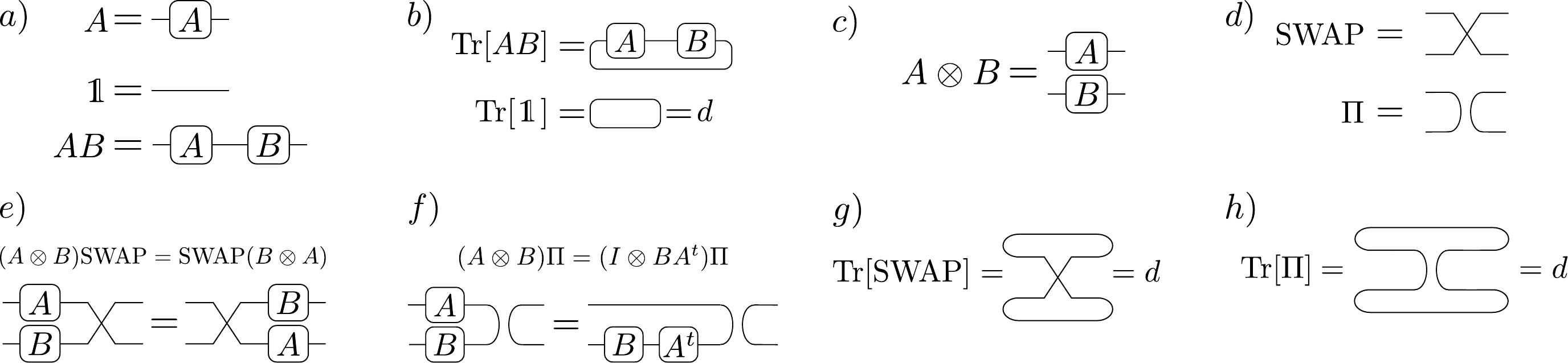}
    \caption{\textbf{Matrix and operations as tensor-network diagrams.} a) An operator $A$ acting on $\HC$ is a  $d\times d $  matrix which represents a linear map from $\HC\rightarrow \HC$. Such a mapping can be pictured as a node with two $d$-dimensional edges. The $d$-dimensional identity operator is simply a line. The product of the matrices is drawn by a node connected to another node. b) When taking the trace of a matrix, we connect its output edge with its input one. The trace of the identity, visualized as a closed loop, is equal to $d$. c) The tensor product between  matrices is represented as their tensors stacked one on top of the other. d) We show the tensor representation of the $\SWAP$ and $\Pi$ operators respectively defined in Eqs.~\eqref{eq:swap} and~\eqref{eq:Pi}. e) and f) Respectively show Eqs.~\eqref{eq:action_swap} and~\eqref{eq:action_Pi} . In g) and h) we use the tensor representation to compute the traces of the $\SWAP$ and $\Pi$ operators.   } 
    \refstepcounter{supfig}
    \label{fig:tens}
\end{figure}

 For the remainder of this SI, we will use standard tensor-network diagram notation to visually represent matrix operations (see Supp. Fig.~\ref{fig:tens}).

\section{Weingarten calculus}\label{si:weingarten}

Let  $G$ be a compact Lie group, and let $R$ be the fundamental (unitary) representation of $G$. That is,  $R(g)=g$. Given an operator $X\in\BC(\HC^{\otimes k})$, we consider the task of computing its  \textit{twirl}, $\TC^{(k)}_G[X]$,  with respect to the $k$-fold tensor representation of $G$. That is,
\begin{align}\label{eq:twirl}
    \TC^{(k)}_{G}[X]=\int_{G}d\mu(g) g^{\otimes k} X (g\ad)^{\otimes k}\,,
\end{align}
where $d\mu(g)$ is the volume element of the Haar measure. 
Now, we can use the following proposition.
\begin{supproposition}\label{prop:twirl-in-comm}
Let $\TC^{(k)}_{G}[X]$ be the twirl of an operator $X$ in $\BC(\HC^{\otimes k})$ with respect to a continuous unitary group $G$ acting on $\HC$.
Then, we have
  \begin{equation}
    \TC^{(k)}_{G}[X]\in\CC^{(k)}(G)\,.
\end{equation}  
\end{supproposition}

\begin{proof}
First, let us note that if an operator $A$ belongs to $\CC^{(k)}(G)$, then $h^{\otimes k}A(h\ad)^{\otimes k}=A$ for all $h\in G$. Then, let us compute 
\begin{align}
    h^{\otimes k}\TC^{(k)}_{G}[X](h\ad)^{\otimes k}&=\int_{G}d\mu(g) h^{\otimes k}g^{\otimes k} X (g\ad)^{\otimes k}(h\ad)^{\otimes k}\nonumber\\
    &=\int_{G}d\mu(g) (hg)^{\otimes k} X ((hg)\ad)^{\otimes k}\nonumber\\
    &=\int_{G}d\mu(g) (g)^{\otimes k} X (g\ad)^{\otimes k}\nonumber\\
    &=\TC^{(k)}_{G}[X]\,,
\end{align}
where in the third line we have used the left-invariance of the Haar measure. That is, we have used the fact that for any integrable function $f(g)$ and for any $h\in G$, we have 
    \begin{equation}
\int_{G}d\mu(g)f(hg)=\int_{G}d\mu(g)f(gh)=\int_{G}d\mu(g)f(g)\,.\label{eq:lef-r-ght-inv}
    \end{equation}

Thus, we have shown that $\TC^{(k)}_{G}[X]\in\CC^{(k)}(G)$.
\end{proof}

From Supplemental Proposition~\ref{prop:twirl-in-comm}, it follows that $\TC^{(k)}_{G}[X]$ can be expressed as 
\begin{equation}\label{eq:twirled_X_comm}
    \TC^{(k)}_{G}[X]=\sum_{\mu=1}^{D} c_\mu(X) P_\mu\,, \quad \text{with} \quad P_\mu\in \SC^{(k)}(G)\,.
\end{equation}
Hence, in order to solve Eq.~\eqref{eq:twirled_X_comm} one needs to determine the $D$ unknown coefficients $\{c_\mu(X)\}_{\mu=1}^D$. This can be achieved by finding $D$ equations to form a linear system problem. In particular, we note that the change $X\rightarrow P_\nu X $ for some $P_\nu\in \SC^{(k)}(G)$ leads to 
\begin{align}\label{eq:twirled_X_comm_2}
    \TC^{(k)}_{G}[P_\nu X  ]&=\int_{G}d\mu(g) g^{\otimes k} P_\nu X  (g\ad)^{\otimes k}\nonumber\\
    &=\int_{G}d\mu(g) P_\nu g^{\otimes k} X  (g\ad)^{\otimes k} \nonumber\\
    &=\sum_{\mu=1}^{D} c_\mu(X) P_\nu P_\mu \,,
\end{align}
where in the second line we have used the fact that $P_\nu$ belongs to the commutant $\CC^{(k)}(G)$.
Then, taking the trace on both sides of Eq.~\eqref{eq:twirled_X_comm_2} leads to 
\begin{align}\label{eq:LSP}
   \Tr[P_\nu X ] =\sum_{\mu=1}^{D}  \Tr[P_\nu P_\mu ] c_\mu(X)\,,
\end{align}
where we used the fact that $\Tr[\TC^{(k)}_{G}[X]]=\Tr[X]$. Repeating Eq.~\eqref{eq:LSP} for all $P_\nu$ in $ \SC^{(k)}(G)$ leads to $D$ equations. Thus, we can find the vector of unknown coefficients $\vec{c}(X)=(c_1(X),\ldots,c_D(X))$ by solving
\begin{equation}
    A\cdot \vec{c}(X)=\vec{b}(X)\,,
\end{equation} 
where $\vec{b}(X)=(\Tr[X P_1],\ldots, \Tr[X P_D])$. Here,  $A$ is a $D\times D$ symmetric matrix, known as the Gram matrix, whose entries are $(A)_{\nu\mu}=\Tr[P_\nu P_\mu]$. By inverting the  $A$ matrix, we can then find 
\begin{equation}\label{eq:inverse-vec-c}
    \vec{c}(X)=A^{-1}\cdot\vec{b}(X)\,.
\end{equation}
The matrix $A^{-1}$ is known as the Weingarten matrix. We refer the reader to Ref.~\cite{collins2006integration} for additional details on the Weingarten matrix.

From the previous, we can see that computing the twirl $\TC^{(k)}_{G}[X]$ in Eq.~\eqref{eq:twirl} requires calculating the matrix $A$ and inverting it (if such inverse exists). In what follows, we will consider the cases when $G$ is the unitary or  orthogonal group. For those cases we will present, in Supp. Info.~\ref{si:unitary} and~\ref{si:ortogonal} respectively, a simple explicit decomposition for the twirl into the $k$-th order commutant as in Eq.~\eqref{eq:twirled_X_comm}, which holds asymptotically in the limit of large $d$.

\section{The unitary group}\label{si:unitary}

In this section we present a series of results that will allow us to compute quantities of the form
\begin{align}
\mathbb{E}_{\mathbb{U}(d)}\left[\prod_{i=1}^k C_j(\rho_i)\right]&=\mathbb{E}_{\mathbb{U}(d)}\left[\prod_{i=1}^k \Tr[U\rho_iU\ad O_j]\right]\,.
\end{align}

\subsection{Twirling over the unitary group}

We begin by recalling that the standard representation of the unitary group of degree $d$, which we denote as $\mathbb{U}(d)$, is the group formed by all $d\times d$ unitary matrices acting on a $d$-dimensional Hilbert space $\HC$. That is
\begin{equation}
   \mathbb{U}(d)=\{U\in \GL(d)\quad|\quad U U\ad =U\ad U=\id \} \subset \GL(d)\,,
\end{equation}
where $\GL(d)$ is the general linear group, i.e., the group of all invertible matrices acting on $\HC$.
 Here, we have used the standard notation $R(U)=U$ for all the elements of the unitary group.

From the Schur-Weyl duality, we know that a basis for the $k$-th order commutant of $\mathbb{U}(d)$ is the  representation $P_d$ of the Symmetric group $S_k$ that permutes the $d$-dimensional subsystems in the $k$-fold tensor product
Hilbert space, $\HC^{\otimes k}$. That is, for a permutation $\sigma\in S_k$,
\begin{equation}\label{eq:rep-S_k}
P_d(\sigma)=\sum_{i_1,\dots,i_k=0}^{d-1} |i_{\sigma^{-1}(1)},\dots,i_{\sigma^{-1}(k)} \ra\la i_1,\dots,i_k|\,.
\end{equation}

Hence, 
\begin{equation}\label{eq:basis-comm-symmetric}
    \SC^{(k)}(\mathbb{U}(d))=\{ P_d(\sigma)\}_{\sigma\in S_k},
\end{equation}
and we note that $\SC^{(k)}(\mathbb{U}(d))$ contains $k!$ elements. 
In Supp. Fig.~\ref{fig:tens_Sk} we show the tensor representation  of the elements in $S_k$ for $k=1,2,3$, as well as an explicit illustrative example which showcases that the elements of $S_k$ commute with $U^{\otimes k}$.

\begin{figure}[t]
    \centering
    \includegraphics[width=1\columnwidth]{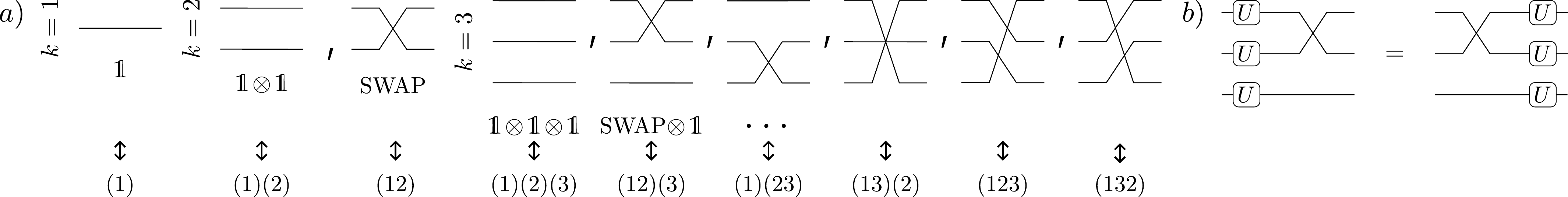}
    \caption{\textbf{Elements of $S_k$.} 
    a) We show the elements of $S_k$ for $k=1,2,3$. Here we can see that the ${\rm SWAP}$ operator in Eq.~\eqref{eq:swap} is a transposition that permutes two subsystems. Below each element of $S_k$ we also show its cycle decomposition. b) We explicitly show that an element of $S_3$ commutes with $U^{\otimes 3}$ for any $U\in\mathbb{U}(d)$.
    } 
    \refstepcounter{supfig}
    \label{fig:tens_Sk}
\end{figure}

To exemplify how one can use the previous result to compute twirls over the unitary group, let us consider the $k=1$ and $k=2$ cases. First, let $k=1$. Here, we can readily see in Supp. Fig.~\ref{fig:tens_Sk} that $S_1$ contains a single element, whose representation is given by the identity matrix $\id$. As such, the basis of the commutant contains one element,
\begin{equation}\label{eq:k_symmetries_unitary_1}
    \SC^{( 1)}(\mathbb{U}(d))=\{\id\}\,.
\end{equation}
The fact that the commutant is trivial also follows from the fact that the representation of $\mathbb{U}(d)$ with $k=1$ copies is irreducible. From here, we can build the  $1\times 1$ Gram matrix  
$A=\begin{pmatrix}
    d
    \end{pmatrix}$, 
so that the Weingarten matrix is $
    A^{-1}=\begin{pmatrix}
    \frac{1}{d}
    \end{pmatrix}
$, and $c_1(X)=\Tr[X]/d$. Hence,
\begin{align}\label{eq:twirl-k1}
    \TC^{(1)}_{\mathbb{U}(d)}[X]=\frac{\Tr[X]}{d}\id\,.
\end{align}

Next, we consider the case of $k=2$. Now the basis of the commutant contains two elements (see Supp. Fig.~\eqref{fig:tens_Sk}),
\begin{equation}\label{eq:k_symmetries_unitary_2}
    \SC^{( 2)}(\mathbb{U}(d))=\{\id\otimes \id,\,\SWAP\}\,.
\end{equation}
The ensuing Gram Matrix is
\begin{equation}\label{eq:Gram_Unitary}
    A=\begin{pmatrix}
    d^2 & d\\
    d & d^2
    \end{pmatrix}\,,
\end{equation}
and the Weingarten matrix is
\begin{equation}
    A^{-1}=\frac{1}{d^2-1}\begin{pmatrix}
    1 & \frac{-1}{d}\\
    \frac{-1}{d} & 1
    \end{pmatrix}\,.
\end{equation}
Now we find
\begin{align}
\begin{pmatrix}
    c_1(X)\\
    c_2(X)
    \end{pmatrix}&=\frac{1}{d^2-1}\begin{pmatrix}
    1 & \frac{-1}{d}\\
    \frac{-1}{d} & 1
    \end{pmatrix}\cdot \begin{pmatrix}
    \Tr[X]\\
    \Tr[X \SWAP]
    \end{pmatrix}= \frac{1}{d^2-1}\begin{pmatrix}
    \Tr[X]-\frac{\Tr[X \SWAP]}{d}\\
    \Tr[X \SWAP]-\frac{\Tr[X]}{d}
    \end{pmatrix}\,.
\end{align}
Hence, 
\begin{align}\label{eq:twirl-k2}
    \TC^{(2)}_{\mathbb{U}(d)}[X]=&\frac{1}{d^2-1}\left(\Tr[X]-\frac{\Tr[X \SWAP]}{d}\right)\id\otimes \id +\frac{1}{d^2-1}\left(\Tr[X \SWAP]-\frac{\Tr[X]}{d}\right) \SWAP\,.
\end{align}

For more general $k$ the process of building the Gram matrix and inverting it can become quite cumbersome as the matrix $A$ will be a $k!\times k!$ dimensional matrix. However, since we are interested in the large-$d$ limit, we can use the following result (presented in the main text as Lemma 3).

\begin{suptheorem}\label{theo-unitary}
Let $X$ be an operator in $\BC(\HC^{\otimes k})$, the twirl of $X$ over $\mathbb{U}(d)$, as defined in Eq.~\eqref{eq:twirl} is
\begin{align}
    \TC^{(k)}_{\mathbb{U}(d)}[X]=\frac{1}{d^k}\sum_{\sigma\in S_k}\Tr[XP_d(\sigma)]P_d(\sigma^{-1})+\frac{1}{d^k}\sum_{\sigma,\pi\in S_k}c_{\sigma,\pi}\Tr[XP_d(\sigma)]P_d(\pi)\,,
\end{align}
where the constants $c_{\sigma,\pi}$ are in $\OC(1/d)$.
\end{suptheorem}

In order to prove Supplemental Theorem~\ref{theo-unitary}, we find it convenient to recall the following definitions.

\begin{supdefinition}[Permutation Cycle]\label{def:permu_cycle} Let $\sigma$ be a permutation belonging to $S_k$. A permutation cycle $c$ is a set of indices $\{i_m, \sigma(i_m), \sigma(\sigma(i_m)), \dots \}$ that are closed under the action of $\sigma$.
\end{supdefinition}

Equipped with Supplemental Definition~\ref{def:permu_cycle}, we can now introduce the cycle decomposition of a permutation.

\begin{supdefinition}[Cycle Decomposition]\label{def:cycle-decom}
Given a permutation $\sigma \in S_k$, its cycle decomposition is an expression of $\sigma$ as a product of disjoint cycles
\begin{equation}\label{eq:cycle-decomp}
\sigma = c_1\cdots c_r\,.
\end{equation}
\end{supdefinition}
We will henceforth refer to  those indexes which are not permuted, i.e., which are contained in length-one  cycles, as fixed points. Moreover, we note that while it is usually standard to drop in the notation of Eq.~\eqref{eq:cycle-decomp} the cycles of length one, (i.e., the  cycles where an element is left unchanged), we assume that  Eq.~\eqref{eq:cycle-decomp} contains \textit{all} cycles, including those of length one (see Supp. Fig.~\ref{fig:tens_Sk}). We also remark that the cycle decomposition is unique, up to permutations of the cycles (since they are disjoint) and up to cyclic shifts within the cycles (since they are cycles). For example, we can express some permutation $\sigma\in S_5$ as $(145)(23)$ or as $(23) (451)$, that is, as the composition of a length-two cycle and a length-three cycle. Here we also recall that length-two cycles are also known as \textit{transpositions}.

While the  cycle notation is useful to identify each element in the symmetric group, we will be interested in counting how many cycles, and of what length, are contained in each $\sigma\in S_k$.  This can be addressed by defining the cycle type.
\begin{supdefinition}[Cycle type]\label{def:cycle-type}
    Given a permutation $\sigma \in S_k$, its cycle type $\nu(\sigma)$ is a vector of length $k$ whose entries indicate how many cycles of each length are present in the cycle decomposition of $\sigma$. That is, 
\begin{equation}
    \nu(\sg) = (\nu_1,\cdots,\nu_k)\,,
\end{equation}
where $\nu_j$ denotes the number of length-$j$ cycles in $\sigma$. 
\end{supdefinition}
The cycle type of every element in $S_k$ is unique. In the previous example, $\nu((145)(23))=(0,1,1,0,0)$. Similarly, expressing the identity element as $e = (1)(2)(3)(4)(5)$, it is easy to see that its cycle type is $\nu(e)=(5,0,0,0,0)$.

These definitions allow us to prove the following proposition.

\begin{supproposition}\label{prop;characters}
Let $S_k$ be the symmetric group, and let $P_d$ be the representation which permutes subsystems of the $k$-fold tensor product of a $d$-dimensional Hilbert space as in  Eq.~\eqref{eq:rep-S_k}.  
Then, the character of an element $\sigma\in S_k$   is
\begin{equation}
    \chi(\sg) = \Tr[ P_d(\sg) ]  = d^{\norm{\nu(\sg)}_1}=d^r\,,
\end{equation}
where $\nu(\sg)$ is the cycle type  of $\sigma$ as defined in Supplemental Definition~\ref{def:cycle-type}, and $r$ is the number of cycles in the cycle decomposition of $\sigma$ as in Supplemental Definition~\ref{def:cycle-decom}.
\end{supproposition}
\begin{proof}
Let us begin by re-writing Eq.~\eqref{eq:rep-S_k} in term of the cycles decomposition of $\sigma$ as in Supplemental Definition~\ref{def:cycle-decom}. That is, given $\sigma = c_1\cdots c_r$, we write

\begin{equation}\label{eq:cycles}
P_d(\sigma)=\bigotimes_{\alpha=1}^r\left(\sum_{i_1,\ldots,i_{|c_\alpha|}=0}^{d-1}  |i_{c_\alpha^{-1}(1)},\dots,i_{c_\alpha^{-1}(|c_\alpha|)} \ra\la i_1,\dots,i_{|c_\alpha|}|\right)\,,\nonumber
\end{equation}
where $|c_\alpha|$ denotes the length of the $c_\alpha$ cycle. Then, the character of $\sigma$ is 
\begin{equation}
\label{eq:almost-chi}
    \chi(\sg) = \Tr[ P_d(\sg) ]
    = \prod_{\alpha=1}^r\left(\sum_{i_1\ldots,i_{|c_\alpha|}=0}^{d-1} \Tr[ |i_{c_\alpha^{-1}(1)},\dots,i_{c_\alpha^{-1}(|c_\alpha|)} \ra\la i_1,\dots,i_{|c_\alpha|}|]\right)\,.
\end{equation}
Here we can use the fact that 
\begin{equation}
    \Tr[ |i_{c_\alpha^{-1}(1)},\dots,i_{c_\alpha^{-1}(|c_\alpha|)} \ra\la i_1,\dots,i_{|c_\alpha|}|]=\prod_{\beta=1}^{|c_\alpha|}\delta_{i_\beta,i_{c_\alpha^{-1}(\beta)}}\,,
\end{equation}
which means that for any cycle $c_\alpha$, independently of its length, one has
\begin{equation}
    \sum_{i_1\ldots,i_{|c_\alpha|}=0}^{d-1}\Tr[ |i_{c_\alpha^{-1}(1)},\dots,i_{c_\alpha^{-1}\left(|c_\alpha|\right)} \ra\la i_1,\dots,i_{|c_\alpha|}|]=\sum_{i_1\ldots,i_{|c_\alpha|}=0}^{d-1}\prod_{\beta=1}^{|c_\alpha|}\delta_{i_\beta,i_{c_\alpha^{-1}(\beta)}}=d\,.
\end{equation}
Replacing in Eq.~\eqref{eq:almost-chi} leads to
\begin{equation}
    \chi(\sg) =\prod_{\alpha=1}^rd=d^r= d^{\norm{\nu(\sg)}_1}\,.
\end{equation}
\end{proof}

\begin{figure}[t]
    \centering
    \includegraphics[width=1\columnwidth]{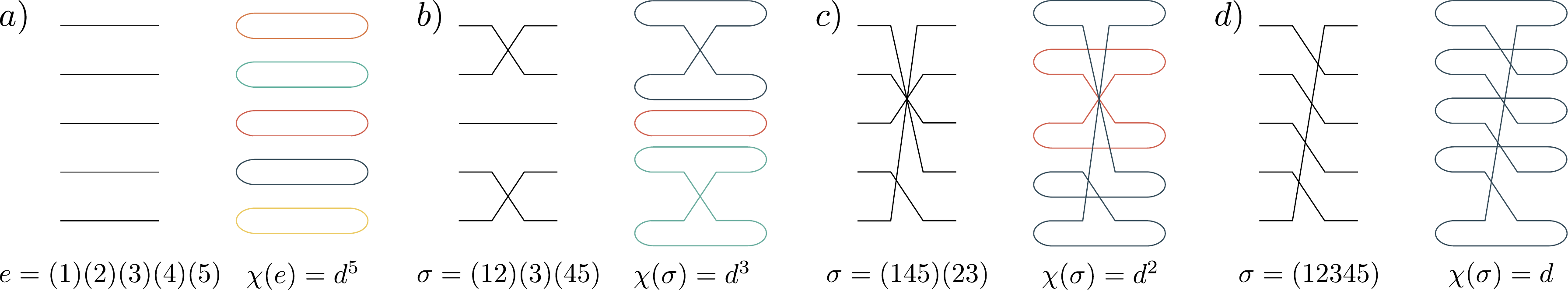}
    \caption{\textbf{Computing the character of elements of $S_5$.} 
    In all panels we present an elements of $S_5$ on the left, and then compute its trace on the right to obtain its character (see Eq.~\eqref{eq:character}).  Here we can verify that, as indicated by Proposition~\ref{prop;characters}, the character is given by $d$ raised to the number of cycles in $\sigma$.
    } 
    \refstepcounter{supfig}
    \label{fig:example-chi}
\end{figure}

In Supp. Fig.~\ref{fig:example-chi} we show an example where we compute the character for elements of $S_5$ and verify that they are indeed equal to $d^{\norm{\nu(\sg)}_1}$.

Supplemental Proposition~\ref{prop;characters} implies the following result.

\begin{supproposition}\label{prop-character-unitary}
The character of any $\sigma\in S_k$ is uniquely maximized by the identity element $e\in S_k$, in which case it  is equal to $\chi(e)=d^k$. 
\end{supproposition}
\begin{proof}
First, let us note that the identity element $e$ is composed of $k$ $1$-cycles. Thus, according to Supplemental Proposition~\ref{prop;characters}, $\chi(e)=d^k$. Next let us note that by definition, all remaining elements in $S_k\setminus \{e\}$ must contain at least some cycle that is not a $1$-cycle. This implies $\chi(\sigma)\leq d^{k-1}$ for any $\sigma\in S_k\setminus \{e\}$. In fact, it is not hard to see that the elements in $S_k$ whose character is exactly equal to $d^{k-1}$ are those composed of a single transposition. Consequently, the function $\chi:S_k \xrightarrow{} \mathbb{C}$ has a unique maximum at $\sg=e$.
\end{proof}

From here, we can show that the following result holds.

\begin{supproposition}\label{prop:products}
Let $P_d$ be the subsystem-permuting representation of $S_k$ as defined in Eq.~\eqref{eq:rep-S_k}. Given a pair of permutations $\sigma,\pi\in S_k$, then
\begin{equation}
 \Tr[P_d(\sg)P_d(\pi)] \begin{cases}
 =d^k \quad\text{if $\pi = \sg^{-1}$} \\ 
 \leq d^{k-1} \text{else}
 \end{cases}\,.
\end{equation}
\end{supproposition}
\begin{proof}
Let us begin by noting that since $S_k$ forms a group, then for any $\sigma$ and $\pi$ in $S_k$, $\xi:=\sigma\pi$ is also  in $S_k$. Moreover,  since $P_d$ is a representation, we have
\begin{equation}
     \Tr[P_d(\sg)P_d(\pi)]= \Tr[P_d(\sg \pi)]= \Tr[P_d(\xi)]=\chi(\xi)\,.
\end{equation}

As such, we now need to ask the question, how large can the character $\chi(\xi)$ be? Supplemental Proposition~\ref{prop-character-unitary} indicates that the character is maximal for the identity element. If $\xi=e$, this implies by the uniqueness of the inverse that $\pi=\sigma^{-1}$ (or equivalently  $\sigma =\pi^{-1}$). In this case we find $ \Tr[P_d(\sg)P_d(\pi)]=\chi(e)=d^k$. However, if $\pi \neq \sg^{-1}$, then $\xi \neq \id$ and via Supplemental Proposition~\ref{prop-character-unitary}, we have $\chi(\xi) \leq d^{k-1}$.
\end{proof}

Let us now go back to computing the twirl $\TC^{(k)}_{G}[X]$. First, we find it convenient to reorder the basis $\SC^{(k)}(\mathbb{U}(d))$ such that its first element are the representation of permutations which are their own inverse, i.e.,  $\sigma=\sigma^{-1}$ . Next,  we order the rest of the elements $\sigma\neq\sigma^{-1}$ by placing $P_d(\sigma)$ next to $P_d(\sigma^{-1})$. We recall  that the elements such that $\sigma=\sigma^{-1}$ are known as involutions and must  consist of a product of disjoint transpositions plus fixed points. It is well known that the number of involutions is given by $I_k=\sum_{\eta=0}^{\lfloor \frac{k}{2}\rfloor}\binom{k}{2\eta}(2\eta-1)!!$.

Then, the following result holds.
\begin{supproposition}
    The $A$ matrix, of dimension $k!\times k!$, can be expressed as
    \begin{equation}
        A= d^k(\widetilde{A}+\frac{1}{d}B)\,.
    \end{equation}
    Here we defined
    \begin{equation}
\widetilde{A}=\id_{I_k}\bigoplus_{j=1}^{\frac{k!-I_k}{2}}\begin{pmatrix}
            0 && 1 \\ 1 && 0
        \end{pmatrix}\,,
    \end{equation}
where here $\id_{I_k}$ denotes the $I_k\times I_k$ dimensional identity. Moreover, the matrix $B$ is such that its entries are  $\OC(1)$.
\end{supproposition}

More visually, the matrix $\widetilde{A}$ is of the form
\begin{equation}
\widetilde{A} = \begin{pmatrix}
\, \underbrace{\boxed{\begin{matrix} 1 & 0 & \cdots &  0 \\ 
 0& 1 & \cdots & 0 \\
\vdots & \vdots & \ddots & \vdots \\
0 & \cdots & 0 & 1 \end{matrix}}}_{I_k\times I_k} &  &   & {\makebox(0,0){\text{\huge0}}} &   \\ 
 &  \begin{matrix} 0 & 1  \\
1 & 0 \end{matrix} &     &   \\ 
 &  &  \ddots &  &   \\ 
  {\makebox(0,0){\text{\huge0}}}&    &  &  \begin{matrix} 0 & 1  \\
1 & 0 \end{matrix} \\ 
\end{pmatrix}\,.\nonumber
\end{equation}
And we remark the fact that $\widetilde{A}$ is its own inverse. That is, $\widetilde{A}^{-1}=\widetilde{A}$.

\begin{proof}
Let us recall that the entries of the matrix $A$ are of the form $A_{\nu\mu}=\Tr[P_\nu P_\mu]$ where $P_\nu,P_\mu\in \SC^{(k)}(\mathbb{U}(d))$. From Supplemental Proposition~\ref{prop:products} we  know that 
\begin{equation}
A = \begin{pmatrix}
\, \begin{matrix} d^k & a_{1,2} & \cdots &  a_{1,I_k}  & a_{1,I_{k1}} &\cdots   &  & &  a_{1,k!}\\ 
 a_{2,1}& d^k & \cdots & a_{2,I_k} &  &   &  & & \vdots\\
\vdots & \vdots & \ddots & \vdots  &  &   &  & \\
a_{I_k,1} & \cdots & a_{I_k,I_k-1}  & d^k &  &   &  &  \\
a_{I_{k+1},1} &  &   &  &  a_{I_k+1,I_k+1} & d^k \\
\vdots &  &   &  & d^k & a_{I_k+2,I_k+2}  \\ 
 &  &   &  & &  &  \ddots &  & \\
  &  &   &  & & & & a_{k!-1,k!-1} & d^k \\
  a_{k!,1}& \cdots  &   &  & & & & d^k & a_{k!,k!}
 \end{matrix}  
\end{pmatrix}\,,\nonumber
\end{equation}
where the matrix elements $a_{ij}\leq d^{k-1}$. This allows us to express the matrix $A$ as
\begin{equation}
    A= d^k(\widetilde{A}+\frac{1}{d}B)\,,
\end{equation}
where the entries in $B$ are at most equal to 1. 
    
\end{proof}

Here we present the following lemma proved in~\cite{caravelli2017complex}.

\begin{suplemma}\label{lem:inv}
Assume that the matrices $M$ and $\Omega+M$ are invertible. Then, 
\begin{equation}
    (\Omega+M)^{-1}=M^{-1}-(\id+M^{-1}\Omega)^{-1}M^{-1}\Omega M^{-1}\,.
\end{equation}
\end{suplemma}

Using Supplemental Lemma~\ref{lem:inv},  setting $M=\widetilde{A}$, $\Omega=\frac{1}{d}B$ and noting that $A\propto (\widetilde{A}+\frac{1}{d}B)$ always has inverse~\cite{collins2006integration,puchala2017symbolic}, we find
\begin{equation}
    A^{-1}=\frac{1}{d^k}\left(\widetilde{A}-\frac{1}{d}(\id+\frac{1}{d}\widetilde{A}B)^{-1}\widetilde{A}B \widetilde{A}\right)    =\frac{1}{d^k}\left(\widetilde{A}-C\right)\,,
\end{equation}
where we have defined
\begin{equation}\label{eq:C-matrix}
    C=\frac{1}{d}(\id+\frac{1}{d}\widetilde{A}B)^{-1}\widetilde{A}B \widetilde{A}\,.
\end{equation}
Now, let us write $\widetilde{A}B=QDQ^{-1}$, where $D$ is diagonal and contains the eigenvalues of $\widetilde{A}B$. Using the Perron-Frobenius theorem, it follows that $|D_{\nu\nu}|\leq k!$ for any $\nu$, and thus $(\id+\frac{1}{d}\widetilde{A}B)^{-1}= Q (\id+\frac{1}{d} D)^{-1} Q^{-1}\approx \id$ for large $d$. As a consequence, the matrix entries of $C$ are in $\OC(1/d)$. 
Combining the previous result with Eqs.~\eqref{eq:twirled_X_comm} and~\eqref{eq:inverse-vec-c} leads to 
\begin{align}
    \TC^{(k)}_{G}[X]=\frac{1}{d^k}\sum_{\sigma\in S_k}\Tr[XP_d(\sigma)]P_d(\sigma^{-1})+\frac{1}{d^k}\sum_{\sigma,\pi\in S_k}c_{\sigma,\pi}\Tr[XP_d(\sigma)]P_d(\pi)\,,
\end{align}
where the $c_{\sigma,\pi}$ are the matrix entries of $C$ as defined in Eq.~\eqref{eq:C-matrix}. This is then precisely the statement of Supplemental Theorem~\ref{theo-unitary}.

\subsection{Computing expectation values of twirled operators}

Let us  consider an expectation value of the form
\begin{equation}
    C(\rho_i)=\Tr[U\rho_i U\ad O]\,,
\end{equation}
where $\rho_i$ is a quantum state and $O$ be some traceless quantum operator such that $O^2=\id$. 
Next, let us consider the task of estimating expectation values of the form
\begin{align}
\mathbb{E}_{\mathbb{U}(d)}\left[\prod_{i=1}^k C(\rho_i)\right]&=\mathbb{E}_{\mathbb{U}(d)}\left[\prod_{i=1}^k \Tr[U\rho_{i}U\ad O]\right]\,.
\end{align}

Here, we will show that in the large-$d$ limit, the following theorem holds.
\begin{suptheorem}\label{theo-moments}
Let $\rho_i$ for $i=1,\ldots,k$ be a multiset of quantum states such that $\Tr[\rho_i\rho_{i'}]\in\Omega\left(\frac{1}{\poly(\log(d))}\right)$ for all $i,i'$,  and let  $O$ be some traceless Hermitian operator such that $O^2=\id$. Then let us define $T_k\subseteq S_k$ the set  of $k/2$ disjoint transpositions. That is, for any $\sigma\in T_k$, its cycle decomposition is $\sigma=c_1\cdots c_{k/2}$ where each $c_\alpha$ is a transposition for all $\alpha=1,\ldots,k/2$. Then, in the large-$d$ limit we have
\begin{align}
\mathbb{E}_{\mathbb{U}(d)}\left[\prod_{i=1}^k C(\rho_i)\right]&=\frac{1}{d^{k/2}}\sum_{\sigma\in T_k}\prod_{\alpha=1}^{k/2}\Tr[\rho_{c_\alpha(1)}\rho_{c_\alpha(2)}]\,.
\end{align}
\end{suptheorem}

To prove this theorem, let us first re-write
\begin{equation}
\mathbb{E}_{\mathbb{U}(d)}\left[\prod_{i=1}^k C(\rho_i)\right]=\mathbb{E}_{\mathbb{U}(d)}\left[\prod_{i=1}^k \Tr[U\rho_{i}U\ad O]\right]=\mathbb{E}_{\mathbb{U}(d)}\left[\Tr[U^{\otimes k}\Lambda \,(U\ad)^{\otimes k} O^{\otimes k}]\right]\,,
\end{equation}
where $\Lambda=\rho_{1}\otimes\cdots\otimes \rho_{k} $. Explicitly, 
\begin{align}
\mathbb{E}_{\mathbb{U}(d)}\left[\prod_{i=1}^k C(\rho_i))\right]&=\int_{\mathbb{U}(d)}d\mu(U) \Tr[U^{\otimes k}\Lambda\,(U\ad)^{\otimes k} O^{\otimes k}]]\nonumber\\
&=\Tr\left[\left(\int_{\mathbb{U}(d)}d\mu(U) U^{\otimes k}\Lambda\,(U\ad)^{\otimes k}\right) O^{\otimes k}\right]\,.
\end{align}
Here we can use Supplemental Theorem~\ref{theo-unitary} to find 
\begin{align}\label{eq:exp-prod-costs}
\mathbb{E}_{\mathbb{U}(d)}\left[\prod_{i=1}^k C(\rho_i)\right]&=\frac{1}{d^k}\sum_{\sigma\in S_k}\Tr[\Lambda P_d(\sigma)]\Tr[P_d(\sigma^{-1})O^{\otimes k} ] +\frac{1}{d^k}\sum_{\sigma,\pi\in S_k}c_{\sigma,\pi}\Tr[\Lambda P_d(\sigma)]\Tr[P_d(\pi)O^{\otimes k}]\,.
\end{align}

Now, let us prove the following proposition.
\begin{supproposition}\label{prop:powers-O}
    Let $O$ be a traceless, Hermitian and unitary operator, in which case $O^2=\id$. Then we have
    $\Tr[P_d(\sigma)O^{\otimes k} ]=0$ for any  $\sigma\in S_k$ if $k$ is odd, and $\Tr[P_d(\sigma)O^{\otimes k} ]=d^{r}$ if  $k$ is even and $\sigma$ is a product of $r$ disjoint cycles of even length. The maximum of $\Tr[P_d(\sigma)O^{\otimes k} ]$ is therefore achieved when $\sigma$ is a product of $k/2$ disjoint transpositions, leading to $\Tr[P_d(\sigma)O^{\otimes k} ]=d^{k/2}$.
\end{supproposition}

\begin{proof}
We first consider the case of $k$ being odd. Let us express $\sigma$ in its cycle decomposition $
\sigma = c_1\cdots c_r$ as in Supplemental  Definition~\ref{def:cycle-decom}. 
We have that 
\begin{equation}\label{eq:cycles-odd}
\Tr[P_d(\sigma)O^{\otimes k} ]=\prod_{\alpha=1}^r    \Tr[P_d(c_{\alpha})O^{\otimes |c_\alpha|} ]\,.
\end{equation}
Because $k$ is odd, we know that there must exist at least one cycle acting on an odd number of subsystems in the right hand side of Eq.~\eqref{eq:cycles-odd}. Let us assume that this occurs for the cycle $c_{\alpha'}$. Then, we will have 
\begin{equation}
    \Tr[P_d(c_{\alpha})O^{\otimes |c_\alpha|} ]=\Tr[O^{|c_\alpha|}]=\Tr[O]=0\,.
\end{equation}
Here we have used the fact that $O^2=\id$, and hence, since $|c_\alpha|$ is odd, we have $O^{|c_\alpha|}=O$.

Next, let us consider the case of $k$ being even. We know from Eq.~\eqref{eq:cycles-odd} that if $\sigma$ contains any cycle acting on an odd number of subsystems, then $\Tr[P_d(\sigma)O^{\otimes k} ]$ will be equal to zero. This means that only the permutations $\sigma$ composed entirely of cycles acting on even number of subsystems will have non-vanishing trace. If this is the case, we will have
\begin{equation}\label{eq:cycles-even}
\Tr[P_d(\sigma)O^{\otimes k} ]=\prod_{\alpha=1}^r    \Tr[P_d(c_{\alpha})O^{\otimes |c_\alpha|} ]=\prod_{\alpha=1}^rd=d^r\,.
\end{equation}
This follows from the fact that if $|c_\alpha|$ is even, then $O^{|c_\alpha|}=\id$.
Moreover, Eq.~\eqref{eq:cycles-even} will be maximized for the case when $r$ is largest, which corresponds to the case when $\sigma$ is a product of $k/2$ disjoint transpositions.
For this special case one finds 
\begin{equation}
\Tr[P_d(\sigma)O^{\otimes k} ]=d^{k/2}\,.
\end{equation}

\end{proof}

Next, let us prove the following result.
\begin{supproposition}\label{prop:powers-rho}
    Let $\Lambda =\rho_{1}\otimes\cdots\otimes \rho_{k} $ be a tensor product of $k$ quantum states. Then $\left|\Tr[\Lambda P_d(\sigma)]+\Tr[\Lambda P_d(\sigma^{-1})]\right|\leq 2$ for all $\sigma\in S_k$.
\end{supproposition}

\begin{proof}
    Let us again decompose $\sigma$ in its cycle decomposition $
\sigma = c_1\cdots c_r$ as in Supplemental Definition~\ref{def:cycle-decom}. First, let us assume that $\rho_1,\dots,\rho_k$ are pure. Then, we will have
\begin{equation} \label{eq:app_states_contribution}
\Tr[\Lambda P_d(\sigma)]=\prod_{\alpha=1}^r\Tr[\rho_{{c_\alpha(1)}}\cdots \rho_{{c_\alpha(|c_\alpha|)}}]= \prod_{\alpha=1}^r \bra{\psi_{c_\alpha(1)}}\ket{\psi_{c_\alpha(2)}} \cdots \bra{\psi_{c_\alpha(|c_\alpha|-1)}}\ket{\psi_{c_\alpha(|c_\alpha|)}} \bra{\psi_{c_\alpha(|c_\alpha|)}}\ket{\psi_{c_\alpha(1)}}  \,.
\end{equation}
Using that $|\bra{\psi_i}\ket{\psi_j}|\leq 1$ $\forall i,j$ and the fact that the modulus of the product of complex numbers is the product of the moduli, we find that $|\Tr[\Lambda P_d(\sigma)]|\leq 1$ for all $\sigma\in S_k$. Next, let us assume that $\rho_1,\dots,\rho_k$ are general arbitrary quantum states, i.e., they are not necessarily pure. Using that any quantum state can be written as a convex combination of orthonormal pure states, i.e., $\rho_i=\sum_{k_i} \lambda_{k_i} \ketbra{\psi_{k_i}}$ with $\lambda_{k_i}$ being real non-negative numbers such that $\sum_{k_i} \lambda_{k_i}=1$, we obtain
\begin{align} 
\Tr[\Lambda P_d(\sigma)]&=\prod_{\alpha=1}^r\Tr[\rho_{{c_\alpha(1)}}\cdots \rho_{{c_\alpha(|c_\alpha|)}}] \nonumber\\&=  \prod_{\alpha=1}^r \sum_{k_{c_\alpha(1)},\dots,k_{c_\alpha(|c_\alpha|)} } \lambda_{k_{c_\alpha(1)}} \cdots \lambda_{k_{c_\alpha(|c_\alpha|)}}\bra{\psi_{k_{c_\alpha(1)}}}\ket{\psi_{k_{c_\alpha(2)}}} \cdots \bra{\psi_{k_{c_\alpha(|c_\alpha|-1)}}}\ket{\psi_{k_{c_\alpha(|c_\alpha|)}}} \bra{\psi_{k_{c_\alpha(|c_\alpha|)}}}\ket{\psi_{k_{c_\alpha(1)}}}\,,
\end{align}
which leads to 
\begin{align} 
    |\Tr[\Lambda P_d(\sigma)]| &= \prod_{\alpha=1}^r \left|\sum_{k_{c_\alpha(1)},\dots,k_{c_\alpha(|c_\alpha|)} } \lambda_{k_{c_\alpha(1)}} \cdots \lambda_{k_{c_\alpha(|c_\alpha|)}}\bra{\psi_{k_{c_\alpha(1)}}}\ket{\psi_{k_{c_\alpha(2)}}} \cdots \bra{\psi_{k_{c_\alpha(|c_\alpha|-1)}}}\ket{\psi_{k_{c_\alpha(|c_\alpha|)}}} \bra{\psi_{k_{c_\alpha(|c_\alpha|)}}}\ket{\psi_{k_{c_\alpha(1)}}} \right|\nonumber \\ &\leq \nonumber \prod_{\alpha=1}^r \sum_{k_{c_\alpha(1)},\dots,k_{c_\alpha(|c_\alpha|)} } \lambda_{k_{c_\alpha(1)}} \cdots \lambda_{k_{c_\alpha(|c_\alpha|)}} \left|\bra{\psi_{k_{c_\alpha(1)}}}\ket{\psi_{k_{c_\alpha(2)}}} \cdots \bra{\psi_{k_{c_\alpha(|c_\alpha|-1)}}}\ket{\psi_{k_{c_\alpha(|c_\alpha|)}}} \bra{\psi_{k_{c_\alpha(|c_\alpha|)}}}\ket{\psi_{k_{c_\alpha(1)}}}\right|\\ & \leq \prod_{\alpha=1}^r \sum_{k_{c_\alpha(1)},\dots,k_{c_\alpha(|c_\alpha|)} } \lambda_{k_{c_\alpha(1)}} \cdots \lambda_{k_{c_\alpha(|c_\alpha|)}} = 1 \,,
\end{align}
where again we used that the modulus of the product of complex numbers is the product of the moduli in the first equality, the triangle inequality in the first inequality, and the fact that $|\bra{\psi_i}\ket{\psi_j}|\leq 1$ $\forall i,j$ in the second inequality.
Hence, using again the triangle inequality, we arrive at
\begin{equation}
\left|\Tr[\Lambda P_d(\sigma)]+\Tr[\Lambda P_d(\sigma^{-1})]\right|\leq \left|\Tr[\Lambda P_d(\sigma)]\right|+\left|\Tr[\Lambda P_d(\sigma^{-1})]\right|\leq 2\,,    
\end{equation}
for any $\sigma\in S_k$.
 \end{proof}

With Supplemental Propositions~\ref{prop:powers-O} and~\ref{prop:powers-rho} we can now state the following result.

\begin{supproposition}\label{prop:off-diag}
     Let $O$ be a traceless Hermitian operator such that $O^2=\id$. Let $\Lambda=\rho_{1}\otimes\cdots\otimes \rho_{k} $ be a tensor product of $k$ quantum states. Then, for all $\pi$ and $\sigma$ in $S_k$,
     \begin{equation}
       \frac{1}{d^k}  \left|(c_{\sigma,\pi}\Tr[\Lambda P_d(\sigma)]+c_{\sigma^{-1},\pi}\Tr[\Lambda P_d(\sigma^{-1})])\Tr[P_d(\pi)O^{\otimes k}]\right|\in\OC\left(\frac{1}{d^{\frac{k+2}{2}}}\right)\,.
     \end{equation}
\end{supproposition}

\begin{proof}
We begin by assuming, without loss of generality, that $|c_{\sigma^{-1},\pi}|\leq |c_{\sigma,\pi}|$. Thus, we have
\begin{equation}
   \frac{1}{d^k} \left|(c_{\sigma,\pi}\Tr[\Lambda P_d(\sigma)]+c_{\sigma^{-1},\pi}\Tr[\Lambda P_d(\sigma^{-1})])\Tr[P_d(\pi)O^{\otimes k}]\right|\leq \frac{|c_{\sigma,\pi}|}{d^k}\left|(\Tr[\Lambda P_d(\sigma)]+\Tr[\Lambda P_d(\sigma^{-1})])\Tr[P_d(\pi)O^{\otimes k}]\right|\nonumber \,.
\end{equation}
Then, from Supplemental Propositions~\ref{prop:powers-O} and~\ref{prop:powers-rho} we find
\begin{equation}
   \frac{|c_{\sigma,\pi}|}{d^k}\left|(\Tr[\Lambda P_d(\sigma)]+\Tr[\Lambda P_d(\sigma^{-1})])\Tr[P_d(\pi)O^{\otimes k}]\right|\leq \frac{|c_{\sigma,\pi}|}{d^k} 2 d^{k/2}=\frac{2|c_{\sigma,\pi}|}{d^{k/2}}\,.
\end{equation}
Since by definition, $|c_{\sigma,\pi}|\in\OC(1/d)$, we have 
\begin{equation}
    \left|\frac{1}{d^k}(c_{\sigma,\pi}\Tr[\Lambda P_d(\sigma)]+c_{\sigma^{-1},\pi}\Tr[\Lambda P_d(\sigma^{-1})])\Tr[P_d(\pi)O^{\otimes k}]\right|\in\OC\left(\frac{1}{d^{\frac{k+2}{2}}}\right)\,.
\end{equation}
    
\end{proof}

Finally, consider the following proposition.

\begin{supproposition}\label{prop:diag}
     Let $O$ be a traceless Hermitian operator such that $O^2=\id$. Let $\Lambda=\rho_{1}\otimes\cdots\otimes \rho_{k} $ be a tensor product of $k$ quantum states such that $\Tr[\rho_{i}\rho_{i'}]\in\Omega\left(\frac{1}{\poly(\log(d))}\right)$ for all $i,i'$. Then, 
     \begin{equation}
         \frac{1}{d^k}\Tr[\Lambda P_d(\sigma)]\Tr[P_d(\sigma)O^{\otimes k} ]\in\widetilde{\Omega}\left(\frac{1}{d^{k/2}}\right)
     \end{equation}
     if $\sigma$ is a product of $k/2$ disjoint transpositions, and
     \begin{equation}
         \frac{1}{d^k}\left|\Tr[|\Lambda P_d(\sigma)]\Tr[P_d(\sigma^{-1})O^{\otimes k} ]+\Tr[\Lambda P_d(\sigma^{-1})]\Tr[P_d(\sigma)O^{\otimes k} ]\right|\in\OC\left(\frac{1}{d^{\frac{k+2}{2}}}\right)
     \end{equation}
     for any other $\sigma$.
\end{supproposition}

\begin{proof}
We start by considering the case when $\sigma$ is a product of $k/2$ disjoint transpositions. We know from Supplemental Propositions~\ref{prop:powers-O} that $\Tr[P_d(\sigma)O^{\otimes k} ]=d^{k/2}$. Then, we have that 
\begin{equation}  \label{eq-ap:trace-perm-state}  \Tr[\Lambda P_d(\sigma)]=\prod_{\alpha=1}^{k/2}\Tr[\rho_{{c_\alpha(1)}}\rho_{{c_\alpha(2)}}]\in\Omega\left(\frac{1}{\poly(\log(d))}\right)\,,
\end{equation}
where we have used the fact that $\Tr[\rho_{i}\rho_{i'}]\in\Omega\left(\frac{1}{\poly(\log(d))}\right)$ for all $\rho_{i},\rho_{i'}$.
Thus, we know that 
\begin{equation}
     \frac{1}{d^k}\Tr[\Lambda P_d(\sigma)]\Tr[P_d(\sigma)O^{\otimes k} ]\in\widetilde{\Omega}\left(\frac{1}{d^{k/2}}\right)\,,
 \end{equation}
where the $\widetilde{\Omega}$ notation omits $(\poly(\log(d)))^{-1}$ factors.

 Next, let us consider the case of $\sigma$ not being a product of $k/2$ disjoint transpositions. Here, we have from Supplemental Proposition~\ref{prop:powers-O}
 \begin{equation}
 \Tr[P_d(\sigma)O^{\otimes k} ]=\Tr[P_d(\sigma^{-1})O^{\otimes k} ]=d^{r}\,,    
 \end{equation}
 with $r\leq \frac{k}{2}-1$.
 Then, the following chain of inequalities holds
  \begin{align}
     \frac{1}{d^k}\left|\Tr[\Lambda P_d(\sigma)]\Tr[P_d(\sigma^{-1})O^{\otimes k} ]+\Tr[\Lambda P_d(\sigma^{-1})]\Tr[P_d(\sigma)O^{\otimes k} ]\right|&\leq\frac{1}{d^{\frac{k+2}{2}}}\left|\Tr[\Lambda P_d(\sigma)]+\Tr[\Lambda P_d(\sigma^{-1})]\right|\nonumber\\&\leq\frac{2}{d^{\frac{k+2}{2}}}\in\OC\left(\frac{1}{d^{\frac{k+2}{2}}}\right)\,.
 \end{align}
 Where in the last line we have used Supplemental Proposition~\ref{prop:powers-rho}.
\end{proof}

To finish the proof of Supplemental Theorem~\ref{theo-moments} we simply combine Supplemental Propositions~\ref{prop:off-diag} and~\ref{prop:diag} and note that in the large-$d$ limit we get

\begin{align}\label{eq:exp-val-un}
\mathbb{E}_{\mathbb{U}(d)}\left[\prod_{\gamma=1}^k C_{\gamma}\right]&=\frac{1}{d^{k/2}}\sum_{\sigma\in T_k}\prod_{\alpha=1}^{k/2}\Tr[\rho_{{c_\alpha(1)}}\rho_{{c_\alpha(2)}}]\,,
\end{align}
where we have defined as $T_k\subseteq S_k$ the set of permutations which are exactly given by a product of $k/2$ disjoint transpositions.

Here we can additionally prove the following corollary from Supplemental Theorem~\ref{theo-moments}.
\begin{supcorollary}\label{cor:moments-gaussian-single}
Let $C(\rho_i)=\Tr[U\rho_i U\ad O]$, with   $O$ be some traceless Hermitian operator such that $O^2=\id$. Then, 
\begin{align}
\mathbb{E}_{\mathbb{U}(d)}\left[ C(\rho_i)^k\right]&= \frac{k!}{d^{k/2}2^{k/2} (k/2)!}\,.
\end{align}
\end{supcorollary}

\begin{proof}
    The proof of Supplemental Corollary~\ref{cor:moments-gaussian-single} simply follows from Supplemental Theorem~\ref{theo-moments} by noting that there are 
\begin{equation}
\frac{1}{(k/2)!}\binom{k}{2,2,\dots,2} = \frac{k!}{2^{k/2} (k/2)!} \,,
\end{equation}
elements in $T_k$.
\end{proof}

\subsection{States with overlap equal to $1/d$}

We will now consider quantum states such that $\Tr[\rho_i \rho_{i'}]=\frac{1}{d}$ for all $\rho_i\neq \rho_{i'}$ and $\Tr[\rho_i^2]\in\Omega\left(\frac{1}{\poly(\log(1/d))}\right)$ for all $\rho_i$. We prove the following Supplemental Theorem.

\begin{suptheorem} \label{sup-th:amoments_uncorrelated} Let $\rho_i$ for $i=1,\ldots,k$ be a multiset of quantum states such that $\Tr[\rho_i \rho_{i'}]=\frac{1}{d}$ for all $\rho_i\neq \rho_{i'}$ and $\Tr[\rho_i^2]\in\Omega\left(\frac{1}{\poly(\log(1/d))}\right)$ for all $\rho_i$,  and let  $O$ be some traceless Hermitian operator such that $O^2=\id$. Then, in the large-$d$ limit we have
\beq \label{eq:app_moments_uncorrelated} \mathbb{E}_{\mathbb{U}(d)}\left[\prod_{i=1}^k C(\rho_i)\right] = \frac{1}{d^{k/2}} \sum_{\sigma\in T^c_k(\Lambda)}  \prod_{(i,i')\in\sigma} \Tr[\rho_i \rho_{i'}] \,,\eeq
where we have defined $T^c_k(\Lambda)$ as the set of permutations belonging to $T_k$ that only connect identical states in $\Lambda=\rho_{1}\otimes\cdots\otimes \rho_{k} $.
\end{suptheorem}

\begin{proof}

We begin by recalling Eq.~\eqref{eq:exp-prod-costs},
\begin{align}
\mathbb{E}_{\mathbb{U}(d)}\left[\prod_{i=1}^k C(\rho_i)\right]&=\frac{1}{d^k}\sum_{\sigma\in S_k}\Tr[\Lambda P_d(\sigma)]\Tr[P_d(\sigma^{-1})O^{\otimes k} ] +\frac{1}{d^k}\sum_{\sigma,\pi\in S_k}c_{\sigma,\pi}\Tr[\Lambda P_d(\sigma)]\Tr[P_d(\pi)O^{\otimes k}]\,.
\end{align}
We know from Supplemental Proposition~\ref{prop:powers-O} that $\Tr[P_d(\sigma^{-1})O^{\otimes k}]$ is maximized by permutations $\sigma$ that are a product of $k/2$ disjoint transpositions. In particular, in that case we have $\Tr[P_d(\sigma^{-1})O^{\otimes k}]=d^{k/2}$. All other terms are suppressed by a factor $\OC(\frac{1}{d})$, and so are the $c_{\sigma,\pi}$. It is easy to see that for $\sigma\in T_k,$ $\Tr[\Lambda P_d(\sigma)]$ is maximized by the permutations that connect only identical states, because $\Tr[\rho_i\rho_{i'}]=\frac{1}{d}$ for all $\rho_i\neq \rho_{i'}$ while $\Tr[\rho_{i}^2]\in\Omega\left(\frac{1}{\poly(\log(1/d))}\right)$ for all $\rho_i$. 
Hence, in order to simultaneously maximize $\Tr[\Lambda P_d(\sigma)]$ and $\Tr[P_d(\sigma^{-1})O^{\otimes k} ]$,  $\sigma$ must be a product of $k/2$ disjoint transpositions such that each transposition connects two identical states in $\Lambda$, i.e. $\sigma$ must be in $T_k^c(\Lambda)$. More precisely, using an analogous derivation than that leading to Eq.~\eqref{eq-ap:trace-perm-state} for $\sigma\in T_k^c(\Lambda)$, we find
\begin{equation} \label{eq:app-lead-contr-uncorrelated}
     \frac{1}{d^k}\Tr[\Lambda P_d(\sigma)]\Tr[P_d(\sigma)O^{\otimes k} ]\in\widetilde{\Omega}\left(\frac{1}{d^{k/2}}\right)\,,
\end{equation}
when $\sigma\in T_k^c(\Lambda)$, and
\begin{equation} \label{eq:app-neglect-contr-uncorrelated}
     \frac{1}{d^k}\Tr[\Lambda P_d(\sigma)]\Tr[P_d(\sigma^{-1})O^{\otimes k} ]\in \OC\left(\frac{1}{d^{\frac{k+2}{2}}}\right)\,,
\end{equation}
otherwise. The latter equation follows from Supplemental Propositions~\ref{prop:powers-O} and~\ref{prop:powers-rho} when $\sigma\notin T_k$, and additionally from the condition $\Tr[\rho_i \rho_{i'}]=\frac{1}{d}$ for all $\rho_i\neq \rho_{i'}$ when $\sigma\in T_k$ but $\sigma\notin T_k(\Lambda)$. 
Furthermore, since Supplemental Proposition~\ref{prop:off-diag} also holds in this case, we recover Eq.~\eqref{eq:app_moments_uncorrelated} in the large-$d$ limit.

Notice that when $\Lambda$ is such that $T_k(\Lambda)=\emptyset$, then $\mathbb{E}_{\mathbb{U}(d)}\left[\prod_{i=1}^k C(\rho_i)\right]\in \OC\left(\frac{1}{d^{\frac{k+2}{2}}}\right)$, which is approximated by $0$ in the large-$d$ limit in Eq.~\eqref{eq:app_moments_uncorrelated}. The justification for this is the following: A set of random variables follows a joint multivariate Gaussian distribution iff any linear combination of them follows a univariate Gaussian. That is, $\YC=\sum_j  a_j C(\rho_j)$, where the $a_j$ are constants,  follows a Gaussian distribution $\NC(0,\sigma^2)$ with $\sigma^2=\sum_{j,j'} a_j a_{j'}{\rm Cov}[C(\rho_j), C(\rho_{j'})]$. Hence, the $k$-th moment of such univariate Gaussian will be a linear combination of all the $k$-th moments of the original random variables, as 
\begin{align}
    &\mathbb{E}_{\mathbb{U}(d)}\left[\YC^k\right] = \sum_{k_1+\dots+k_m=k} \binom{k}{k_1,\dots, k_m}  \mathbb{E}_{\mathbb{U}}\left[a_1^{k_1}C(\rho_1)^{k_1}\cdots a_m^{k_m} C(\rho_m)^{k_m}\right]\,,
\end{align}
Hence, in the large-$d$ limit we can neglect those contributions that are $\OC\left(\frac{1}{d^{\frac{k+2}{2}}}\right)$ against those that are $\widetilde{\Omega}\left(\frac{1}{d^{k/2}}\right)$.
\end{proof}

Finally, we explain here why when the overlaps between the states in $\mathscr{D}$ are such that $\Tr[\rho_i \rho_{i'}]\in o\left(\frac{1}{\poly(\log(1/d))}\right)$ for all $\rho_i\neq \rho_{i'}$ and $\Tr[\rho_{i}^2]\in\Omega\left(\frac{1}{\poly(\log(1/d))}\right)$ for all $\rho_i$, Supplemental Theorem~\ref{sup-th:amoments_uncorrelated} still holds. The reason is that in this case, the covariance matrix $\vec{\Sigma}_{i,i'}^{\mathbb{U}}$ can be approximated in the large-$d$ limit by that of the uncorrelated case, as $\frac{\Tr[\rho_i\rho_{i'}]}{\Tr[\rho_i^2]}\in o(1)$ for $\rho_i\neq \rho_{i'}$  (we recall from Lemma 1 that $\vec{\Sigma}_{i, i'}^{\mathbb{U}}=\frac{d}{d^2-1}\left(\Tr[\rho_{i}\rho_{i'}]-\frac{1}{d}\right)$). In turn, for the computation of higher moments it is easy to see that Eq.~\eqref{eq:app-lead-contr-uncorrelated} holds (this is direct, as we have the same condition on the purities of the states as before, i.e., $\Tr[\rho_{i}^2]\in\Omega\left(\frac{1}{\poly(\log(1/d))}\right)$ for all $\rho_i$), and so does Eq.~\eqref{eq:app-neglect-contr-uncorrelated}  when $\sigma\notin T_k$ (this follows from Supplemental Propositions~\ref{prop:powers-O} and~\ref{prop:powers-rho}). When $\sigma\in T_k$ but $\sigma\notin T_k(\Lambda)$, then we find that 
\begin{equation}
    \frac{\Tr[\Lambda P_d(\sigma)]}{\Tr[\Lambda P_d(\pi)]}\in o(1),
\end{equation}
for all $\pi\in T_k(\Lambda)$, which again is a direct consequence of the fact that  $\frac{\Tr[\rho_i\rho_{i'}]}{\Tr[\rho_i^2]}\in o(1)$ for all $\rho_i\neq \rho_{i'}$. Supplemental Proposition~\ref{prop:off-diag} also holds, as it does not make any assumption on $\Tr[\Lambda P_d(\sigma)]$. Thus, we conclude that Supplemental Theorem~\ref{sup-th:amoments_uncorrelated} is valid in this regime as well. It is important to remark here that in this case not all the terms that are neglected in Eq.~\eqref{eq:app_moments_uncorrelated} are suppressed as $\OC(\frac{1}{d})$ with respect to those that are retained. Instead, the separation is now $\OC\left({\rm max}_{i,i'}\left( \frac{\Tr[\rho_i\rho_{i'}]^2}{\Tr[\rho_i^2]\Tr[\rho_{i'}^2]}\right)\right)$. We recover the separation $\OC(\frac{1}{d})$ when the overlaps are such that $\Tr[\rho_i\rho_{i'}]\in \OC(\frac{1}{d})$ for all $\rho_i\neq \rho_{i'}$.

\subsection{Orthogonal states}

Here we will prove the following theorem for the special case when the states in $\Lambda$ are all mutually orthogonal.

\begin{suptheorem}\label{theo-moments-orth}
Let $\rho_i$ for $i=1,\ldots,k$ be a set of  mutually orthogonal quantum states such that $\Tr[\rho_i^2]\in \Omega\left(\frac{1}{\poly(\log d)}\right)$ for all $i$,  and let  $O$ be some traceless Hermitian operator such that $O^2=\id$. Then, in the large-$d$ limit we have
\begin{align}\label{eq:scaling-unit}
\mathbb{E}_{\mathbb{U}(d)}\left[\prod_{i=1}^k C(\rho_i)\right]&=\frac{k!}{2^{k/2} (k/2)!}\frac{(-1)^{k/2}}{d^k}\,.
\end{align}
\end{suptheorem}

Going back to Eq.~\eqref{eq:exp-prod-costs}, which we recall for convenience, we have 
\begin{align}
\mathbb{E}_{\mathbb{U}(d)}\left[\prod_{i=1}^k C(\rho_i)\right]&=\frac{1}{d^k}\sum_{\sigma\in S_k}\Tr[\Lambda P_d(\sigma)]\Tr[P_d(\sigma^{-1})O^{\otimes k} ] +\frac{1}{d^k}\sum_{\sigma,\pi\in S_k}c_{\sigma,\pi}\Tr[\Lambda P_d(\sigma)]\Tr[P_d(\pi)O^{\otimes k}]\,,
\end{align}
and we can  see that all the terms in the first summation are zero. This follows from the fact that $\Tr[\Lambda P_d(\sigma)]=\prod_{\alpha=1}^r\Tr[\rho_{{c_\alpha(1)}}\cdots \rho_{{c_\alpha(|c_\alpha|)}}]=0$ for all $\sigma\neq e$ (where we recall that  $P_d(e)=\id^{\otimes k}$) as all the states are orthogonal. Moreover, for the case of $\sigma= e$ one has $\Tr[P_d(e^{-1})O^{\otimes k} ]=\Tr[P_d(e)O^{\otimes k} ]=\Tr[O]^k=0$. Thus, one must here study the terms coming from the second summation.  

Following a similar argument as the one previously given, we see that the only terms that survive in the second summation are those of the form $\Tr[\Lambda P_d(e)]\Tr[P_d(\pi)O^{\otimes k}]=\Tr[P_d(\pi)O^{\otimes k}]$. 
Next, let us prove the following result.
\begin{supproposition}\label{sup:prop-ort}
    The term $\frac{1}{d^k}c_{e,\pi}\Tr[P_d(\pi)O^{\otimes k}]$ is maximized when $\pi$ is a product of $k/2$ disjoint transpositions.
\end{supproposition}

\begin{proof}

Let us start by using known results for the asymptotics of the Weingarten functions~\cite{collins2003moments}. We know that in the large-$d$ limit  $\frac{c_{e,\pi}}{d^k}=\frac{\mu(\pi)}{d^{k+|\pi|_t}}$, where $|\pi|_t$ is the smallest number of transpositions that $\pi$ is a product of, and $\mu(\pi)=\prod_{\alpha=1}^r (-1)^{|c_\alpha|_t} C_{|c_{\alpha}|_t}$, with $C_{|c_{\alpha}|_t}=\frac{1}{|c_{\alpha}|_t+1} \binom{2|c_{\alpha}|_t}{|c_{\alpha}|_t}$ the Catalan numbers. Here, we used that  $\pi$ has $r$ cycles which we denote as $c_1,\ldots,c_r$ following Supplemental Definition~\ref{def:cycle-decom}. Note that $\mu(\pi)$ is independent of $d$, and also that $C_0=C_1=1$. We can easily compute $|\pi|_t$ by noting that since each cycle can be decomposed as $|c_\alpha|-1$ transpositions we have $|\pi|_t=\sum_{\alpha=1}^r (|c_\alpha|-1)$. Note that if $\pi$ is a product of $k/2$ disjoint transpositions then $|\pi|_t=\sum_{\alpha=1}^{k/2}=\frac{k}{2}$, while if $\pi$ is a single $k$-cycle, then $|\pi|_t=\sum_{\alpha=1}^{1}(k-1)=k-1$. More generally, we have that $|\pi|_t=\sum_{\alpha=1}^r (|c_\alpha|-1)=k-r$.  Combining this result with Supplemental Proposition~\ref{prop:powers-O}, we have
\begin{equation}
    \frac{1}{d^k}c_{e,\pi}\Tr[P_d(\pi)O^{\otimes k}]=\mu(\pi)\frac{d^{2r}}{d^{2k}}
    \,.
\end{equation}
Then, assuming that all cycles are of even length, the max of $\frac{1}{d^k}c_{e,\pi}\Tr[P_d(\pi)O^{\otimes k}]$ is achieved when $\pi$ is a product of $k/2$ disjoint transpositions, and it is equal to $(-1)^{k/2}/d^k$.
\end{proof}

Using Supplemental Proposition~\ref{sup:prop-ort}, along with the fact that there are $\frac{k!}{2^{k/2} (k/2)!}$ products of $k/2$ disjoint transpositions leads to  the proof of Supplemental Theorem~\ref{theo-moments-orth},
\begin{align}\label{eq:exp-val-un-ort}
\mathbb{E}_{\mathbb{U}(d)}\left[\prod_{i=1}^k C(\rho_i)\right]&=\frac{k!}{2^{k/2} (k/2)!}\frac{(-1)^{k/2}}{d^{k}}\,.
\end{align}

Next, let us consider the case when  $\Lambda$ contains $k_1$ copies of the same state $\rho_1$, $k_2$ of the same state $\rho_2$, and so on. In total, we assume that $\Lambda$ contains $q$ different states and that $\sum_{\beta=1}^q k_{\beta}=k$. Moreover, we denote as $K_1$ as the set of indexes $k_\beta$ equal to one, and $K_2$ as the set of indexes $k_\beta$ larger or equal than 2.   That is 
\begin{equation}
    k_\beta\in\begin{cases}
        K_1\,, \quad \text{if $k_\beta=1$}\,,\\
        K_2\,, \quad \text{if $k_\beta\geq 2$}\,.
    \end{cases}
\end{equation}
We henceforth  assume that there is at least one $k_\beta$ which is larger than 2, so that $K_2\neq\emptyset$.
Next, let us denote  as $R=\sum_{k_\beta \in K_2}\lfloor\frac{k_\beta}{2}\rfloor=\sum_{k_\beta }\lfloor\frac{k_\beta}{2}\rfloor$. Note that $1\leq R \leq\frac{k}{2}$ where the upper bound is reached when all $k_\beta$ are even, and the lower bound when $K_2$ contains a single index $k_\beta$ such that $k_\beta=2$. That is, when just a single state in $\Lambda$ is repeated a single time. Finally, let us define the subsets $T_{2\lfloor k_1/2\rfloor},\ldots,T_{2\lfloor k_q/2\rfloor}\subseteq S_k$ of transpositions where $T_{2 \lfloor k_\beta/2\rfloor}$ pairs copies of the same state. Here we can prove that

\begin{suptheorem}\label{theo-moments-orth-mixed}
Let $\rho_i$ for $i=1,\ldots,k$ be a multiset of  pure  and mutually orthogonal states such that $\Tr[\rho_i^2]\in \Omega\left(\frac{1}{\poly(\log d)}\right)$ for all $i$. Then, let $\Lambda$ contain $k_1$ copies  of $\rho_1$, $k_2$ copies of $\rho_2$, and so on. In total, we assume that $\Lambda$ contains $q$ different states and that $\sum_{\beta=1}^q k_{\beta}=k$. Moreover, we assume that  $O$ is  some traceless Hermitian operator such that $O^2=\id$. Then, in the large-$d$ limit we have
\begin{align}
\mathbb{E}_{\mathbb{U}(d)}\left[\prod_{i=1}^k C(\rho_i)\right]&=(-1)^{\frac{k}{2} -R}\,\frac{d^R}{d^{k}} \left( \frac{(k - 2R)!}{2^{\frac{k}{2}-R} \left(\frac{k}{2}-R\right)!} \prod_{k_\beta\%2=1 }k_\beta \frac{(2\lfloor k_\beta/2\rfloor)!}{2^{\lfloor k_\beta/2\rfloor} (\lfloor k_\beta/2\rfloor)!} \prod_{k_\beta\%2=0 } \frac{ k_\beta!}{2^{ k_\beta/2} (k_\beta/2)!} \prod_{\beta=1}^q \Tr[\rho_\beta^2]^{\lfloor\frac{k_\beta}{2}\rfloor}\right)\,,
\end{align}
\normalsize
where $R=\sum_{k_\beta }\lfloor\frac{k_\beta}{2}\rfloor$.
\end{suptheorem}

Note that if there exists a $k_\beta=k$, i.e., all the states are the same,  then we recover the result in  Supplemental Corollary~\ref{cor:moments-gaussian-single}.

\begin{proof}

Let us consider Eq.~\eqref{eq:exp-prod-costs}, which we (again)  recall here
\begin{align}
\mathbb{E}_{\mathbb{U}(d)}\left[\prod_{i=1}^k C(\rho_i)\right]&=\frac{1}{d^k}\sum_{\sigma\in S_k}\Tr[\Lambda P_d(\sigma)]\Tr[P_d(\sigma^{-1})O^{\otimes k} ] +\frac{1}{d^k}\sum_{\sigma,\pi\in S_k}c_{\sigma,\pi}\Tr[\Lambda P_d(\sigma)]\Tr[P_d(\pi)O^{\otimes k}]\,.
\end{align}
We already know that $\Tr[P_d(\pi)O^{\otimes k}]$ is maximal when it is composed of $k/2$ disjoint transpositions. If $\left(\bigcup_{\beta=1}^q T_{2\lfloor k_\beta/2\rfloor}\right)\cap T_k\neq \emptyset$ then there will be terms in the first summation which are non-zero. In this case, we will have that for large $d$,
\begin{align}
\mathbb{E}_{\mathbb{U}(d)}\left[\prod_{i=1}^k C(\rho_i)\right]&=\frac{1}{d^k}\sum_{\sigma\in \left(\bigcup_{\beta=1}^q T_{2\lfloor k_\beta/2\rfloor}\right)\cap T_k}\Tr[\Lambda P_d(\sigma)]\Tr[P_d(\sigma^{-1})O^{\otimes k} ]\nonumber\\&=\frac{1}{d^{k/2}}\sum_{\sigma\in \left(\bigcup_{\beta=1}^q T_{2\lfloor k_\beta/2\rfloor}\right)\cap T_k} \prod_{\beta=1}^q\Tr[\rho_\beta^2]^{\frac{k_\beta}{2}}=\frac{1}{d^{k/2}}\prod_{\beta=1}^q\frac{k_\beta!}{2^{k_\beta/2} (k_\beta/2)!}\Tr[\rho_\beta^2]^{\frac{k_\beta}{2}}\,,
\end{align}
where we have used the fact that $\Tr[\Lambda P_d(\sigma)]$ can be expressed as a product of terms of the form $\Tr[\rho_\beta \rho_\beta]\in \Omega\left(\frac{1}{\poly(\log d)}\right)$, and where we have replaced $\Tr[P_d(\sigma^{-1})O^{\otimes k} ]=d^{k/2}$. Moreover,   $\sum_{\sigma\in \left(\bigcup_{\beta=1}^q T_{2\lfloor k_\beta/2\rfloor}\right)\cap T_k}$ is the number of ways in which we can pair all the states in $\Lambda$ such that only identical states are paired among themselves. Clearly, this requires that all $k_\beta$ are even, and therefore we can simply express $\sum_{\sigma\in \left(\bigcup_{\beta=1}^q T_{2\lfloor k_\beta/2\rfloor}\right)\cap T_k}=\prod_{\beta=1}^q\frac{k_\beta!}{2^{k_\beta/2} (k_\beta/2)!}$.

However, if $\left(\bigcup_{\beta=1}^q T_{2\lfloor k_\beta/2\rfloor}\right)\cap T_k= \emptyset$, or alternatively, if there is some $k_\beta$ which is odd, then all the terms in the first summation will be zero, and we need to consider the second summation. Now, the terms in the second summation that will be non-zero are the ones where $\pi$ is composed of cycles of even length, and where $\sigma$ is composed of cycles which only connect identical states. Again, we can use  known results for the asymptotics of the Weingarten functions~\cite{collins2003moments} to have that  in the large-$d$ limit  $\frac{c_{\sigma,\pi}}{d^k}=\frac{\mu(\sigma\pi)}{d^{k+|\sigma \pi|_t}}$, where now  $|\sigma\pi|_t$ is the smallest number of transpositions that the product of $\sigma$ and $\pi$ is a product of. Therefore, $\frac{1}{d^k}c_{\sigma,\pi}\Tr[P_d(\pi)O^{\otimes k}]=\frac{d^{k/2}}{d^{k+|\sigma \pi|_t}}$ will be the largest when $\sigma$ and $\pi$ are composed of the largest possible number of transpositions on the same set of indexes (as in this case $\sigma \pi$ will contain the most $1$-cycles), and $\sigma$ acts as the identity on the rest. In particular, in that case we have that
\begin{equation}
    \frac{1}{d^k}c_{\sigma,\pi}\Tr[P_d(\pi)O^{\otimes k}]=(-1)^{\frac{k}{2}-R}\,\frac{d^{k/2+R}}{d^{3k/2}}\,.
\end{equation}
Hence, we find
\begin{align}
\mathbb{E}_{\mathbb{U}(d)}\left[\prod_{i=1}^k C(\rho_i)\right]&=(-1)^{\frac{k}{2} -R}\,\frac{d^R}{d^{k}} \left( \frac{(k - 2R)!}{2^{\frac{k}{2}-R} \left(\frac{k}{2}-R\right)!} \prod_{k_\beta\%2=1 }k_\beta \frac{(2\lfloor k_\beta/2\rfloor)!}{2^{\lfloor k_\beta/2\rfloor} (\lfloor k_\beta/2\rfloor)!} \prod_{k_\beta\%2=0 } \frac{ k_\beta!}{2^{ k_\beta/2} (k_\beta/2)!} \prod_{\beta=1}^q \Tr[\rho_\beta^2]^{\lfloor\frac{k_\beta}{2}\rfloor}\right)\,,
\end{align}
where $\prod_{k_\beta\%2=1 }k_\beta \frac{(2\lfloor k_\beta/2\rfloor)!}{2^{\lfloor k_\beta/2\rfloor} (\lfloor k_\beta/2\rfloor)!}\prod_{k_\beta\%2=0 } \frac{ k_\beta!}{2^{ k_\beta/2} (k_\beta/2)!}$ is the number of ways in which one can pair the same states with themselves. Here, the first product arises from the cases when $k_\beta\geq 2$ and odd, and the second product when $k_\beta\geq 2$ and even. Additionally, the factor $\frac{(k - 2R)!}{2^{\frac{k}{2}-R} \left(\frac{k}{2}-R\right)!}$ accounts for the number of ways in which one can pair the remaining (all different) states.

\end{proof}

\begin{figure}[t] 
    \centering
\includegraphics[width=1\columnwidth]{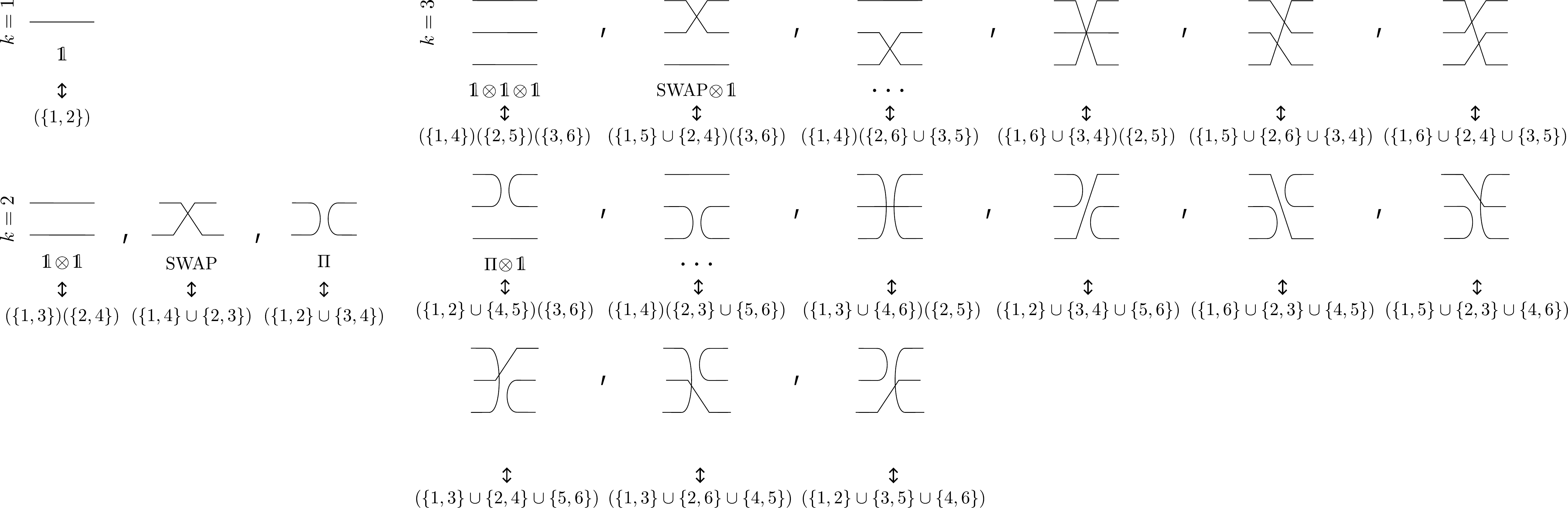}
    \caption{\textbf{Elements of $\mathfrak{B}_k$.} 
     We show the elements of $\mathfrak{B}_k$ for $k=1,2,3$. Here we can see that the ${\rm SWAP}$ and $\Pi$ in Eq.~\eqref{eq:swap} and Eq.~\eqref{eq:Pi} are operators in $B_2$. Below each element of $\mathfrak{B}_k$ we also show its cycle decomposition.
    } 
    \refstepcounter{supfig}
    \label{fig:elements_Bk}
\end{figure}

\section{The Orthogonal group}\label{si:ortogonal}

In this section we will present a series of results that will allow us to compute quantities of the form
\begin{align}
\mathbb{E}_{\mathbb{O}(d)}\left[\prod_{i=1}^k C_j(\rho_i)\right]&=\mathbb{E}_{\mathbb{O}(d)}\left[\prod_{i=1}^k \Tr[U\rho_iU\ad O_j]\right]\,.
\end{align}
\subsection{Twirling over the orthogonal group}

The standard representation of the orthogonal group of degree $d$, which we denote as $\mathbb{O}(d)$, is the group consisting of all $d\times d$ orthogonal matrices  with real entries. That is,
\begin{equation}
   \mathbb{O}(d)=\{U\in \GL(d)\quad|\quad UU\ad =U\ad U=UU^T =U^T U =\id \} \subset \GL(d)\,,
\end{equation}
where $U^T$ denotes the transpose of $U$, and the entries of $U$ are real, i.e., $U_{ij}\in\mathbb{R}$.
 We note that we have employed again the standard notation $R(U)=U$ for all the elements of the orthogonal group.
For this group, a basis for the $k$-th order commutant is given by a representation $F_d$ of the Brauer algebra $\mathfrak{B}_k(d)$ acting on the $k$-fold tensor product
Hilbert space, $\HC^{\otimes k}$. That is, 
\begin{equation}\label{eq:basis-comm-brauer}
    \SC^{(k)}(\mathbb{O}(d))=\{ F_d(\sigma) \}_{\sigma\in \mathfrak{B}_k(d)}.
\end{equation}

Here we recall that  the Brauer algebra is composed of  all possible pairings  on a set of $2k$ items. That is, given a set of $2k$ items, the elements of $\mathfrak{B}_k(d)$ correspond to  all possible ways of splitting them in pairs. Hence, the basis of the commutant, $\SC^{(k)}(\mathbb{O}(d))$, contains $\frac{(2k)!}{2^k k!}$ elements.  For the sake of illustration, the tensor representation of the Brauer algebra for $k=1,2,3$ is depicted in  Supp. Fig.~\ref{fig:elements_Bk}, and in Supp. Fig.~\ref{fig:explicit-comm-b3} we explicitly show that an element of $\mathfrak{B}_3$ commutes with $U^{\otimes 3}$ for any $U$ in $\mathbb{O}(d)$.

An element $\sigma\in\brauer$ can be completely specified by $k$ disjoint pairs, as $\sigma=\{\lambda_1, \sigma(\lambda_1)\}\cup\dots\cup\{\lambda_k, \sigma(\lambda_k)\} $. Moreover, we find it convenient to also define for any $\sigma$ its \textit{transpose} as $\sigma^T=\{\lambda_{1}+k, \sigma(\lambda_{1})+k\}\cup\dots\cup\{\lambda_{k}+k) , \sigma(\lambda_{k})+k\}$ where the sum is taken ${\rm mod} (k)$. Note that $\forall \sigma \in   \mathfrak{B}_k(d)$, then $\sigma^T\in \mathfrak{B}_k(d)$.  Let us now consider an explicit example and write down an element of $\mathfrak{B}_7(d)$. For instance, consider $\sigma=(\{1,2\} \cup \{8,9\}) (\{3,5\} \cup \{4,10\} \cup \{11,12\}) (\{6,14\} \cup \{7,13\})$, where the parenthesis correspond to cycles (as defined below in Supplemental Definition~\ref{def:permu_cycle-br}). We present this element in Supp. Fig.~\ref{fig:sigma}(a). In bra-ket notation, this is equivalent to
\beq  F_d(\sigma) = \sum_{i_1,i_2=0}^{d-1}\ket{i_1,i_1}\bra{i_2,i_2} \otimes \sum_{i_3,i_4,i_5=0}^{d-1}\ket{i_3, i_4, i_3}\bra{i_4, i_5, i_5} \otimes  \sum_{i_6,i_7=0}^{d-1}\ket{i_6,i_7}\bra{i_7,i_6} \,. \eeq
where $F_d(\sg)$ is a representation of the Brauer algebra element $\sigma$.  Here, we find $\sigma^T=(\{1,2\} \cup \{8,9\}) (\{3,11\} \cup \{4,5\} \cup \{10,12\}) (\{6,14\} \cup \{7,13\})$ (see Supp. Fig.~\ref{fig:sigma}(b)) and 
\begin{equation}
F_d(\sigma^T)=F_d(\sigma)^T\,.    
\end{equation}
More generally, given an element $\sigma\in\brauer$, we can express it as
\beq \label{eq:rep-b_k} F_d(\sigma) = \sum_{i_1,\dots,i_{2k}=0}^{d-1}\ket{i_{k+1},i_{k+2},\dots,i_{2k}} \bra{i_1,i_2,\dots,i_k} \prod_{\gamma=1}^{k} \delta_{i_{\lambda_\gamma}, i_{\sigma(\lambda_\gamma)}} \, .\eeq

\begin{figure}[t]
    \centering
\includegraphics[width=.6\columnwidth]{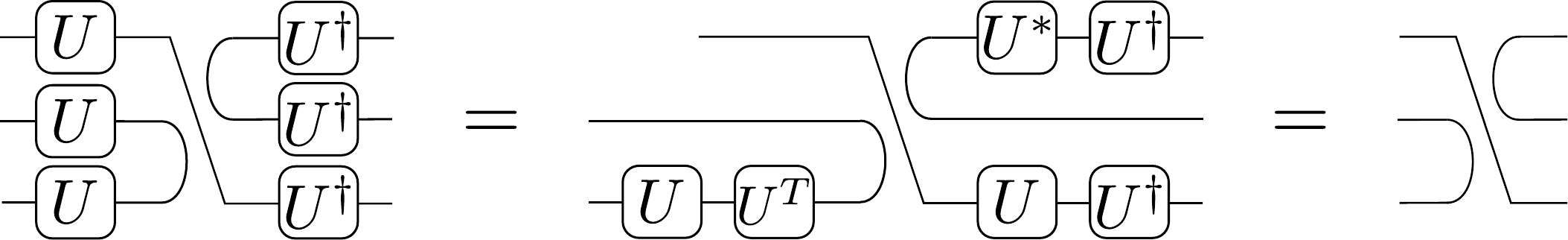}
    \caption{\textbf{Commutation relation.} 
    We explicitly show that an element of $B_3$ commutes with $U^{\otimes 3}$ for any $U\in\mathbb{O}(d)$. Here we use the ricochet property of Eq.~\eqref{eq:action_Pi} plus the fact that $UU\ad=UU^T=U^*U\ad=\id$.
    } 
    \refstepcounter{supfig}
    \label{fig:explicit-comm-b3}
\end{figure}

It is important to stress  that, in contrast to the $k$-th order commutant of the unitary group, $\SC^{(k)}(\mathbb{O}(d))$ is not a group itself but a $\mathbb{Z}(d)$-algebra. This implies that when we multiply two elements in $\brauer$, we do not necessarily obtain an element of $\brauer$ but rather an element of $\brauer$ times an integer power of $d$. Diagrammatically, this means that when we connect (multiply) two diagrams, closed loops can appear. Then, the power to which the factor $d$ is raised is equal to the number of closed loops. We illustrate this fact in Supp. Fig.~\ref{fig:sigma}(c). Furthermore, not every element in $\brauer$ has an inverse. We remark that the symmetric group $S_k$ is a subalgebra of the Brauer algebra. We will denote as $B_k = \brauer \textbackslash S_k$ the elements in $\brauer$ that do not belong to $S_k$, and recall that the elements of $B_k$ do not have an inverse.

First, let us consider the case of $k=1$. As shown in Supp. Fig.~\ref{fig:elements_Bk},  $\mathfrak{B}_1(d)$ contains a single element whose representation is given by $\{\id\}$. As such, we recover the same result as for the unitary group, where the Gram matrix is $
    A=\begin{pmatrix}
    d
    \end{pmatrix}$, and thus
\begin{align}\label{eq:twirl-o-k1}
    \TC^{(1)}_{\mathbb{O}(d)}[X]=\frac{\Tr[X]}{d}\id\,.
\end{align}

Next, we consider the case of $k=2$. Now $\mathfrak{B}_2(d)$ contains three elements (see Supp. Fig.~\ref{fig:elements_Bk}) whose representation is given by 
\begin{equation}\label{eq:k_symmetries_orthogonal_2}
    \SC^{( 2)}(\mathbb{O}(d))=\{\id\otimes \id,\SWAP,\Pi\}\,,
\end{equation}
where $\Pi$ was defined in Eq.~\eqref{eq:Pi}. 
The ensuing Gram Matrix is
\begin{equation}
    A=\begin{pmatrix}
    d^2 & d & d \\
    d & d^2 & d \\
    d & d &  d^2
    \end{pmatrix}\,,
\end{equation}
and the Weingarten matrix
\begin{equation}
    A^{-1}=\frac{1}{d(d+2)(d-1)}\begin{pmatrix}
    d+1 & -1 & -1\\
    -1 & d+1 & -1 \\
    -1 & -1 & d+1 
    \end{pmatrix}\,.
\end{equation}

\begin{figure}[t] 
    \centering
\includegraphics[width=1\columnwidth]{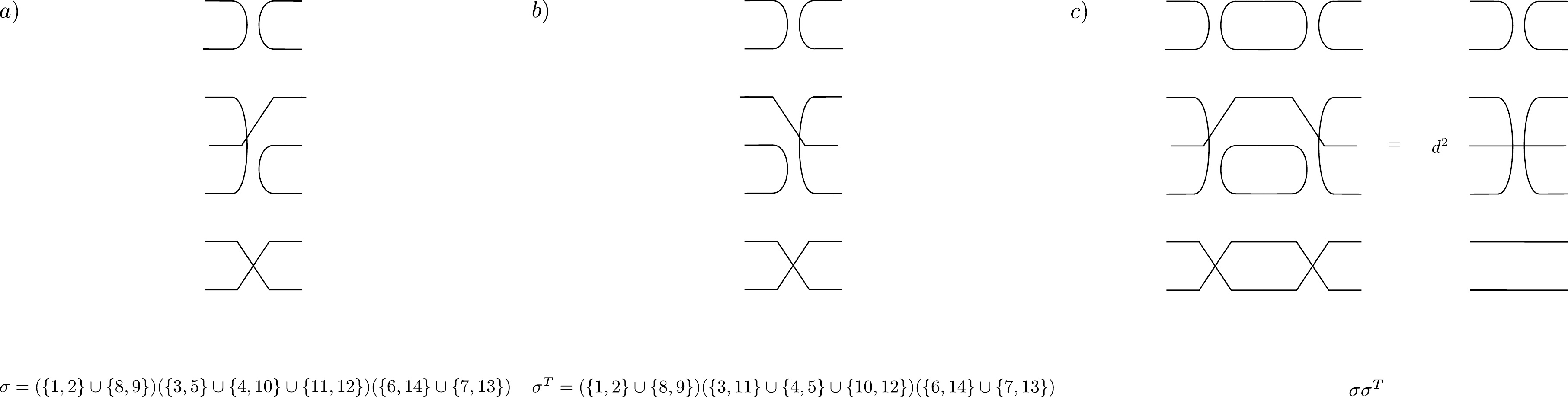}
    \caption{\textbf{Element and its transpose.} In panels a) and b) we respectively present two elements of $\mathfrak{B}_7$, $\sigma$, and its transpose $\sigma^T$. Then, in panel c) we present their composition $\sigma\sigma^T$. Here we can verify that $\mathfrak{B}_k$ is a 
    $\mathbb{Z}(d)$-algebra, as the multiplications of two elements of $\brauer$ is an element of $\brauer$ times an integer power of $d$.
    } 
    \label{fig:sigma}
\end{figure}

Thus, we find
\begin{align}
\begin{pmatrix}
    c_1(X)\\
    c_2(X)\\
    c_3(X)
    \end{pmatrix}&=\frac{1}{d(d+2)(d-1)}\begin{pmatrix}
    d+1 & -1 & -1\\
    -1 & d+1 & -1 \\
    -1 & -1 & d+1 
    \end{pmatrix}\cdot \begin{pmatrix}
    \Tr[X]\\
    \Tr[X\SWAP]\\
    \Tr[X \Pi]
    \end{pmatrix}\,.
\end{align}
Hence, 
\begin{align}\label{eq:twirl-o-k2}
    \TC^{(2)}_{\mathbb{O}(d)}[X]=&\frac{1}{d(d+2)(d-1)}\left((d+1)\Tr[X]  -\Tr[X\SWAP]  -\Tr[X \Pi]\right)\id\otimes \id\nonumber\\
    &+\frac{1}{d(d+2)(d-1)}\left(-\Tr[X] + (d+1)\Tr[X\SWAP]  -\Tr[X \Pi]\right) \SWAP\nonumber\\
    &+\frac{1}{d(d+2)(d-1)}\left(-\Tr[X]  -\Tr[X\SWAP] + (d+1) \Tr[X \Pi]\right)\Pi\,.
\end{align}

Similarly to the unitary case, building the Gram matrix for large  $k$ can be quite cumbersome.  However, since we are interested in the large-$d$ limit, we can use the following result. 
\begin{suptheorem}\label{theo-orthogonal}
Let $X$ be an operator in $\BC(\HC^{\otimes k})$, then for large Hilbert space dimension $d$, the twirl of $X$ over $\mathbb{O}(d)$, as defined in Eq.~\eqref{eq:twirl} is
\begin{align} \label{eq:app_twirl_orto}
    \TC^{(k)}_{\mathbb{O}(d)}[X]=\frac{1}{d^k}\sum_{\sigma\in \mathfrak{B}_k}\Tr[XF_d(\sigma)]F_d(\sigma^{T})+\frac{1}{d^k}\sum_{\sigma,\pi\in \mathfrak{B}_k}c_{\sigma,\pi}\Tr[XF_d(\sigma)]F_d(\pi)\,,
\end{align}
where the constants $c_{\sigma,\pi}$ are in $\OC(1/d)$.
\end{suptheorem}

In order to prove Supplemental Theorem~\ref{theo-orthogonal}, we recall the following definitions.

\begin{supdefinition}[Cycle]\label{def:permu_cycle-br} Let $\sigma$ be an element belonging to $\brauer$. A  cycle $c$ is a set of indices $\{i_m, \sigma(i_m), \overline{\sigma(i_m)}, \sigma(\overline{\sigma(i_m)}),\overline{\sigma(\overline{\sigma(i_m)})},  \sigma(\overline{\sigma(\overline{\sigma(i_m)})}), \dots \}$ that are closed under the action of $\sigma$. Here, we use the notation $\overline{i_m}$ to denote the opposite of $i_m$, i.e.,  $\overline{i_m} = i_m + k$ if $i_m \leq k$ and $\overline{i_m} = i_m - k$ if $i_m > k$. Moreover, we will refer to the number of indices in the cycle divided by two as the length of the cycle.
\end{supdefinition}

\begin{supdefinition}[Cycle Decomposition] \label{def:cycle-decomp-br}
Given an element $\sigma$ belonging to $\brauer$, its cycle decomposition is an expression of $\sigma$ as a product of disjoint cycles
\begin{equation}\label{eq:cycle-decomp-br}
\sigma = c_1\cdots c_r\,.
\end{equation}
\end{supdefinition}
We will  refer to  those indexes contained in length-one cycles (such that $\sigma(i_m)=\overline{i_m}$) as fixed points. Moreover, as we did for the unitary group, we assume that  Eq.~\eqref{eq:cycle-decomp-br} contains \textit{all} cycles, including those of length one (see Supp. Fig.~\ref{fig:tens_Sk}). We remark that Supplemental Definitions~\ref{def:permu_cycle-br} and~\ref{def:cycle-decomp-br} generalize and include as particular cases  Supplemental Definitions~\ref{def:permu_cycle} and~\ref{def:cycle-decom} respectively. As in the case of the elements of $S_k$, the cycle decomposition is unique, up to permutations of the cycles (since they are disjoint) and up to cyclic shifts within the cycles (since they are cycles). 

We will be interested in counting how many cycles, and of what length, are contained in each $\sigma\in \brauer$.  To that end, we introduce the definition of the cycle type.
\begin{supdefinition}[Cycle type]\label{def:cycle-type-br}
    Given a $\sigma \in \brauer$, its cycle type $\nu(\sigma)$ is a vector of length $k$ whose entries indicate how many cycles of each length are present in the cycle decomposition of $\sigma$. That is, 
\begin{equation}
    \nu(\sg) = (\nu_1,\cdots,\nu_k)\,,
\end{equation}
where $\nu_j$ denotes the number of length-$j$ cycles in $\sigma$. 
\end{supdefinition}

We now introduce the following propositions,

\begin{supproposition}\label{prop:characters-br}
The character of an element $\sigma\in \brauer$   is\begin{equation}
    \chi(\sg) = \Tr[ \sg ]  = d^{\norm{\nu(\sg)}_1}=d^r\,,
\end{equation}
where $\nu(\sg)$ is the cycle type  of $\sigma$ as defined in Supplemental Definition~\ref{def:cycle-type-br}, and $r$ is the number of cycles in the cycle decomposition of $\sigma$ as in Supplemental Definition~\ref{def:cycle-decomp-br}.
\end{supproposition}
\begin{proof}
Let us begin by re-writing Eq.~\eqref{eq:rep-b_k} in term of the cycles decomposition of $\sigma$ as in Supplemental Definition~\ref{def:cycle-decomp-br}. That is, given $\sigma = c_1\cdots c_r$, we write
\begin{equation}\label{eq:cycles-br}
F_d(\sigma)=\bigotimes_{\alpha=1}^r\left(\sum_{i_{\lambda^\alpha_1},\ldots,i_{\lambda^\alpha_{2|c_\alpha|}}=0}^{d-1}  |i_{\lambda^\alpha_{|c_\alpha|+1}},\dots,i_{\lambda^\alpha_{2|c_\alpha|}} \ra\la i_{\lambda^\alpha_1},\dots,i_{\lambda^\alpha_{|c_\alpha|}} |\prod_{\gamma=1}^{|c_\alpha|} \delta_{i_{\alpha_\gamma}, i_{\sigma(\alpha_\gamma)}} \right)\,,
\end{equation}
where $|c_\alpha|$ denotes the length of the $c_\alpha$ cycle and we used the notation $c_\alpha=\bigcup_{\gamma=1}^{|c_\alpha|}\{\alpha_\gamma,\sigma(\alpha_\gamma)\}$. Then, the character of $\sigma$ is 
\begin{equation}
\label{eq:almost-chi-br}
    \chi(\sg) = \Tr[ F_d(\sg) ]
    = \prod_{\alpha=1}^r\left(\sum_{i_{\lambda^\alpha_1},\ldots,i_{\lambda^\alpha_{2|c_\alpha|}}=0}^{d-1}  \Tr[|i_{\lambda^\alpha_{|c_\alpha|+1}},\dots,i_{\lambda^\alpha_{2|c_\alpha|}} \ra\la i_{\lambda^\alpha_1},\dots,i_{\lambda^\alpha_{|c_\alpha|}} |\prod_{\gamma=1}^{|c_\alpha|} \delta_{i_{\alpha_\gamma}, i_{\sigma(\alpha_\gamma)}}] \right)\,.
\end{equation}
We now compute
\beq 
\Tr[|i_{\lambda^\alpha_{|c_\alpha|+1}},\dots,i_{\lambda^\alpha_{2|c_\alpha|}} \ra\la i_{\lambda^\alpha_1},\dots,i_{\lambda^\alpha_{|c_\alpha|}} |\prod_{\gamma=1}^{|c_\alpha|} \delta_{i_{\alpha_\gamma}, i_{\sigma(\alpha_\gamma)}}] = \prod_{\gamma=1}^{|c_\alpha|} \delta_{i_{\alpha_\gamma}, i_{\sigma(\alpha_\gamma)}}  \prod_{\beta=1}^{|c_\alpha|}\delta_{i_{\lambda_\beta},i_{\lambda_{|c_\alpha|+\beta}}} \,,
\eeq
which implies that for any cycle $c_\alpha$, independently of its length, we have
\begin{equation}
    \sum_{i_{\lambda^\alpha_1},\ldots,i_{\lambda^\alpha_{2|c_\alpha|}}=0}^{d-1} \Tr[|i_{\lambda^\alpha_{|c_\alpha|+1}},\dots,i_{\lambda^\alpha_{2|c_\alpha|}} \ra\la i_{\lambda^\alpha_1},\dots,i_{\lambda^\alpha_{|c_\alpha|}} |\prod_{\gamma=1}^{|c_\alpha|} \delta_{i_{\alpha_\gamma}, i_{\sigma(\alpha_\gamma)}}]  =\sum_{i_{\lambda^\alpha_1},\ldots,i_{\lambda^\alpha_{2|c_\alpha|}}=0}^{d-1} \prod_{\gamma=1}^{|c_\alpha|} \delta_{i_{\alpha_\gamma}, i_{\sigma(\alpha_\gamma)}}  \prod_{\beta=1}^{|c_\alpha|}\delta_{i_{\lambda_\beta},i_{\lambda_{|c_\alpha|+\beta}}}=d\,,
\end{equation}
where we used the fact that a cycle either has no fixed points or is itself a fixed point.
Replacing in Eq.~\eqref{eq:almost-chi-br}, we obtain
\begin{equation}
    \chi(\sg) =\prod_{\alpha=1}^rd=d^r= d^{\norm{\nu(\sg)}_1}\,.
\end{equation}
\end{proof}

\begin{supproposition} \label{prop:loops-br}
    Let $\sg$ be an element of $\brauer$. The number of cycles $||\nu(\sg)||_1$ is at most $k-l$, where $l$ is the number of indices pairs $\left(i_{\lambda_\gamma},i_{\sg(\lambda_\gamma)}\right)$ such that $\lambda_\gamma, \sigma(\lambda_\gamma)\leq k$. Moreover, this maximum is uniquely achieved when for every pair $\lambda_\gamma, \sigma(\lambda_\gamma)\leq k$ there exists a pair $\lambda_{\gamma'}, \sigma(\lambda_{\gamma'}) > k$ such that $\lambda_\gamma= \overline{\lambda_{\gamma'}}$ and $\sigma(\lambda_\gamma) = \overline{\sg(\lambda_{\gamma'})}$, or $\sigma(\lambda_\gamma)=\overline{\lambda_{\gamma'}}$ and $\gamma=\overline{\sigma(\lambda_{\gamma'})}$, and the rest of indices are fixed points.
\end{supproposition}
\begin{proof}
    Using Supplemental Definition~\ref{def:permu_cycle-br}, we first have that every fixed point is a cycle. Then, if for a pair $\lambda_\gamma, \sigma(\lambda_\gamma)\leq k$ there exists a pair $\lambda_\gamma', \sigma(\lambda_\gamma') > k$ such that $\lambda_\gamma= \overline{\lambda_{\gamma'}}$ and $\sigma(\lambda_\gamma) = \overline{\sg(\lambda_{\gamma'})}$, or $\sigma(\lambda_\gamma)=\overline{\lambda_{\gamma'}}$ and $\lambda_\gamma=\overline{\sigma(\lambda_{\gamma'})}$, then those two pairs form a cycle, as the sequences

    \beq \label{eq:cycle-bell-1} \lambda_\gamma,\; \sigma(\lambda_\gamma) = \overline{\sigma(\lambda_{\gamma'})}, \;\sigma(\lambda_{\gamma'}),\;\sg(\sigma(\lambda_{\gamma'}))= \lambda_{\gamma'},\;\lambda_\gamma\,,\eeq
    or 
    \beq \label{eq:cycle-bell-2}  \lambda_\gamma,\; \sigma(\lambda_\gamma) = \overline{\lambda_{\gamma'}},\; \lambda_{\gamma'},\; \sigma(\lambda_{\gamma'})=\overline{\lambda_\gamma}, \; \lambda_\gamma \,, \eeq
    are closed under $\sg$. Therefore, there are $l$ such cycles plus $k-2l$ fixed points, which add up to $k-l$ cycles.

    Then, we note that given pair of indices such that $\lambda_\gamma, \sigma(\lambda_\gamma)\leq k$, a cycle containing them must consist of at least four indices. This is so because $\overline{\sigma(\lambda_\gamma)}\neq \lambda_\gamma, \sigma(\lambda_\gamma)$. By direct inspection, it is clear that the only possible sequences that lead to cycles of four indices are~\eqref{eq:cycle-bell-1} or~\eqref{eq:cycle-bell-2}. Therefore,  
    if the conditions stated above are not satisfied, $||\nu(\sigma)||_1<k-l$.
    
\end{proof}

\begin{supproposition} \label{prop:trace-prod-br}
    Let $\sigma$ be an element of $B_k$ and $\pi$ an element of $\brauer$. Then, it holds that
\begin{equation}
 \Tr[\sigma \pi] \begin{cases}
 =d^k \qquad\text{if $ \pi= \sg^T$} \\ 
 \leq d^{k-1} \quad \text{else}
 \end{cases}\,,
\end{equation}
\end{supproposition}
\begin{proof}
Let us start with $\pi=\sigma^T$. We have
\beq \begin{split} \sigma \pi &=  \sum_{\substack{i_1,\dots,i_{2k}=0\\ i'_1,\dots,i'_{2k}=0}}^{d-1}   \ket{i_{k+1},i_{k+2},\dots,i_{2k}}  \bra{i'_1,i'_2,\dots,i'_k}   \bra{i_1,i_2,\dots,i_k}  \ket{i'_{k+1},i'_{k+2},\dots,i'_{2k}} \prod_{\gamma=1}^{k} \delta_{i_{\lambda_\gamma}, i_{\sigma\left(\lambda_\gamma\right)}} \prod_{\beta=1}^{k} \delta_{i'_{\lambda_{\beta}} , i'_{\pi(\lambda_{\beta})}} \\ &= \sum_{\substack{i_1,\dots,i_{2k}=0\\ i'_1,\dots,i'_{k}=0}}^{d-1}   \ket{i_{k+1},i_{k+2},\dots,i_{2k}}  \bra{i'_1,i'_2,\dots,i'_k}   \bra{i_1,i_2,\dots,i_k}  \ket{i_{1},i_{2},\dots,i_{k}} \prod_{\gamma=1}^{k} \delta_{i_{\lambda_\gamma}, i_{\sigma\left(\lambda_\gamma\right)}} \prod_{\beta=1}^{k} \delta_{i'_{\lambda_{\beta}} , i'_{\pi(\lambda_{\beta})}} \\ &= \sum_{\substack{i_{k+1},\dots,i_{2k}=0\\ i'_1,\dots,i'_{k}=0}}^{d-1}   \ket{i_{k+1},i_{k+2},\dots,i_{2k}}  \bra{i'_1,i'_2,\dots,i'_k} \sum_{\substack{i_{\lambda_{\gamma_1}},\dots,i_{\lambda_{\gamma_{2l}}}=0 \\ \lambda_\gamma,\sg(\lambda_\gamma) \leq k}}^{d-1} \bra{i_{\lambda_{\gamma_1}},\dots,i_{\lambda_{\gamma_{2l}}}}  \ket{i_{\lambda_{\gamma_1}},\dots,i_{\lambda_{\gamma_{2l}}}} \prod_{\substack{ \gamma \\ \lambda_\gamma,\sg(\lambda_\gamma) \leq k }} \delta_{i_{\lambda_\gamma}, i_{\sg\left(\lambda_\gamma\right)}} \\ & \qquad \otimes\sum_{\substack{i_{\lambda_{\gamma_{2l+1}}},\dots,i_{\lambda_{\gamma_{k-2l}}}=0 \\ \sg(\lambda_\gamma) > k}}^{d-1}  \bra{i_{\lambda_{2l+1}},\dots i_{\lambda_{k-2l}}}  \ket{i_{\lambda_{2l+1}},\dots i_{\lambda_{k-2l}}}  \prod_{\substack{ \gamma \\ \sigma(\lambda_\gamma) > k }} \delta_{i_{\lambda_\gamma}, i_{\sg\left(\lambda_\gamma\right)}} \prod_{\substack{\beta \\ \lambda_\beta \leq k}} \delta_{i'_{\lambda_{\beta}} , i'_{\pi(\lambda_{\beta})}}  \\ & = d^l \sum_{\substack{i_{k+1},\dots,i_{2k}=0\\ i'_1,\dots,i'_{k}=0}}^{d-1} \ket{i_{k+1},i_{k+2},\dots,i_{2k}}  \bra{i'_1,i'_2,\dots,i'_k}  \prod_{\substack{ \gamma \\ \lambda_\gamma, \sigma(\lambda_\gamma) > k }} \delta_{i_{\lambda_\gamma}, i_{\sg\left(\lambda_\gamma\right)}} \prod_{\substack{\beta \\ \lambda_\beta, \pi(\lambda_\beta) \leq k}} \delta_{i'_{\lambda_{\beta}} , i'_{\pi(\lambda_{\beta})}} \prod_{\substack{\omega\\ \lambda_\omega \leq k \\ \pi(\omega)> k }} \delta_{i_{\lambda_{\omega+k}},i'_{\lambda_\omega}} \\ & = d^l \tau \,,\end{split} \eeq
where $\tau\in B_k$ and
$l$ is the number of indices pairs such that $\lambda_\gamma,\sigma(\lambda_\gamma)\leq k$. That is, $l$ is the number of closed loops that appear when multiplying $\pi$ and $\sigma$.
Therefore, using Supplemental Propositions~\ref{prop:characters-br} and~\ref{prop:loops-br},
\beq \Tr[\sigma\pi] =  d^l d^{||\nu(\sigma)||_1} = d^l d^{k-l} = d^k \,. \eeq

Now, let us assume that $\pi\neq \sigma^T$. Let us further define $C_k(\sg,\pi) $ as the set of indices $\{\gamma, \gamma'\}$  such that $i_\gamma, \sigma(\gamma) \leq k$, $i'_{\gamma'}, \pi(\gamma') > k$, and such that for every $\gamma\in C_k(\sigma,\pi)$ there exists $\gamma'\in C_k(\sigma,\pi)$  leading to $\gamma=\gamma'$ or $\gamma=\pi(\gamma')$ and $\sigma(\gamma)= \gamma'$ or $\sigma(\gamma)= \pi(\gamma')$. Accordingly, we introduce
\beq \tau=\sum_{i_\gamma, i'_{\gamma'}=0}^{d-1}   \bigotimes_{\gamma,\gamma'\in C_k(\sg,\pi)} \ket{i_\gamma}  \bra{i'_{\gamma'}} \prod_{\gamma\in C_k(\sg,\pi)} \delta_{i_\gamma, i_{\sigma\left(\gamma\right)}} \prod_{\gamma'\in C_k(\sg,\pi)} \delta_{i'_{\gamma'} , i'_{\pi(\gamma')}} \,, \eeq
and write its cycle decomposition as $\tau=c_1\cdots c_r$ (since $\tau\in\mathfrak{B}_s$ where $s=\sum_\alpha |c_\alpha|$). We then find
\beq \begin{split}
\sigma\pi &=  \sum_{\substack{i_1,\dots,i_{2k}=0\\ i'_1,\dots,i'_{2k}=0}}^{d-1}   \ket{i_{k+1},i_{k+2},\dots,i_{2k}}  \bra{i'_1,i'_2,\dots,i'_k}   \bra{i_1,i_2,\dots,i_k}  \ket{i'_{k+1},i'_{k+2},\dots,i'_{2k}} \prod_{\gamma=1}^{k} \delta_{i_{\lambda_\gamma}, i_{\sigma\left(\lambda_\gamma\right)}} \prod_{\beta=1}^{k} \delta_{i'_{\lambda_{\beta}} , i'_{\pi(\lambda_{\beta})}} \\ &= \sum_{\substack{i_1,\dots,i_{2k}=0\\ i'_1,\dots,i'_{k}=0}}^{d-1}   \ket{i_{k+1},i_{k+2},\dots,i_{2k}}  \bra{i'_1,i'_2,\dots,i'_k}   \bra{i_1,i_2,\dots,i_k}  \ket{i_{1},i_{2},\dots,i_{k}} \prod_{\gamma=1}^{k} \delta_{i_{\lambda_\gamma}, i_{\sigma\left(\lambda_\gamma\right)}} \prod_{\beta=1}^{k} \delta_{i'_{\lambda_{\beta}} , i'_{\pi(\lambda_{\beta})}} \\ &= \sum_{\substack{i_{k+1},\dots,i_{2k}=0\\ i'_1,\dots,i'_{k}=0}}^{d-1}   \ket{i_{k+1},i_{k+2},\dots,i_{2k}}  \bra{i'_1,i'_2,\dots,i'_k} \bigotimes_{\alpha=1}^r \sum_{i_{\lambda^\alpha_1},\dots,i_{\lambda^\alpha_{|c_\alpha|}}=0}^{d-1} \bra{i_{\lambda^\alpha_1},\dots,i_{\lambda^\alpha_{|c_\alpha|}}} \ket{i_{\lambda^\alpha_1},\dots,i_{\lambda^\alpha_{|c_\alpha|}}} \prod_{ \gamma=1}^{|c_{\alpha}|} \delta_{i_{\lambda^\alpha_\gamma}, i_{\sg\left(\lambda^\alpha_\gamma\right)}} \\ & \qquad \otimes\sum_{\gamma,\gamma'=0}^{d-1} \bigotimes_{\gamma,\gamma'\not\in C_k} \bra{i_\gamma}  \ket{i'_{\gamma'}}  \prod_{\gamma\not\in C_k(\sg,\pi)} \delta_{i_\gamma, i_{\sigma\left(\gamma\right)}} \prod_{\gamma'\not\in C_k(\sg,\pi)} \delta_{i'_{\gamma'} , i'_{\pi(\gamma')}} \prod_{\substack{ \gamma \\ \lambda_\gamma, \sigma(\lambda_\gamma) > k }} \delta_{i_{\lambda_\gamma}, i_{\sg\left(\lambda_\gamma\right)}} \prod_{\substack{\beta \\ \lambda_\beta, \pi(\lambda_\beta) \leq k}} \delta_{i'_{\lambda_{\beta}} , i'_{\pi(\lambda_{\beta})}}  \\ &= d^r \xi \,,\end{split} \eeq
where $\xi\in\brauer$. Thus, using Supplemental Proposition~\ref{prop:characters-br} we find $\Tr[\sg\pi] = d^r d^{||\nu(\xi)||_1}$. Finally, using Supplemental Proposition~\ref{prop:loops-br} it follows that when $\pi\neq\sg^T$ either $r<l$ or $||\nu(\xi)||_1< k-l$, so that
\beq \Tr[\sigma\pi]\leq d^{k-1}\,.\eeq

\end{proof}

Let us now go back to computing the twirl $\TC^{(k)}_{\mathbb{O}}[X]$. First, analogously to what we did for the case of the unitary group, we reorder the basis $\SC^{(k)}(\mathbb{O}(d))$ in such a way that the first element is $F_d(e)$, followed by the elements $F_d(\sigma)$ that fulfill $F_d(\sigma)=F_d(\sigma)^T$, and finally we order the rest of the elements $\sigma\neq\sigma^T$ by placing $F_d(\sigma)$ next to $F_d(\sigma)^{T}$. Here, we remark that if $\sg\in S_k$, $F_d(\sg)^T=F_d(\sg^T)=F_d(\sg^{-1})$. We also recall again that the elements such that $\sigma=\sigma^{-1}$ are known as involutions and must  consist of a product of disjoint transpositions plus fixed points. The number of involutions is given by $I_k$ (defined above). More generally, we have that for an element $\sg\in\brauer$ to satisfy that $F_d(\sigma)=F_d(\sigma)^T$ it must consist of a product of length-two cycles and fixed points. The number of such elements is $\mathfrak{I}_k=\sum_{\eta=0}^{\lfloor \frac{k}{2}\rfloor}2^\eta\binom{k}{2\eta}(2\eta-1)!!$.

Then, the following result holds.
\begin{supproposition}
    The $A$ matrix, of dimension $\frac{(2k)!}{2^k k!}\times \frac{(2k)!}{2^k k!}$, can be expressed as
    \begin{equation}
        A= d^k(\widetilde{A}+\frac{1}{d}B)\,.
    \end{equation}
    Here we defined
    \begin{equation}
\widetilde{A}=\id_{\mathfrak{I}_k}\bigoplus_{j=1}^{\left(\frac{(2k)!}{2^{k+1} k!}-\frac{\mathfrak{I}_k}{2}\right)}\begin{pmatrix}
            0 && 1 \\ 1 && 0
        \end{pmatrix}\,,
    \end{equation}
where $\id_{\mathfrak{I}_k}$ denotes the $\mathfrak{I}_k\times \mathfrak{I}_k$ dimensional identity. Moreover, the matrix $B$ is such that its entries are  $\OC(1)$.
\end{supproposition}

More visually, the matrix $\widetilde{A}$ is of the form
\begin{equation}
\widetilde{A} = \begin{pmatrix}
\, \underbrace{\boxed{\begin{matrix} 1 & 0 & \cdots &  0 \\ 
 0& 1 & \cdots & 0 \\
\vdots & \vdots & \ddots & \vdots \\
0 & \cdots & 0 & 1 \end{matrix}}}_{\mathfrak{I}_k\times \mathfrak{I}_k} &  &   & {\makebox(0,0){\text{\huge0}}} &   \\ 
 &  \begin{matrix} 0 & 1  \\
1 & 0 \end{matrix} &     &   \\ 
 &  &  \ddots &  &   \\ 
  {\makebox(0,0){\text{\huge0}}}&    &  &  \begin{matrix} 0 & 1  \\
1 & 0 \end{matrix} \\ 
\end{pmatrix}\,.\nonumber
\end{equation}
And we remark the fact that $\widetilde{A}$  is its own inverse. That is, $\widetilde{A}^{-1}=\widetilde{A}$.

\begin{proof}
Let us recall that the entries of the matrix $A$ are of the form $A_{\nu\mu}=\Tr[P_\nu P_\mu]$ where $P_\nu,P_\mu\in \SC^{(k)}(\mathbb{O}(d))$. From Supplemental Proposition~\ref{prop:trace-prod-br} it follows that 
\begin{equation}
A = \begin{pmatrix}
\, \begin{matrix} d^k & a_{1,2} & \cdots &  a_{1,\mathfrak{I}_k} &  a_{1,\mathfrak{I}_{k+1}} & \cdots \\ 
 a_{2,1}& d^k & \cdots & a_{2,\mathfrak{I}_k} \\
\vdots & \vdots & \ddots & \vdots \\
a_{\mathfrak{I}_k,1} & \cdots & a_{\mathfrak{I}_k,\mathfrak{I}_k-1}  & d^k \\
a_{\mathfrak{I}_{k+1},1}  &  &   &  & a_{\mathfrak{I}_k+1,\mathfrak{I}_k+1} & d^k \\
\vdots &  &   &  &  d^k & a_{\mathfrak{I}_k+2,\mathfrak{I}_k+2} \\
&  &   &  &  & & \ddots \\ 
&  &   &  &  & &  & a_{\frac{(2k)!}{2^k k!}-1,\frac{(2k)!}{2^k k!}-1} & d^k \\
&  &   &  &  & &  & d^k & a_{\frac{(2k)!}{2^k k!},\frac{(2k)!}{2^k k!}}

\end{matrix}
 \\ 
\end{pmatrix}\,,\nonumber
\end{equation}
where the matrix elements $a_{ij}\leq d^{k-1}$. This allows us to express the matrix $A$ as
\begin{equation}
    A= d^k(\widetilde{A}+\frac{1}{d}B)\,,
\end{equation}
where the entries in $B$ are at most equal to 1. 
    
\end{proof}

Using Supplemental Lemma~\ref{lem:inv},  setting $M=\widetilde{A}$, $\Omega=\frac{1}{d}B$ and noting that $A\propto (\widetilde{A}+\frac{1}{d}B)$ always has inverse~\cite{collins2006integration,puchala2017symbolic} we find
\begin{equation}
    A^{-1}=\frac{1}{d^k}\left(\widetilde{A}-\frac{1}{d}(\id+\frac{1}{d}\widetilde{A}B)^{-1}\widetilde{A}B \widetilde{A}\right)    =\frac{1}{d^k}\left(\widetilde{A}-C\right)\,,
\end{equation}
where we have defined
\begin{equation}\label{eq:C-matrix-br}
    C=\frac{1}{d}(\id+\frac{1}{d}\widetilde{A}B)^{-1}\widetilde{A}B \widetilde{A}\,.
\end{equation}
It is easy to verify that the matrix entries of $C$ are in $\OC(1/d)$, following an analogous argument as that in Supplemental Lemma~\ref{lem:inv}. 
Combining the previous result with Eqs.~\eqref{eq:twirled_X_comm} and~\eqref{eq:inverse-vec-c} leads to 
\begin{align}
    \TC^{(k)}_{\mathbb{O}}[X]=\frac{1}{d^k}\sum_{\sigma\in \brauer}\Tr[XF_d(\sigma)]F_d(\sigma^{T})+\frac{1}{d^k}\sum_{\sigma,\pi\in \brauer}c_{\sigma,\pi}\Tr[XF_d(\sigma)]F_d(\pi)\,,
\end{align}
where the $c_{\sigma,\pi}$ are the matrix entries of $C$ as defined in Eq.~\eqref{eq:C-matrix-br}. This is then precisely the statement of Supplemental Theorem~\ref{theo-orthogonal}.

\subsection{Computing expectation values of twirled operators}

Let us now consider an expectation value of the form
\begin{equation}
    C(\rho_i)=\Tr[U\rho_i U\ad O]\,,
\end{equation}
where $\rho_i$ is a  quantum state and $O$ is a traceless quantum operator such that $O^2=\id$. 
Next, let us consider the task of estimating expectation values of the form
\begin{align}
\mathbb{E}_{\mathbb{O}(d)}\left[\prod_{i=1}^k C(\rho_i)\right]&=\mathbb{E}_{\mathbb{O}(d)}\left[\prod_{i=1}^k \Tr[U\rho_{i}U\ad O]\right]\,.
\end{align}

Here, we will show that in the large-$d$ limit, the following theorem holds.
\begin{suptheorem}\label{theo-moments-orto}
Let $\rho_i$ for $i=1,\ldots,k$ be a multiset of real-valued quantum states such that $\Tr[\rho_i\rho_{i'}]\in\Omega\left(\frac{1}{\poly(\log(d))}\right)$ for all $i,i'$,  and let  $O$ be some real-valued traceless Hermitian operator such that $O^2=\id$. Then let us define $\mathfrak{T}_k\subseteq \brauer$ the set  of all possible $k/2$ disjoint cycles of length two. That is, for any $\sigma\in \mathfrak{T}_k$, its cycle decomposition is $\sigma=c_1\cdots c_{k/2}$ where $c_\alpha$ is a length-two cycle for all $\alpha=1,\ldots,k/2$. Then, in the large-$d$ limit we have
\begin{align}
\mathbb{E}_{\mathbb{O}(d)}\left[\prod_{i=1}^k C(\rho_i)\right]&=\frac{1}{d^{k/2}}\sum_{\sigma\in \mathfrak{T}_k}\prod_{\alpha=1}^{k/2}\Tr[\rho_{c_\alpha(1)}\rho_{c_\alpha(2)}]\,.
\end{align}
\end{suptheorem}

To prove this theorem, let us first re-write
\begin{equation}
\mathbb{E}_{\mathbb{O}(d)}\left[\prod_{i=1}^k C(\rho_i)\right]=\mathbb{E}_{\mathbb{O}(d)}\left[\prod_{i=1}^k \Tr[U\rho_{i}U\ad O]\right]=\mathbb{E}_{\mathbb{O}(d)}\left[\Tr[U^{\otimes k}\Lambda \,(U\ad)^{\otimes k} O^{\otimes k}]\right]\,,
\end{equation}
where $\Lambda=\rho_{1}\otimes\cdots\otimes \rho_{k} $. Explicitly, 
\begin{align}
\mathbb{E}_{\mathbb{O}(d)}\left[\prod_{i=1}^k C(\rho_i))\right]&=\int_{\mathbb{O}(d)}d\mu(U) \Tr[U^{\otimes k}\Lambda\,(U\ad)^{\otimes k} O^{\otimes k}]]\nonumber\\
&=\Tr\left[\left(\int_{\mathbb{O}(d)}d\mu(U) U^{\otimes k}\Lambda\,(U\ad)^{\otimes k}\right) O^{\otimes k}\right]\,.
\end{align}
Using Supplemental Theorem~\ref{theo-orthogonal} we readily find 
\begin{align}\label{eq:moment-twirl-br}
\mathbb{E}_{\mathbb{O}(d)}\left[\prod_{i=1}^k C(\rho_i)\right]&=\frac{1}{d^k}\sum_{\sigma\in \brauer}\Tr[\Lambda F_d(\sigma)]\Tr[F_d(\sigma^T)O^{\otimes k} ] +\frac{1}{d^k}\sum_{\sigma,\pi\in \brauer}c_{\sigma,\pi}\Tr[\Lambda F_d(\sigma)]\Tr[F_d(\pi)O^{\otimes k}]\,.
\end{align}

Now, let us prove the following proposition.
\begin{supproposition}\label{prop:powers-O-br}
    Let $O$ be a traceless real-valued Hermitian operator such that $O^2=\id$. Then we have
    $\Tr[F_d(\sigma)O^{\otimes k} ]=0$ for any  $\sigma\in \brauer$ if $k$ is odd, and $\Tr[F_d(\sigma)O^{\otimes k} ]=d^{r}$ if  $k$ is even and $\sigma$ is a product of $r$ disjoint cycles of even length. The maximum of $\Tr[F_d(\sigma)O^{\otimes k} ]$ is therefore achieved when $\sigma$ is a product of $k/2$ disjoint cycles of length two, leading to $\Tr[F_d(\sigma)O^{\otimes k} ]=d^{k/2}$.
\end{supproposition}

\begin{proof}
We first consider the case of $k$ being odd. Let us express $\sigma$ in its cycle decomposition $
\sigma = c_1\cdots c_r$ as in Definition~\ref{def:cycle-decomp-br}. 
We have that 
\begin{equation}\label{eq:cycles-odd-br}
\Tr[F_d(\sigma)O^{\otimes k} ]=\prod_{\alpha=1}^r    \Tr[F_d(c_{\alpha})O^{\otimes |c_\alpha| }]\,.
\end{equation}
Because $k$ is odd, we know that there must exist at least one cycle acting on an odd number of subsystems in the right hand side of Eq.~\eqref{eq:cycles-odd-br}. Let us assume that this occurs for the cycle $c_{\alpha'}$. Then, we will have 
\begin{equation}
    \Tr[F_d(c_{\alpha'})O^{\otimes |c_{\alpha'}|} ]=\Tr[O^{|c_{\alpha'}|}]=\Tr[O]=0\,.
\end{equation}
Here we have used the fact that $O^2=\id$, $O^T=O$, and hence, since $|c_{\alpha'}|$ is odd, we have $O^{|c_{\alpha'}|}=O$.

Next, let us consider the case of $k$ being even. We know from Eq.~\eqref{eq:cycles-odd-br} that if $\sigma$ contains any cycle acting on an odd number of subsystems, then $\Tr[F_d(\sigma)O^{\otimes k} ]$ will be equal to zero. This means that only the elements $\sigma$ composed entirely of cycles acting on an even number of subsystems will have non-vanishing trace. If this is the case, we will have
\begin{equation}\label{eq:cycles-even-br}
\Tr[F_d(\sigma)O^{\otimes k} ]=\prod_{\alpha=1}^r    \Tr[F_d(c_{\alpha})O^{\otimes |c_\alpha|} ]=\prod_{\alpha=1}^rd=d^r\,.
\end{equation}
This follows from the fact that if $|c_\alpha|$ is even, then $O^{|c_\alpha|}=\id$.
Moreover, Eq.~\eqref{eq:cycles-even-br} will be maximized for the case when $r$ is largest, which corresponds to the case when $\sigma$ is a product of $k/2$ disjoint length-two cycles.
For this special case one finds 
\begin{equation}
\Tr[F_d(\sigma)O^{\otimes k} ]=d^{k/2}\,.
\end{equation}

\end{proof}

Next, let us prove the following proposition.
\begin{supproposition}\label{prop:powers-rho-br}
    Let $\Lambda =\rho_{1}\otimes\cdots\otimes \rho_{k} $ be a tensor product of $k$ real-valued quantum states. Then $\left|\Tr[\Lambda F_d(\sigma)]+\Tr[\Lambda F_d(\sigma^T)]\right|\leq 2$ for all $\sigma\in \brauer$.
\end{supproposition}

\begin{proof}

 Let us  decompose $\sigma$ in its cycle decomposition $
\sigma = c_1\cdots c_r$ as in Supplemental Definition~\ref{def:cycle-decomp-br}. First, let us assume that $\rho_1,\dots,\rho_k$ are pure. Then, we will have
\begin{equation} \label{eq:app_states_contribution-br}
\Tr[\Lambda P_d(\sigma)]=\prod_{\alpha=1}^r\Tr[\rho_{{c_\alpha(1)}}\cdots \rho_{{c_\alpha(|c_\alpha|)}}]= \prod_{\alpha=1}^r \bra{\psi_{c_\alpha(1)}}\ket{\psi_{c_\alpha(2)}} \cdots \bra{\psi_{c_\alpha(|c_\alpha|-1)}}\ket{\psi_{c_\alpha(|c_\alpha|)}} \bra{\psi_{c_\alpha(|c_\alpha|)}}\ket{\psi_{c_\alpha(1)}}  \,.
\end{equation}
where in the first equality we have used the fact that since all the states are real-valued, then $\rho^T_i=\rho_i$. 
Using that $|\bra{\psi_i}\ket{\psi_j}|\leq 1$ $\forall i,j$, we find that $|\Tr[\Lambda P_d(\sigma)]|\leq 1$ for all $\sigma\in \brauer$. Next, let us assume that $\rho_1,\dots,\rho_k$ are general arbitrary (real-valued) quantum states, i.e., they are not necessarily pure. Using that any quantum state can be written as a convex combination of orthonormal pure states, i.e., $\rho_i=\sum_{k_i} \lambda_{k_i} \ketbra{\psi_{k_i}}$ with $\lambda_{k_i}$ being real non-negative numbers such that $\sum_{k_i} \lambda_{k_i}=1$, we obtain
\begin{align} 
\Tr[\Lambda P_d(\sigma)]&=\prod_{\alpha=1}^r\Tr[\rho_{{c_\alpha(1)}}\cdots \rho_{{c_\alpha(|c_\alpha|)}}] \nonumber\\&=  \prod_{\alpha=1}^r \sum_{k_{c_\alpha(1)},\dots,k_{c_\alpha(|c_\alpha|)} } \lambda_{k_{c_\alpha(1)}} \cdots \lambda_{k_{c_\alpha(|c_\alpha|)}}\bra{\psi_{k_{c_\alpha(1)}}}\ket{\psi_{k_{c_\alpha(2)}}} \cdots \bra{\psi_{k_{c_\alpha(|c_\alpha|-1)}}}\ket{\psi_{k_{c_\alpha(|c_\alpha|)}}} \bra{\psi_{k_{c_\alpha(|c_\alpha|)}}}\ket{\psi_{k_{c_\alpha(1)}}}\,,
\end{align}
which leads to 
\begin{align} 
    |\Tr[\Lambda P_d(\sigma)]| &= \prod_{\alpha=1}^r \left|\sum_{k_{c_\alpha(1)},\dots,k_{c_\alpha(|c_\alpha|)} } \lambda_{k_{c_\alpha(1)}} \cdots \lambda_{k_{c_\alpha(|c_\alpha|)}}\bra{\psi_{k_{c_\alpha(1)}}}\ket{\psi_{k_{c_\alpha(2)}}} \cdots \bra{\psi_{k_{c_\alpha(|c_\alpha|-1)}}}\ket{\psi_{k_{c_\alpha(|c_\alpha|)}}} \bra{\psi_{k_{c_\alpha(|c_\alpha|)}}}\ket{\psi_{k_{c_\alpha(1)}}} \right|\nonumber \\ &\leq \nonumber \prod_{\alpha=1}^r \sum_{k_{c_\alpha(1)},\dots,k_{c_\alpha(|c_\alpha|)} } \lambda_{k_{c_\alpha(1)}} \cdots \lambda_{k_{c_\alpha(|c_\alpha|)}} \left|\bra{\psi_{k_{c_\alpha(1)}}}\ket{\psi_{k_{c_\alpha(2)}}} \cdots \bra{\psi_{k_{c_\alpha(|c_\alpha|-1)}}}\ket{\psi_{k_{c_\alpha(|c_\alpha|)}}} \bra{\psi_{k_{c_\alpha(|c_\alpha|)}}}\ket{\psi_{k_{c_\alpha(1)}}}\right|\\ & \leq \prod_{\alpha=1}^r \sum_{k_{c_\alpha(1)},\dots,k_{c_\alpha(|c_\alpha|)} } \lambda_{k_{c_\alpha(1)}} \cdots \lambda_{k_{c_\alpha(|c_\alpha|)}} = 1 \,,
\end{align}
where  we used the triangle inequality in the first inequality, and the fact that $|\bra{\psi_i}\ket{\psi_j}|\leq 1$ $\forall i,j$ in the second inequality.
Hence, using again the triangle inequality, we arrive at
\begin{equation}
\left|\Tr[\Lambda P_d(\sigma)]+\Tr[\Lambda P_d(\sigma^{-1})]\right|\leq \left|\Tr[\Lambda P_d(\sigma)]\right|+\left|\Tr[\Lambda P_d(\sigma^{-1})]\right|\leq 2\,,    
\end{equation}
for any $\sigma\in \brauer$.

 \end{proof}

With Supplemental Propositions~\ref{prop:powers-O-br} and~\ref{prop:powers-rho-br} we can now state the following result.

\begin{supproposition}\label{prop:off-diag-br}
     Let $O$ be a traceless real-valued Hermitian operator such that $O^2=\id$. Let $\Lambda=\rho_{1}\otimes\cdots\otimes \rho_{k} $ be a tensor product of $k$ real-valued quantum states. Then, for all $\pi$ and $\sigma$ in $\brauer$
     \begin{equation}
       \frac{1}{d^k}  \left|(c_{\sigma,\pi}\Tr[\Lambda F_d(\sigma)]+c_{\sigma^T,\pi}\Tr[\Lambda P_d(\sigma^T)])\Tr[F_d(\pi)O^{\otimes k}]\right|\in\OC\left(\frac{1}{d^{\frac{k+2}{2}}}\right)\,.
     \end{equation}
\end{supproposition}

\begin{proof}
We begin by assuming, without loss of generality, that $|c_{\sigma^T,\pi}|\leq |c_{\sigma,\pi}|$. Thus, we have
\begin{equation}
   \frac{1}{d^k} \left|(c_{\sigma,\pi}\Tr[\Lambda F_d(\sigma)]+c_{\sigma^T,\pi}\Tr[\Lambda F_d(\sigma^T)])\Tr[F_d(\pi)O^{\otimes k}]\right|\leq \frac{|c_{\sigma,\pi}|}{d^k}\left|(\Tr[\Lambda F_d(\sigma)]+\Tr[\Lambda F_d(\sigma^T)])\Tr[F_d(\pi)O^{\otimes k}]\right|\nonumber \,.
\end{equation}
Then, from Supplemental Propositions~\ref{prop:powers-O-br} and~\ref{prop:powers-rho-br} we find
\begin{equation}
   \frac{|c_{\sigma,\pi}|}{d^k}\left|(\Tr[\Lambda F_d(\sigma)]+\Tr[\Lambda F_d(\sigma^T)])\Tr[F_d(\pi)O^{\otimes k}]\right|\leq \frac{|c_{\sigma,\pi}|}{d^k} 2 d^{k/2}=\frac{2|c_{\sigma,\pi}|}{d^{k/2}}\,.
\end{equation}
Finally, since from Supplemental Theorem~\ref{theo-orthogonal} $|c_{\sigma,\pi}|\in\OC(1/d)$, we have 
\begin{equation}
    \left|\frac{1}{d^k}(c_{\sigma,\pi}\Tr[\Lambda F_d(\sigma)]+c_{\sigma^T,\pi}\Tr[\Lambda F_d(\sigma^T)])\Tr[F_d(\pi)O^{\otimes k}]\right|\in\OC\left(\frac{1}{d^{\frac{k+2}{2}}}\right)\,.
\end{equation}
    
\end{proof}

Finally, consider the following proposition.

\begin{supproposition}\label{prop:diag-br}
     Let $O$ be a real-valued traceless Hermitian operator such that $O^2=\id$. Let $\Lambda=\rho_{1}\otimes\cdots\otimes \rho_{k} $ be a tensor product of $k$ ral-valued quantum states such that $\Tr[\rho_{i}\rho_{i'}]\in\Omega\left(\frac{1}{\poly(\log(d))}\right)$ for all $i,i'$. Then, 
     \begin{equation}
         \frac{1}{d^k}\Tr[\Lambda F_d(\sigma)]\Tr[F_d(\sigma)O^{\otimes k} ]\in\widetilde{\Omega}\left(\frac{1}{d^{k/2}}\right)
     \end{equation}
     if $\sigma$ is a product of $k/2$ disjoint length-two cycles, and
     \begin{equation}
         \frac{1}{d^k}\left|\Tr[\Lambda F_d(\sigma)]\Tr[F_d(\sigma^T)O^{\otimes k} ]+\Tr[\Lambda F_d(\sigma^T)]\Tr[F_d(\sigma)O^{\otimes k} ]\right|\in\OC\left(\frac{1}{d^{\frac{k+2}{2}}}\right)
     \end{equation}
     for any other $\sigma$.
\end{supproposition}

\begin{proof}
We start by considering the case when $\sigma$ is a product of $k/2$ disjoint length-two cycles. This implies that $\sg=\sg^T$. From Supplemental Proposition~\ref{prop:powers-O-br}, we know that $\Tr[F_d(\sigma)O^{\otimes k} ]=d^{k/2}$. Then, we find that 
\begin{equation}  \label{eq-ap:trace-perm-state-br}  \Tr[\Lambda F_d(\sigma)]=\prod_{\alpha=1}^{k/2}\Tr[\rho_{{c_\alpha(1)}}\rho_{{c_\alpha(2)}}]\in\Omega\left(\frac{1}{\poly(\log(d))}\right)\,,
\end{equation}
where we have used the fact that $\Tr[\rho_{i}\rho_{i'}]\in\Omega\left(\frac{1}{\poly(\log(d))}\right)$ for all $i,i'$.
Thus, it follows that 
\begin{equation}
     \frac{1}{d^k}\Tr[\Lambda F_d(\sigma)]\Tr[F_d(\sigma^T)O^{\otimes k} ]\in\widetilde{\Omega}\left(\frac{1}{d^{k/2}}\right)\,,
 \end{equation}
where the $\widetilde{\Omega}$ notation omits $(\poly(\log(d)))^{-1}$ factors.

 Next, let us consider the case of $\sigma$ not being a product of $k/2$ disjoint length-two cycles but containing only cycles of even length (if $\sg$ contains a cycle of odd length, it follows from Supplemental Proposition~\ref{prop:powers-O-br} that $\Tr[F_d(\sigma^T)O^{\otimes k} ]=0$). Here, we have from Supplemental Proposition~\ref{prop:powers-O-br} that
 \begin{equation}
 \Tr[F_d(\sigma)O^{\otimes k} ]=\Tr[F_d(\sigma^T)O^{\otimes k} ]=d^{r}    
 \end{equation}
 with $r\leq \frac{k}{2}-1$.
 Then, the following chain of inequalities hold
  \begin{align}
     \frac{1}{d^k}\left|\Tr[\Lambda F_d(\sigma)]\Tr[F_d(\sigma^T)O^{\otimes k} ]+\Tr[\Lambda F_d(\sigma^T)]\Tr[F_d(\sigma)O^{\otimes k} ]\right|& \leq\frac{1}{d^{\frac{k+2}{2}}}\left|\Tr[\Lambda F_d(\sigma)]+\Tr[\Lambda F_d(\sigma^T)]\right|\nonumber\\&\leq\frac{2}{d^{\frac{k+2}{2}}}\in\OC\left(\frac{1}{d^{\frac{k+2}{2}}}\right)\,.
 \end{align}
 Here, we have used Supplemental Proposition~\ref{prop:powers-rho-br} for the last inequality. 
\end{proof}

To finish the proof of Supplemental Theorem~\ref{theo-moments-orto} we simply combine Supplemental Propositions~\ref{prop:off-diag-br} and~\ref{prop:diag-br} and note that in the large-$d$ limit we get

\begin{align}\label{eq:exp-val-orto}
\mathbb{E}_{\mathbb{O}(d)}\left[\prod_{i=1}^k C_{i}\right]&=\frac{1}{d^{k/2}}\sum_{\sigma\in \mathfrak{T}_k}\prod_{\alpha=1}^{k/2}\Tr[\rho_{{c_\alpha(1)}}\rho_{{c_\alpha(2)}}]\,,
\end{align}
where we have defined as $\mathfrak{T}_k\subseteq \brauer$ the set of elements in the Brauer algebra which are exactly given by a product of $k/2$ disjoint length-two cycles. Here we note that every such two-cycle can either be of the form $(\{i,j+k\}\cup\{j,i+k\})$ for $i,j\leq k$ or  $(\{i,j\}\cup\{i+k,j+k\})$ for $i,j\leq k$. Moreover since for either of those two cases the term $\Tr[\rho_{{c_\alpha(1)}}\rho_{{c_\alpha(2)}}]$ will be equal, we will have redundancies in the summation of Eq.~\eqref{eq:exp-val-orto}. We can remove this redundancy by only summing over disjoint transpositions, and adding a coefficient $2^{k/2}$ that accounts for the fact that in each of the $k/2$ length-two cycles, we can choose either $(\{i,j+k\}\cup\{j,i+k\})$ for $i,j\leq k$ or  $(\{i,j\}\cup\{i+k,j+k\})$ for $i,j\leq k$. Then we obtain 
\begin{align}\label{eq:exp-val-orto-ortogonal}
\mathbb{E}_{\mathbb{O}(d)}\left[\prod_{i=1}^k C_{i}\right]&=\frac{2^{k/2}}{d^{k/2}}\sum_{\sigma\in T_k}\prod_{\alpha=1}^{k/2}\Tr[\rho_{{c_\alpha(1)}}\rho_{{c_\alpha(2)}}]\,.
\end{align}

\subsection{States with overlap equal to $1/d$}

We will now consider pure states such that $\Tr[\rho_i \rho_{i'}]=\frac{1}{d}$ for all $\rho_i\neq \rho_{i'}$. We prove the following theorem.

\begin{suptheorem} \label{sup-th:amoments_uncorrelated-br} Let $\rho_i$ for $i=1,\ldots,k$ be a multiset of real-valued quantum states such that $\Tr[\rho_i,\rho_{i'}]=\frac{1}{d}$ for all $\rho_i\neq \rho_{i'}$ and $\Tr[\rho_i^2]\in\Omega\left(\frac{1}{\poly(\log d)}\right)$ for all $\rho_i$,  and let  $O$ be some real-valued traceless Hermitian operator such that $O^2=\id$. Then, in the large-$d$ limit we have
\beq \label{eq:app_moments_uncorrelated_orto} \mathbb{E}_{\mathbb{O}(d)}\left[\prod_{i=1}^k C(\rho_i)\right] = \frac{1}{d^k} \sum_{\sigma\in \mathfrak{T}^c_k(\Lambda)}  d^{k/2} \,,\eeq
where we have defined $\mathfrak{T}^c_k(\Lambda)$ as the set of elements of the Brauer algebra belonging to $\mathfrak{T}_k$ that only connect identical states in $\Lambda=\rho_{1}\otimes\cdots\otimes \rho_{k} $.
\end{suptheorem}

\begin{proof}

Let us first recall Eq.~\eqref{eq:app_twirl_orto},
\begin{align} 
    \TC^{(k)}_{G}[X]=\frac{1}{d^k}\sum_{\sigma\in \mathfrak{B}_k}\Tr[XF_d(\sigma)]F_d(\sigma^{T})+\frac{1}{d^k}\sum_{\sigma,\pi\in \mathfrak{B}_k}c_{\sigma,\pi}\Tr[XF_d(\sigma)]F_d(\pi)\,,
\end{align}
We know from Supplemental Proposition~\ref{prop:powers-O-br} that $\Tr[F_d(\sigma^{T})O^{\otimes k}]$ is maximized by elements $\sigma$ of the Brauer algebra $\mathfrak{B}_k$ that are a product of $k/2$ disjoint cycles of length two. In particular, in that case we have $\Tr[F_d(\sigma^{T})O^{\otimes k}]=d^{k/2}$. It is easy to see that $\Tr[\Lambda F_d(\sigma)]$ is maximized by the elements that connect only identical states, because the states satisfy $\Tr[\rho_i^2]\in\Omega\left(\frac{1}{\poly(\log d)}\right)$ and $\Tr[\rho_i \rho_{i'}]=\frac{1}{d}$ for all $\rho_i\neq \rho_{i'}$. Else, there will be at least two different states in at least one cycle, resulting in a $\OC(\frac{1}{d})$ factor.
 Indeed, using an analogous derivation than that leading to Eq.~\eqref{eq-ap:trace-perm-state-br} for $\sigma\in \mathfrak{T}_k^c(\Lambda)$, we find
\begin{equation} \label{eq:app-lead-contr-uncorrelated-br}
     \frac{1}{d^k}\Tr[\Lambda P_d(\sigma)]\Tr[P_d(\sigma)O^{\otimes k} ]\in\widetilde{\Omega}\left(\frac{1}{d^{k/2}}\right)\,,
\end{equation}
when $\sigma\in \mathfrak{T}_k^c(\Lambda)$, and
\begin{equation} \label{eq:app-neglect-contr-uncorrelated-br}
     \frac{1}{d^k}\Tr[\Lambda P_d(\sigma)]\Tr[P_d(\sigma^{-1})O^{\otimes k} ]\in \OC\left(\frac{1}{d^{\frac{k+2}{2}}}\right)\,,
\end{equation}
otherwise. The latter equation follows from Supplemental Propositions~\ref{prop:powers-O-br} and~\ref{prop:powers-rho-br} when $\sigma\notin \mathfrak{T}_k$, and additionally from the condition $\Tr[\rho_i \rho_{i'}]=\frac{1}{d}$ for all $\rho_i\neq \rho_{i'}$ when $\sigma\in \mathfrak{T}_k$ but $\sigma\notin \mathfrak{T}_k(\Lambda)$. 
Furthermore, since Supplemental Proposition~\ref{prop:off-diag-br} also holds in this case, we recover Eq.~\eqref{eq:app_moments_uncorrelated_orto} in the large-$d$ limit.

Notice that when $\Lambda$ is such that $T_k(\Lambda)=\emptyset$, then $\mathbb{E}_{\mathbb{U}(d)}\left[\prod_{i=1}^k C(\rho_i)\right]\in \OC\left(\frac{1}{d^{\frac{k+2}{2}}}\right)$, which is approximated by $0$ in the large-$d$ limit in Eq.~\eqref{eq:app_moments_uncorrelated_orto}. The justification for this is the same as in the unitary case (see Supp. Info.~\ref{si:unitary}.3).
\end{proof}

Finally, we remark that when the overlaps between the states in $\mathscr{D}$ are such that $\Tr[\rho_i \rho_{i'}]\in o\left(\frac{1}{\poly(\log(1/d))}\right)$ for all $\rho_i\neq \rho_{i'}$ and $\Tr[\rho_{i}^2]\in\Omega\left(\frac{1}{\poly(\log(1/d))}\right)$ for all $\rho_i$, Supplemental Theorem~\ref{sup-th:amoments_uncorrelated-br} still holds, following the same argument as for the unitary group presented in Supp. Info.~\ref{si:unitary}.3.

\subsection{Orthogonal states}

When considering orthogonal states we can derive the following results.

\begin{suptheorem}\label{theo-moments-orth-O}
Let $\rho_i$ for $i=1,\ldots,k$ be a set of  pure and mutually orthogonal real-valued quantum states,  and let  $O$ be some traceless real-valued Hermitian operator such that $O^2=\id$. Then, in the large-$d$ limit we have
\begin{align}\label{eq:scalingexp-orth}
\mathbb{E}_{\mathbb{O}(d)}\left[\prod_{i=1}^k C(\rho_i)\right]&= (-1)^{k/2}\frac{k!}{2^{k/2} (k/2)!}\frac{2^{k/2}}{d^k}\,.
\end{align}
\end{suptheorem}
The proof of this theorem follows similarly to that of Supplemental Theorem~\ref{theo-moments-orth} but noting that here the terms that contribute are not only the $k/2$ disjoint transpositions, but also any element of $\brauer$ where any transposition $(\{i,j+k\}\cup\{j,i+k\})$ for $i,j\leq k$ is replaced by  a term $(\{i,j\}\cup\{i+k,j+k\})$ for $i,j\leq k$.  Since we have $k/2$ such choices, then we get an extra factor of $2^{k/2}$. Said otherwise, we will get all the contributions of the element of $\brauer$ which are composed only of cycles of length two. 

With a similar argument we can find that the following result holds. 
\begin{suptheorem}\label{theo-moments-orth-mixed-orto}
Let $\rho_i$ for $i=1,\ldots,k$ be a multiset of  pure  and mutually orthogonal real-valued states. Then, let $\Lambda$ contain $k_1$ copies  of $\rho_1$, $k_2$ copies of $\rho_2$, and so on. In total, we assume that $\Lambda$ contains $q$ different states and that $\sum_{\beta=1}^q k_{\beta}=k$. Moreover, we assume that  $O$ is  some traceless real-valued Hermitian operator such that $O^2=\id$. Then, in the large-$d$ limit we have
\begin{align}
\mathbb{E}_{\mathbb{O}(d)}\left[\prod_{i=1}^k C(\rho_i)\right]&=(-1)^{\frac{k}{2} -R}\,\frac{2^{k/2} d^R}{d^{k}} \left( \frac{(k - 2R)!}{2^{\frac{k}{2}-R} \left(\frac{k}{2}-R\right)!} \prod_{k_\beta\%2=1 }k_\beta \frac{(2\lfloor k_\beta/2\rfloor)!}{2^{\lfloor k_\beta/2\rfloor} (\lfloor k_\beta/2\rfloor)!} \prod_{k_\beta\%2=0 } \frac{ k_\beta!}{2^{ k_\beta/2} (k_\beta/2)!} \prod_{\beta=1}^q \Tr[\rho_\beta^2]^{\lfloor\frac{k_\beta}{2}\rfloor}\right)\,.
\end{align}
\end{suptheorem}

\section{Proof of Lemma 1}\label{sec:proof-lem1}

In this section we present a proof of Lemma 1, which we recall for convenience.

\begin{lemma}\label{si-lem:exp-cov-SM}
    Let $C_j(\rho_i)$ be the expectation value of a Haar random QNN as in Eq. (1). Then for any $\rho_i\in \mathscr{D}$, $O_j\in \mathscr{O}$,
    \begin{equation}
        \mathbb{E}_{\mathbb{U}(d)}[C_j(\rho_i)]=\mathbb{E}_{\mathbb{O}(d)}[C_j(\rho_i)]=0\,.
    \end{equation}
    Moreover, for any pair of states $\rho_i,\rho_{i'}\in \mathscr{D}$ and operators $O_{j},O_{j'}\in \mathscr{O}$ we have
    \begin{equation}
       {\rm Cov}_{\mathbb{U}(d)}[C_j(\rho_i)C_{j'}(\rho_{i'})]={\rm Cov}_{\mathbb{O}(d)}[C_j(\rho_i)C_{j'}(\rho_{i'})]=0\nonumber\,,
    \end{equation}
    if $j\neq j'$ and 
    \small
    \begin{align}
       \vec{\Sigma}_{i, i'}^{\mathbb{U}}&=\frac{d}{d^2-1}\left(\Tr[\rho_{i}\rho_{i'}]-\frac{1}{d}\right)\,,\\
       \vec{\Sigma}_{i, i'}^{\mathbb{O}}&=\frac{2(d+1)}{(d+2)(d-1)}\left(\Tr[\rho_{i}\rho_{i'}]\left(1-\frac{1}{d+1}\right)-\frac{1}{d+1}\right)\,,
    \end{align}
    \normalsize
    if $j=j'$. Here, we have defined $\vec{\Sigma}_{i, i'}^{G}={\rm Cov}_{G}[C_j(\rho_i)C_{j}(\rho_{i'})]$, where $G=\mathbb{U}(d),\mathbb{O}(d)$.
\end{lemma}

\begin{proof}
We begin by considering expectation values over $\mathbb{U}(d)$. Using Eq.~\eqref{eq:twirl-k1} we can compute
\begin{equation}
    \mathbb{E}_{\mathbb{U}(d)}[C_j(\rho_i)]=\frac{\Tr[\rho_i]\Tr[O_j]}{d}=0\,,
\end{equation}
where we have used the fact that $O_j$ is traceless. 

Next, we can use Eq.~\eqref{eq:twirl-k2} to find 
\begin{align}
    {\rm Cov}_{\mathbb{U}(d)}[C_j(\rho_i)C_{j'}(\rho_{i'})]=&\mathbb{E}_{\mathbb{U}(d)}[C_j(\rho_i)C_{j'}(\rho_{i'})]\nonumber\\
    =&\frac{1}{d^2-1}\left(\Tr[(\rho_i\otimes \rho_{i'})]-\frac{\Tr[(\rho_i\otimes \rho_{i'}) \SWAP]}{d}\right)\Tr[(O_j\otimes O_{j'})(\id\otimes \id)]\nonumber\\
    &+\frac{1}{d^2-1}\left(\Tr[(\rho_i\otimes \rho_{i'}) \SWAP]-\frac{\Tr[(\rho_i\otimes \rho_{i'})]}{d}\right) \Tr[(O_j\otimes O_{j'}) \SWAP]\nonumber\\
    =&\frac{1}{d^2-1}\left(\Tr[\rho_i \rho_{i'}]-\frac{1}{d}\right) \Tr[O_j O_{j'}]\nonumber\\
    =&\frac{d}{d^2-1}\left(\Tr[\rho_i \rho_{i'}]-\frac{1}{d}\right) \delta_{j,j'}\nonumber\,.
\end{align}
Here we have used the fact that, by definition, $\Tr[O_j O_{j'}]=d \delta_{j,j'}$. Hence, we find that if $j\neq j'$, then ${\rm Cov}_{\mathbb{U}(d)}[C_j(\rho_i)C_{j'}(\rho_{i'})]=0$, whereas if $j=j'$ one obtains
\begin{equation}
    {\rm Cov}_{\mathbb{U}(d)}[C_j(\rho_i)C_j(\rho_{i'})]=\frac{d}{d^2-1}\left(\Tr[\rho_i \rho_{i'}]-\frac{1}{d}\right)\,.
\end{equation}

Next, let us take expectation values  over $\mathbb{O}(d)$. Using Eq.~\eqref{eq:twirl-o-k1} we again find
\begin{equation}
    \mathbb{E}_{\mathbb{O}(d)}[C_j(\rho_i)]=\frac{\Tr[\rho_i]\Tr[O_j]}{d}=0\,.
\end{equation}
Next, from Eq.~\eqref{eq:twirl-o-k2} we obtain 
\small
\begin{align}
    {\rm Cov}_{\mathbb{O}(d)}[C_j(\rho_i)C_{j'}(\rho_{i'})]=&\mathbb{E}_{\mathbb{O}(d)}[C_j(\rho_i)C_{j'}(\rho_{i'})]\nonumber\\
    =&\frac{1}{d(d+2)(d-1)}\left((d+1)\Tr[(\rho_i\otimes \rho_{i'})]  -\Tr[(\rho_i\otimes \rho_{i'})\SWAP]  -\Tr[(\rho_i\otimes \rho_{i'}) \Pi]\right)\Tr[(O_j\otimes O_{j'}) (\id\otimes \id)]\nonumber\\
    &+\frac{1}{d(d+2)(d-1)}\left(-\Tr[(\rho_i\otimes \rho_{i'})] + (d+1)\Tr[(\rho_i\otimes \rho_{i'})\SWAP]  -\Tr[(\rho_i\otimes \rho_{i'}) \Pi]\right) \Tr[(O_j\otimes O_{j'}) \SWAP]\nonumber\\
    &+\frac{1}{d(d+2)(d-1)}\left(-\Tr[(\rho_i\otimes \rho_{i'})]  -\Tr[(\rho_i\otimes \rho_{i'})\SWAP] + (d+1) \Tr[(\rho_i\otimes \rho_{i'}) \Pi]\right)\Tr[(O_j\otimes O_{j'}) \Pi]\nonumber\,.
\end{align}
\normalsize
Using the ricochet property of Eq.~\eqref{eq:action_Pi} plus the fact that the operators $\rho_i$, $\rho_{i'}$, $O_j$ and $O_{j'}$ are real-valued leads to
\begin{align}
    {\rm Cov}_{\mathbb{O}(d)}[C_j(\rho_i)C_{j'}(\rho_{i'})]&=\frac{2}{d(d+2)(d-1)}\left(-1 + (d+1)\Tr[\rho_i\rho_{i'}]  -\Tr[\rho_i\rho_{i'}]\right) \delta_{j,j'}\nonumber\\
    &=\frac{2(d+1)}{(d+2)(d-1)}\left(\Tr[\rho_{i}\rho_{i'}]\left(1-\frac{1}{d+1}\right)-\frac{1}{d+1}\right)\delta_{j,j'}\,.
\end{align} 
If $j\neq j'$, then ${\rm Cov}_{\mathbb{O}(d)}[C_j(\rho_i)C_{j'}(\rho_{i'})]=0$, whereas if $j= j'$ one obtains
\begin{equation}
    {\rm Cov}_{\mathbb{O}(d)}[C_j(\rho_i)C_j(\rho_{i'})]=\frac{2(d+1)}{(d+2)(d-1)}\left(\Tr[\rho_{i}\rho_{i'}]\left(1-\frac{1}{d+1}\right)-\frac{1}{d+1}\right)\,.
\end{equation}

\end{proof}

\section{Proof of Lemma 2}\label{sec:proof-lem2}

Let us now prove  Lemma 2, which we recall for convenience.

\begin{lemma} \label{si-lem:moments}
    Let $\mathscr{C}$ be a vector of expectation values of a Haar random QNN as in Eq. (2), where one measures the same operator $O_j$ over a set  of states from $\mathscr{D}$. Furthermore, let $\rho_{i_1},\ldots,\rho_{i_k}\in\mathscr{D}$ be a multiset of states taken from those appearing in $\mathscr{C}$.
    In the large-$d$ limit we find that if $k$ is odd then $\mathbb{E}_{\mathbb{U}(d)}\left[C_{j}(\rho_{i_1})\cdots C_{j}(\rho_{i_k})\right]=\mathbb{E}_{\mathbb{O}(d)}\left[C_{j}(\rho_{i_1})\cdots C_{j}(\rho_{i_k})\right]=0$. Moreover, if $k$ is even and   $\Tr[\rho_i\rho_{i'}]\in\Omega\left(\frac{1}{\poly(\log(d))}\right)$ for all $i,i'$,  we have
    \begin{align}
\mathbb{E}_{\mathbb{U}(d)}\left[C_{j}(\rho_{i_1})\cdots C_{j}(\rho_{i_k})\right]&=\frac{1}{d^{k/2}}\sum_{\sigma\in T_{k}}\prod_{\{t,t'\}\in \sigma} \Tr[\rho_{t}\rho_{t'}]\\
&=\frac{\mathbb{E}_{\mathbb{O}(d)}\left[C_{j}(\rho_{i_1})\cdots C_{j}(\rho_{i_k})\right]}{2^{k/2}}\,,\nonumber
\end{align}
where the summation runs over all the possible disjoint pairing of indexes in the set $\{1,2,\ldots,k\}$, $T_k$, and the product is over the different pairs in each pairing.
\end{lemma}

\begin{proof}

The proof of this lemma follows from results previously derived in this SI. Namely, we simply need to use Eq.~\eqref{eq:exp-val-un}  for the unitary group, and Eq.~\eqref{eq:exp-val-orto-ortogonal}  for the orthogonal group. 

\end{proof}

\section{Proof of Corollary 1}\label{sec:proof-cor1}

For ease of calculation, let us first prove Corollary 1 before our main theorems. The statement of Corollary 1 is as follows. 
\begin{corollary}\label{si-cor:gauss}
    Let $C_j(\rho_i)$ be the expectation value of a Haar random QNN as in Eq.(1).   Then for any $\rho_i\in \mathscr{D}$ and $O_j\in \mathscr{O}$, we have 
    \begin{equation}
        P(C_j(\rho_i))=\NC(0,\sigma^2)\,,
    \end{equation}
 where $\sigma^2=\frac{1}{d},\frac{2}{d}$ when $U$ is Haar random over $\mathbb{U}(d)$ and $\mathbb{O}(d)$, respectively.
\end{corollary}

\begin{proof}

In order to prove Corollary~\ref{si-cor:gauss} we need to consider two distinct cases: when we sample over the unitary and orthogonal groups. 

We start by taking averages over $\mathbb{U}(d)$. Using Lemma~\ref{si-lem:moments} we have that the $k$-th moment (for $k$ even) of the distribution is 
\begin{align}
\mathbb{E}_{\mathbb{U}(d)}\left[C_{j}(\rho)^k\right]&=\frac{1}{d^{k/2}}\sum_{\sigma\in T_{k}}\prod_{\{t,t'\}\in \sigma} \Tr[\rho\rho]=\frac{1}{d^{k/2}}\sum_{\sigma\in T_{k}}\prod_{\{t,t'\}\in \sigma} 1 = \frac{k!}{d^{k/2}2^{k/2} (k/2)!}.
\end{align}
Clearly, we can verify that these moments match those of a Gaussian distribution, as
\begin{equation}
    \frac{\mathbb{E}_{\mathbb{U}(d)}\left[C_{j}(\rho)^k\right]}{\mathbb{E}_{\mathbb{U}(d)}\left[C_{j}(\rho)^2\right]^{k/2}}=\frac{k!}{2^{k/2} (k/2)!}\,.
\end{equation}
To prove that these moments uniquely determine the distribution of $\mathscr{C}$, we use Carleman's condition.
\begin{suplemma}[Carleman's condition, Hamburger case~\cite{kleiber2013multivariate}]\label{lem:carleman}
Let $\gamma_k$ be the (finite) moments of the distribution of a random variable $X$ that can take values on the real line $\mathbb{R}$. These moments determine uniquely the distribution of $X$ if
\beq \sum_{k=1}^\infty \gamma_{2k}^{-1/2k} = \infty \ .\eeq
\end{suplemma}

In our case, we have
\begin{equation}
    \sum_{k=1}^\infty \left(\frac{1}{d^k} \frac{(2k)!}{2^k k!}\right)^{-1/2k} = \sqrt{2d} \sum_{k=1}^\infty \left((2k)\cdots(k+1)\right)^{-1/2k} \geq  \sqrt{2d} \sum_{k=1}^\infty \left((2k)^k\right)^{-1/2k} = \sqrt{2d}\sum_{k=1}^\infty \frac{1}{\sqrt{2k}} = \infty \,, 
\end{equation}
and so, according to Supplemental Lemma~\ref{lem:carleman}, Carleman's condition is satisfied. Therefore, $P(C_j(\rho_i))$ is distributed according to  a Gaussian distribution with zero-mean and variance $\sigma^2=\frac{1}{d}$.

Next, we consider taking averages over $\mathbb{O}(d)$. Using Lemma~\ref{si-lem:moments} we have that the $k$-th moment (for $k$ even) of the distribution is 
\begin{align}
\mathbb{E}_{\mathbb{O}(d)}\left[C_{j}(\rho)^k\right]= \frac{ k!}{d^{k/2} (k/2)!}.
\end{align}
We can again check that these moments match those of a Gaussian distribution, as one has
\begin{equation}
    \frac{\mathbb{E}_{\mathbb{O}(d)}\left[C_{j}(\rho)^k\right]}{\mathbb{E}_{\mathbb{O}(d)}\left[C_{j}(\rho)^2\right]^{k/2}}=\frac{k!}{2^{k/2} (k/2)!}\,,
\end{equation}
and also that Carleman's condition is satisfied,
\begin{equation}
    \sum_{k=1}^\infty \left(\frac{(2k)!}{ d^k k!}\right)^{-1/2k} = \sqrt{d} \sum_{k=1}^\infty \left((2k)\cdots(k+1)\right)^{-1/2k} \geq   \sum_{k=1}^\infty \left((2k)^k\right)^{-1/2k} = \sum_{k=1}^\infty \frac{1}{\sqrt{2k}} = \infty \,.
\end{equation}
Therefore, $P(C_j(\rho_i))$ is distributed according to  a Gaussian distribution with zero mean and variance $\sigma^2=\frac{2}{d}$.

\end{proof}

\section{Proof of Theorem 1}\label{sec:proof-theo1}

Let us now present a proof for Theorem 1, which we recall here.
\begin{theorem} \label{si-th:1}
Under the same conditions for which Lemma~\ref{si-lem:moments} holds, the vector $\mathscr{C}$ forms a GP with mean vector $\vec{\mu}=\vec{0}$ and covariance matrix given by $
    \vec{\Sigma}_{i, i'}^{\mathbb{U}}=\frac{\vec{\Sigma}_{i, i'}^{\mathbb{O}}}{2}= \frac{\Tr[\rho_{i}\rho_{i'}]}{d}$.
\end{theorem}

\begin{proof}

First, let us recall that a multivariate Gaussian distribution $\NC(\vec{\mu},\vec{\Sigma})$,  is fully defined by its $m$-dimensional mean vector $\vec{\mu}=(\mathbb{E}[X_{1}],\ldots,\mathbb{E}[X_{m}])$, and its $m\times m$ dimensional covariance matrix with entries $(\vec{\Sigma})_{\alpha\beta}={\rm Cov}[X_{\alpha},X_{\beta}]$. From here, any higher order  moments can be obtained from Isserlis  theorem~\cite{isserlis1918formula}. Specifically, if we want to compute a $k$-th  order moment, then we have $\mathbb{E}[X_{1}X_{2}\cdots X_{k}]=0$ if $k$ is odd, and 
\begin{equation}\label{eq:wick}
    \mathbb{E}[X_{1}X_{2}\cdots X_{k}]=
        \sum_{\sigma\in T_{k}}\prod_{\{t,t'\}\in \sigma} {\rm Cov} [X_{t},X_{t'}]\,,
\end{equation}
if $k$ is even. Here, the summation runs over all the possible pairing of indexes, $T_k$, in the set $\{1,2,\ldots,k\}$.

A direct comparison between the results in Lemma~\ref{si-lem:moments} and Eq.~\eqref{eq:wick} shows that indeed the moments of $\mathscr{C}$ match those of a GP with covariance matrix $
    \vec{\Sigma}_{i, i'}^{\mathbb{U}}=\frac{\vec{\Sigma}_{i, i'}^{\mathbb{O}}}{2}= \frac{\Tr[\rho_{i}\rho_{i'}]}{d}$. To prove that these moments uniquely determine the distribution of $\mathscr{C}$, we use the fact that since its marginal distributions are determinate via Carleman's condition (see the previous section), then so is the distribution of $\mathscr{C}$~\cite{kleiber2013multivariate}. This result holds for both the unitary and orthogonal groups.

\end{proof}

\section{Proof of Theorem 2}\label{sec:proof-theo2}

Here we present a proof for Theorem 2, which we restate for convenience. 
\begin{theorem} \label{si-th:2}
    Let $\mathscr{C}$ be a vector of expectation values of an operator in $\mathscr{O}$ over a set of states in $\mathscr{D}$. If $\Tr[\rho_i\rho_{i'}]=\frac{1}{d}$ for all $i\neq i'$, then in the large-$d$ limit $\mathscr{C}$ forms a GP with mean vector $\vec{\mu}=\vec{0}$ and diagonal covariance matrix
    \begin{equation}
\vec{\Sigma}_{i, i'}^{\mathbb{U}}= \frac{\vec{\Sigma}_{i, i'}^{\mathbb{O}}}{2} = \begin{cases}
    \frac{1}{d}\quad\text{if $i=i'$} \\
    0\quad \text{if $i\neq i'$}\
\end{cases}\,.
\end{equation}
\end{theorem}

\begin{proof}
   We use Isserlis  theorem~\cite{isserlis1918formula} (see Eq.~\eqref{eq:wick}), that states that if $\mathscr{C}$ is a GP, then
\begin{equation} \label{eq:suppI_wick}
    \mathbb{E}_G[C_j(\rho_1)\cdots C_j(\rho_k)]=
        \sum_{\sigma\in T_{k}}\prod_{\{i,i'\}\in \sigma} {\rm Cov}_G [C_j(\rho_i)C_j(\rho_{i'})]\,.
\end{equation}
for $G=\mathbb{U}(d),\mathbb{O}(d)$. Using Lemma~\ref{si-lem:exp-cov-SM}, we know that the covariance matrix is diagonal with entries that in the large-$d$ limit are given by $\vec{\Sigma}_{i, i}^{\mathbb{U}}=\frac{\vec{\Sigma}_{i, i}^{\mathbb{O}}}{2}= \frac{1}{d}$. Hence, the terms that contribute in~\eqref{eq:suppI_wick} are those permutations in $T_k$ that only connect identical states. Comparing~\eqref{eq:suppI_wick} with Eqs.~\eqref{eq:app_moments_uncorrelated} and~\eqref{eq:app_moments_uncorrelated_orto}, it is clear that the entries of $\mathscr{C}$ follow a joint multivariate Gaussian distribution.
\end{proof}

\section{Proof of Theorem 3}\label{sec:proof-theo3}

We now provide a proof for Theorem 3, whose statement is as follows. 
\begin{theorem} \label{si-th:3}
 Let $\mathscr{C}$ be a vector of expectation values of an operator in $\mathscr{O}$ over a set of states in $\mathscr{D}$. If $\Tr[\rho_i\rho_{i'}]=0$ for all $i\neq i'$, then in the large-$d$ limit $\mathscr{C}$ forms a GP with mean vector $\vec{\mu}=\vec{0}$ and covariance matrix
\begin{equation}
\vec{\Sigma}_{i, i'}^{\mathbb{U}(d)} = \frac{\vec{\Sigma}_{i, i'}^{\mathbb{O}(d)}}{2}=\begin{cases}
    \frac{\Tr[\rho_i^2]}{d}\quad\text{if $i=i'$} \\
    -\frac{1}{d^2}\quad \text{if $i\neq i'$}\,
\end{cases}.  
\end{equation}
\end{theorem}

\begin{proof}
In what follows, we will again us the strategy of proving that the moments of $\mathscr{C}$ match those of a GP. From Isserlis  theorem~\cite{isserlis1918formula} (see Eq.~\eqref{eq:wick}), we know that if $\mathscr{C}$ is indeed a GP, then
\begin{equation}
    \mathbb{E}_G[C_j(\rho_1)\cdots C_j(\rho_k)]=
        \sum_{\sigma\in T_{k}}\prod_{\{i,i'\}\in \sigma} {\rm Cov}_G [C_j(\rho_i)C_j(\rho_{i'})]\,.
\end{equation}
for $G=\mathbb{U}(d),\mathbb{O}(d)$.

There are three situations we must consider. When the states $\rho_1,\ldots,\rho_k$ are all  the same, when they are all different,  and when there are $k_1$ of the same state $\rho_1$, $k_2$ of the same state $\rho_2$, and so on. In the last case, we assume that there are $q$ different states so that $\sum_{\beta=1}^q k_{\beta}=k$. 

Let us start with taking averages over $\mathbb{U}(d)$. If all the states are the same we simply use Supplemental Corollary~\ref{cor:moments-gaussian-single} to find 
\begin{align}
\mathbb{E}_{\mathbb{U}(d)}\left[\prod_{\gamma=1}^k C_{\gamma}\right]&=\frac{1}{d^{k/2}}\sum_{\sigma\in T_k}\prod_{\alpha=1}^{k/2}=\frac{k!}{d^{k/2}2^{k/2} (k/2)!}\,,
\end{align}
and we see that the moments match. Then, if all the states are orthogonal, we use Supplemental Theorem~\ref{theo-moments-orth} to obtain
\begin{align}
\mathbb{E}_{\mathbb{U}(d)}\left[\prod_{i=1}^k C(\rho_i)\right]&=\frac{k!}{2^{k/2} (k/2)!}\frac{(-1)^{k/2}}{d^k}\,,
\end{align}
which again shows that the moments match. Finally, if there are $k_\beta$ copies of $\rho_\beta$, we know from Supplemental Theorem~\ref{theo-moments-orth-mixed} that in the large-$d$ limit  the moments go as
\begin{align}
\mathbb{E}_{\mathbb{U}(d)}\left[\prod_{i=1}^k C(\rho_i)\right]&=(-1)^{\frac{k}{2} -R}\,\frac{d^R}{d^{k}} \left( \frac{(k - 2R)!}{2^{\frac{k}{2}-R} \left(\frac{k}{2}-R\right)!} \prod_{k_\beta\%2=1 }k_\beta \frac{(2\lfloor k_\beta/2\rfloor)!}{2^{\lfloor k_\beta/2\rfloor} (\lfloor k_\beta/2\rfloor)!} \prod_{k_\beta\%2=0 } \frac{ k_\beta!}{2^{ k_\beta/2} (k_\beta/2)!} \prod_{\beta=1}^q \Tr[\rho_\beta^2]^{\lfloor\frac{k_\beta}{2}\rfloor}\right)\,.
\end{align}
This equation can be understood as follows. First, we need to count how many copies of the same state $\rho_\beta$ can we pair at the same time. This is simply given by $R=\sum_{k_\beta\%2=0 }\lfloor\frac{k_\beta}{2}\rfloor$. Next, we see that the first and second term in the multiplications respectively counts how many different ways we have to pair copies of the same state when $k_\beta\geq 2$ is odd or  even. Additionally, the term $\frac{(k - 2R)!}{2^{\frac{k}{2}-R} \left(\frac{k}{2}-R\right)!}$ counts how many different ways we have to pair copies of the remaining states.

Inspecting Isserlis  theorem we see what each term in the sum
$\sum_{\sigma\in T_{k}}\prod_{\{i,i'\}\in \sigma} {\rm Cov}_G [C_j(\rho_i)C_j(\rho_{i'})]$ will contain factors that are equal to $\frac{\Tr[\rho_i^2]}{d}$ and which come  from ${\rm Cov}_G [C_j(\rho_i)C_j(\rho_{i})]$; and factors that are equal to $-\frac{1}{d^2}$ coming from ${\rm Cov}_G [C_j(\rho_i)C_j(\rho_{i'})]$ for $i\neq i'$. One can readily see that one will have a number of terms of the form ${\rm Cov}_G [C_j(\rho_i)C_j(\rho_{i})]$ equal to  $\sum_{k_\beta\%2=0 }\lfloor\frac{k_\beta}{2}\rfloor$. As such, Isserlis  theorem indicates that the leading order terms of $\mathbb{E}_G[C_j(\rho_1)\cdots C_j(\rho_k)]$ will scale as $\frac{d^R}{d^{k}}$. Then, counting how many such terms exist leads to $\frac{(k - 2R)!}{2^{\frac{k}{2}-R} \left(\frac{k}{2}-R\right)!} \prod_{k_\beta\%2=1 }k_\beta \frac{(2\lfloor k_\beta/2\rfloor)!}{2^{\lfloor k_\beta/2\rfloor} (\lfloor k_\beta/2\rfloor)!}\prod_{k_\beta\%2=0 } \frac{ k_\beta!}{2^{ k_\beta/2} (k_\beta/2)!}$.

Taken together, the previous results show that the moments of the deep QNN outcomes indeed match those of a GP. Here, we can again prove that these moments uniquely determine the distribution of $\mathscr{C}$ from the  fact that since its marginal distributions are determinate via Carleman's condition (see Corollary~\ref{si-cor:gauss}), then so is the distribution of $\mathscr{C}$~\cite{petz2004asymptotics,kleiber2013multivariate}. 

A similar argument can be made for the orthogonal group using Supplemental Theorem~\ref{theo-moments-orth-mixed-orto}. 
\end{proof}

\section{Proof of Theorem 4}\label{sec:proof-theo4}

In this section we will provide a proof for Theorem 4. 

\begin{theorem}
    Consider a GP obtained from a Haar random QNN. Given the set of observations $(y(\rho_1),\ldots,y(\rho_m))$ obtained from $N\in\OC(\poly(\log(d)))$ measurements, then  the predictive distribution of the GP is trivial:
    \small
    \begin{equation}
P(C_j(\rho_{m+1})|C_j(\rho_{1}),\ldots,C_j(\rho_{m}))=P(C_j(\rho_{m+1}))=\NC(0,\sigma^2)\,,\nonumber
    \end{equation}
    \normalsize
where $\sigma^2$ is given by Corollary~\ref{si-cor:gauss}.
\end{theorem}

\begin{proof}
    
Let us consider  a setting where we are interested in computing the expectation value of an operator $O$ at the output of a Haar random QNN, $U$. That is, we define the quantity $C(\rho)=\Tr[U\rho UO]$.  Then, given $m$ quantum states $\rho_1,\cdots,\rho_m$, we want to predict the  expectation value $C(\rho_{m+1})$, given the quantities $C(\rho_1),\ldots, C(\rho_m)$. As we have seen in the Methods section, in order to make predictions with the Gaussian process, we need the covariances
\begin{align}
    {\rm Cov}_{\mathbb{U}(d)}[C(\rho_i),C(\rho_{i'})]&=\frac{d}{d^2-1}\left(\Tr[\rho_{i}\rho_{i '}]-\frac{1}{d}\right)\,,\\
    {\rm Cov}_{\mathbb{O}(d)}[C(\rho_i),C(\rho_{i'})]&=\frac{2(d+1)}{(d+2)(d-1)}\left(\Tr[\rho_{i}\rho_{i'}]\left(1-\frac{1}{d+1}\right)-\frac{1}{d+1}\right)\,.
\end{align}
And we  recall that in the large-$d$ limit, $P(C(\rho_i))=\NC(0,\sigma^2)$ with $\sigma^2=\frac{1}{d},\frac{2}{d}$ for the unitary and orthogonal groups, respectively (see Corollary~\ref{si-cor:gauss}). Moreover, we will assume that the expectation values are computed with $N$ shots. This leads to a statistical noise in the observation procedure that is modeled by a zero-mean Gaussian with a variance given by $\sigma_N^2=\frac{1}{N}$, so that $\frac{1}{\sigma_N^2}\in\OC(N)$.

From Eq. (20), we know that we need to invert the matrix $\vec{\Sigma}+\sigma_N^2\id$.  Using Supplemental Lemma~\ref{lem:inv}, we can then write 
\begin{equation}
    (\vec{\Sigma}+\sigma_N^2\id)^{-1}=\frac{1}{\sigma_N^2}\id-\frac{1}{\sigma_N^4}(\id+\frac{1}{\sigma_N^2}\vec{\Sigma})\vec{\Sigma}\,.
\end{equation}
Noting that the absolute value of the matrix elements of $\vec{\Sigma}$ are at most in $\OC(\frac{1}{d})$, then in the large-$d$ limit we will have 
\begin{equation}
    (\vec{\Sigma}+\sigma_N^2\id)^{-1}\sim\frac{1}{\sigma_N^2}\id\,.
\end{equation}
Hence, the mean and variance  of $P(C(\rho_{m+1})|C(\rho_{1}),\ldots,C(\rho_{m}))$ will be 
\begin{align}
    \mu(C(\rho_{m+1}))&=\frac{1}{\sigma_N^2} \vec{m}^T \cdot\vec{C}\,,\\
    \sigma^2(C(\rho_{m+1}))&=\sigma^2- \frac{1}{\sigma_N^2} \vec{m}^T\cdot\vec{m}\,.
\end{align}
Next, we note that the correction $\frac{1}{\sigma_N^2} \vec{m}^T \cdot\vec{C}\in\OC(\frac{N}{d})$ and $\frac{1}{\sigma_N^2} \vec{m}^T\cdot\vec{m}\in\OC(\frac{N}{d^2})$, where we have use the fact that $|C(\rho_\gamma)|\leq 1$ for all $\rho_\gamma$ and that the absolute values of the entries of $\vec{m}$ are  in $\OC(\frac{1}{d})$. If 
$N\in\OC(\poly(\log(d)))$,  we have that in the large-$d$ limit 
\begin{align}
    \mu(C(\rho_{m+1}))&=0\,,\\
    \sigma^2(C(\rho_{m+1}))&=\sigma^2\,.
\end{align}
Thus, using a poly-logarithmic-in-$d$ number of shots we cannot use the Gaussian process to learn anything about the probability of $C(\rho_{m+1})$.

\end{proof}

\section{Proof of Corollary 2}\label{sec:proof-coro3}

In this section we present a proof for Corollary 2, which we restate for convenience. 

\begin{corollary}
    Let $C_j(\rho_i)$ be the expectation value of a Haar random QNN as in Eq. (1). Assuming that there exists a parametrized gate in $U$ of the form $e^{-i \theta H}$ for some Pauli operator $H$, then 
    \begin{equation}
       P(|C_j(\rho_i)|\geq c), \,P(|\partial_{\theta}C_j(\rho_i)|\geq c)\in\OC\left(\frac{1}{ce^{dc^2}\sqrt{d}}\right)\,.\nonumber
    \end{equation}
\end{corollary}

\begin{proof}
Let us start by evaluating the probability $ P(|C_j(\rho_i)|\geq c)$. Since we know that $C_j(\rho_i)$ follows a  Gaussian distribution $\NC(0,\sigma^2)$ with $\sigma^2=\frac{1}{d}$ (see Corollary~\ref{si-cor:gauss}), we can use the equality
\begin{equation}\label{eq:tail-probab}
   P(|C_j(\rho_i)|\geq c)=\frac{2\sqrt{d}}{\sqrt{2\pi}}\int_{c}^{\infty} dx e^{-x^2 d}=\frac{{\rm Erfc}\left[c\sqrt{d}\right]}{\sqrt{2}}\,,
\end{equation}
where ${\rm Erfc}$ denotes the complementary error function. Using the fact that for large $x$, ${\rm Erfc}[x]\leq \frac{e^{-x^2}}{x\sqrt{\pi}}$, we then find 
\begin{equation}\label{eq:order-prob}
     P(|C_j(\rho_i)|\geq c)\in\OC\left(\frac{1}{ce^{dc^2}\sqrt{d}}\right)\,.
\end{equation}

Next, let us consider the probability $P(|\partial_{\theta}C_j(\rho_i)|\geq c)$, and let us write again the explicit dependence of $U$ on some set of parameters $\thv$. That is, we write $U(\thv)$. Then, let us note that if the parameter $\theta\in\thv$ appears in $U$ as $e^{-i \theta H}$ for some Pauli operator $H$, we can use the parameter shift-rule~\cite{cerezo2020variationalreview} to compute 
\begin{equation}    \partial_{\theta}C_j(\rho_i)=C_j^+(\rho_i)-C_j^-(\rho_i)\,,
\end{equation}
where 
\begin{equation}
C_j^\pm(\rho_i)=\Tr[U(\thv^\pm)\rho_iU\ad(\thv^\pm)O_j]\,,
\end{equation}
and $\thv^\pm=\thv+\frac{\pi}{4}\hat{e}_{\theta}$. Here $\hat{e}_{\theta}$ denotes a unit vector with an entry equal to one in the same position as that of $\theta$.   Then, let us define the events $\EC_\pm$ as $|C_j^\pm(\rho_i)|>2c$. Clearly, $|\partial_{\theta}C_j(\rho_i)|>c$ is a subset of $\EC_+\cup\EC_-$. Hence, 
\begin{align}
P(|\partial_{\theta}C_j(\rho_i)|\geq c)&\leq P(\EC_+\cup\EC_-)\nonumber\\
&\leq P(\EC_+)+P(\EC_-)\nonumber\\
&= {\rm Erfc}\left[\frac{2c\sqrt{d}}{\sqrt{2}}\right]+{\rm Erfc}\left[\frac{2c\sqrt{d}}{\sqrt{2}}\right]\in \OC\left(\frac{1}{ce^{4dc^2}\sqrt{d}}\right)\,.\label{eq:union-probs}
\end{align}
In the second inequality we have used the union bound, and in the first equality we used Eqs.~\eqref{eq:tail-probab} and~\eqref{eq:order-prob}.
    
\end{proof}

\section{Proof of Corollary 3}\label{sec:proo-coro4}

Here we provide a proof for Proof of Corollary 3. 

\begin{corollary}
    Let $U$ be drawn from a $t$-design. Then, under the same conditions for which Theorems~\ref{si-th:1},~\ref{si-th:2} and~\ref{si-th:3} hold, the vector $\mathscr{C}$ matches the first $t$ moments of a GP.
\end{corollary}

\begin{proof}
    The proof follows directly by using Theorems~\ref{si-th:1},~\ref{si-th:2} and~\ref{si-th:3}, and the fact that by definition a $t$-design matches the first $t$ moments of a Haar random unitary. Hence, the first $t$-moments of $\mathscr{C}$ match those of a GP under the same conditions under which Theorems~\ref{si-th:1},~\ref{si-th:2} and~\ref{si-th:3} hold.
\end{proof}

Here we also provide a proof for the following equation presented in the main text:
\beq P(|C_j(\rho_i)|\geq c), \,P(|\partial_{\theta}C_j(\rho_i)|\geq c)\in\OC\left(\frac{\left(2\left\lfloor\frac{t}{2}\right\rfloor\right)!}{2^{\left\lfloor\frac{t}{2}\right\rfloor }(dc^2)^{ \left\lfloor\frac{t}{2}\right\rfloor}\left(\left\lfloor\frac{t}{2}\right\rfloor\right)!}\right)\,.\eeq 

We start by considering $P(|C_j(\rho_i)|\geq c)$ Here we can use the generalization of Chebyshev's inequality to higher-order moments:
\begin{equation}\label{eq:cheb-higher}
    \Pr \left(|X- \mathbb{E}[X]|\geq c\right)\leq \frac{  \mathbb{E}[|X-\mathbb{E} [X]|^k]}{c^{k}}\,,
\end{equation}
for $c>0$ and for $k\geq 2$.  If $U$ forms a $t$-design, then we can evaluate $\mathbb{E}[C_j(\rho_i)^{t}]$ by noting that it matches the $t$ first moment of a $\NC(0,\sigma^2)$ distribution with $\sigma^2=\frac{1}{d}$ (see Corollary~\ref{si-cor:gauss}). In particular, we know that since the odd moments are zero, we only need to care about the largest even moment that the $t$-design matches. Hence we can use  
\begin{equation}
    \mathbb{E}[|C_j(\rho_i)^{2\lfloor \frac{t}{2}\rfloor }]=\frac{\left(2\left\lfloor\frac{t}{2}\right\rfloor\right)!}{2^{\left\lfloor\frac{t}{2}\right\rfloor }\left(\left\lfloor\frac{t}{2}\right\rfloor\right)!}\mathbb{E}[C_j(\rho_i)^{2}]^{\left\lfloor\frac{t}{2}\right\rfloor }=\frac{\left(2\left\lfloor\frac{t}{2}\right\rfloor\right)!}{2^{\left\lfloor\frac{t}{2}\right\rfloor }d^{ \left\lfloor\frac{t}{2}\right\rfloor}\left(\left\lfloor\frac{t}{2}\right\rfloor\right)!}.
\end{equation}
Combining this with Eq.~\eqref{eq:cheb-higher} leads to the desired result.

One can obtain a similar bound for $\mathbb{E}[|\partial_\theta C_j(\rho_i)^{2\lfloor \frac{t}{2}\rfloor }]$ following the same steps as the ones used to derive Eq.~\eqref{eq:union-probs}.

\section{Proof of Theorem 5}\label{sec:theo5}

Let us here provide a proof for Theorem 5.

\begin{theorem} \label{si-th:5}
    The results of Theorems~\ref{si-th:1} and~\ref{si-th:2} will hold, on average, if $\mathbb{E}_{\rho_i,\rho_{i'}\sim \mathscr{D}}\Tr[\rho_i\rho_{i'}]\in\Omega\left(\frac{1}{\poly(\log(d))}\right)$ and  $\mathbb{E}_{\rho_i,\rho_{i'}\sim \mathscr{D}}\Tr[\rho_i\rho_{i'}]=\frac{1}{d}$, respectively.
\end{theorem}

\begin{proof}
    The proof for Theorem~\ref{si-th:5} simply follows that of Theorems~\ref{si-th:1} and~\ref{si-th:2}: anytime we used in these proofs the fact that $\Tr[\rho_i\rho_{i'}]\in\Omega\left(\frac{1}{\poly(\log(d))}\right)$ or  $\Tr[\rho_i\rho_{i'}]=\frac{1}{d}$, we replace that by $\mathbb{E}_{\rho_i,\rho_{i'}\sim \mathscr{D}}\Tr[\rho_i\rho_{i'}]\in\Omega\left(\frac{1}{\poly(\log(d))}\right)$ and  $\mathbb{E}_{\rho_i,\rho_{i'}\sim \mathscr{D}}\Tr[\rho_i\rho_{i'}]=\frac{1}{d}$, respectively. The rest of the proofs are the same.
\end{proof}

\end{document}